\documentclass{SciPost}
\usepackage[T1]{fontenc}
\usepackage{lmodern}
\usepackage{microtype}
\usepackage{amsmath,amssymb,amsthm,mathtools,bm}
\usepackage{mathrsfs}
\usepackage{booktabs,tabularx,array,longtable,multirow}
\usepackage{enumitem}
\usepackage{graphicx}
\usepackage{xcolor}
\usepackage[nameinlink,capitalise,noabbrev]{cleveref}
\usepackage{caption}
\usepackage{subcaption}
\usepackage{slashed}
\usepackage{bbm}
\usepackage{url}
\usepackage{cite}
\usepackage{listings}
\usepackage{etoolbox}
\usepackage{placeins}
\usepackage{float}

\definecolor{spblue}{RGB}{30,76,130}
\definecolor{exactgreen}{RGB}{20,100,65}
\definecolor{numorange}{RGB}{165,85,5}
\definecolor{conjpurple}{RGB}{100,50,125}
\definecolor{openred}{RGB}{145,35,35}
\hypersetup{colorlinks=true,linkcolor=spblue,citecolor=spblue,urlcolor=spblue,
 pdftitle={Analytical results for the Shannon entropy of the critical transverse-field Ising chain},
 pdfauthor={M. A. Rajabpour}}
\newcommand{\dd}{\,\mathrm d}
\newcommand{\EE}{\mathbb E}
\newcommand{\PP}{\mathbb P}

\newcommand{\Ztwo}{\mathbb Z_2}

\newcommand{\cM}{\mathfrak M}
\newcommand{\cA}{\mathcal A}
\newcommand{\cB}{\mathcal B}
\newcommand{\FP}{\operatorname{FP}}
\newcommand{\Pf}{\operatorname{Pf}}
\newcommand{\Hf}{\operatorname{Hf}}
\newcommand{\diag}{\operatorname{diag}}

\newcommand{\arcosh}{\operatorname{arcosh}}
\newcommand{\cyc}{\mathrm{cyc}}
\newcommand{\edge}{\mathrm{edge}}

\newcommand{\linesys}{\mathrm{line}}

\newtheorem{theorem}{Theorem}[section]
\newtheorem{proposition}[theorem]{Proposition}
\newtheorem{lemma}[theorem]{Lemma}

\theoremstyle{definition}

\newtheorem{conjecture}[theorem]{Conjecture}
\theoremstyle{remark}

\begin{document}
\begin{center}
{\Large\bfseries Analytical results for the Shannon entropy of the critical transverse-field Ising chain}\par\vspace{1.0em}
{\large M.~A.~Rajabpour$^{1,2}$}\par\vspace{0.6em}
$^1$ Instituto de F\'isica, Universidade Federal Fluminense, Av.~Gal.~Milton Tavares de Souza s/n, Gragoat\'a, 24210-346, Niter\'oi, RJ, Brazil\\
$^2$ Theoretical Physics III, Center for Electronic Correlations and Magnetism, Institute of Physics, University of Augsburg, D-86135 Augsburg, Germany\par\vspace{0.45em}
{\small Corresponding author: \href{mailto:mohammad1.rajabpour@uni-a.de}{mohammad1.rajabpour@uni-a.de}}\par\vspace{0.45em}
{\small \today}
\end{center}

\section*{Abstract}
\textbf{Local-basis Shannon entropies of critical wave functions contain universal subleading information, but the Shannon point of the critical transverse-field Ising chain is singular in the conventional R\'enyi approach.  We treat it directly at $n=1$ by representing the complete computational-basis Born distribution as the odd-degree boundary of independent long-range Bernoulli edges.  The Shannon chain rule separates the entropy into explicit independent-edge and conditional cycle-space terms.  For the periodic chain this isolates the analytic edge anomaly and organizes the remaining constant by linked vertex support.  Conditional on the stated finite-part matching and forest assumptions, every nonvanishing linked coefficient is represented by an explicit finite-dimensional integral, giving an all-orders analytic hierarchy; the four-vertex term is evaluated in closed form, and the first few linked sectors already nearly saturate the established periodic constant, which is also reconstructed independently from the exact finite-size distribution.  For an interval, the exterior reduces exactly to one ghost vertex and the explicit squared logarithms cancel, leaving the analytic edge contribution $\gamma_E=\log2/4-1/16$.  The conditional cycle term has an exact all-support decomposition into physical-line and ghost-linked sectors; conditional on closure of the boundary forest subtraction, these sectors define an all-orders hierarchy of renormalized line and monomer--dimer boundary periods for the remaining logarithmic coefficient.  The first complete augmented linked coefficient is evaluated in closed form, with an exact cancellation between its line and ghost boundary-layer anomalies; an independent reconstruction from exact finite-size probabilities recovers the established coefficient $\gamma_1=0.060020(3)$.}

\noindent\rule{\textwidth}{0.4pt}
\tableofcontents
\noindent\rule{\textwidth}{0.4pt}

\section{Introduction}
\label{sec:intro}

Local-basis Shannon--R\'enyi entropies characterize how a many-body wave function is distributed over the outcomes of a chosen local measurement.  They are classical entropies of Born probabilities rather than reduced-density-matrix entanglement entropies, and consequently depend on the measurement basis.  Their leading contribution is generally extensive, but critical wave functions can carry universal subleading constants or logarithms.  This structure was established in one-dimensional critical states and Luttinger liquids~\cite{Stephan2009,Zaletel2011,Stephan2011}, while the related multifractal structure of spin-chain ground-state amplitudes was analyzed in Ref.~\cite{AtasBogomolny2012}.  Subsequent work has explored interacting systems, disordered chains, symmetry-broken phases and line subsystems~\cite{Luitz2014PRL,Luitz2014PRB,Luitz2014JSTAT,Monthus2015,GetelinaAlcarazHoyos2016,Misguich2017,LuitzLaflorencie2017,ZhaoLaflorencieMeng2025}.

The critical transverse-field Ising (TFI) chain is a particularly sharp test because two different universal quantities arise from the same ground-state probability distribution.  For a complete periodic chain and for a finite interval $A_\ell$ in the infinite chain, respectively, we seek the subleading terms in
\begin{align}
 H_1^{\rm ring,(z)}(L)&=s_1L+r_1^{(z)}+o(1),\label{eq:intro-ring-scaling}\\
 H_1(A_\ell)&=s_1\ell+\gamma_1\log\ell+O(1).\label{eq:intro-interval-scaling}
\end{align}
The periodic constant was found numerically in the Ising transfer-matrix setting, where the universal term is discontinuous at the Shannon point~\cite{Stephan2010}; a later high-precision analysis gives the $x$-basis benchmark used below~\cite{LauGrassberger2013}.  The corresponding $z$-basis constant is $r_1^{(z)}=r_1^{(x)}-\log2$, by the exact duality shift.  For subsystems, early work proposed a central-charge relation for the Shannon mutual information~\cite{AlcarazRajabpour2013}, while high-precision calculations for the Ising chain found instead the universal coefficient $b_1\simeq0.4801629(2)$~\cite{Stephan2014}.  Further studies clarified the role of conformal bases and R\'enyi generalizations~\cite{AlcarazRajabpour2014,Alcaraz2016}; a distinct generalized mutual information based on R\'enyi divergence was studied in Ref.~\cite{AlcarazRajabpour2015}.  In the interval convention above, the coefficient of Ref.~\cite{Stephan2014} corresponds to $\gamma_1=b_1/8$, as derived in Sec.~\ref{sec:subsystem-direct}.  More generally, critical quadratic chains separate into $U(1)$ and non-$U(1)$ branches, with the TFI value at $n=1$ belonging to a distinct $\alpha\leq1$ non-$U(1)$ branch rather than to the $\alpha>1$ central-charge continuation~\cite{Tarighi2022}.  Related classical and quantum mutual-information diagnostics were studied in Refs.~\cite{Wilms2011,UmParkHinrichsen2012}.

The obstacle is not the solution of the Ising Hamiltonian itself.  After the Jordan--Wigner transformation the ground state is Gaussian, so individual configuration probabilities are accessible through free-fermion determinants or Pfaffians.  Emptiness and formation probabilities in XY-type chains have a long history, beginning with Toeplitz-determinant representations and asymptotics~\cite{Shiroishi2001,AbanovFranchini2003,FranchiniAbanov2005} and extending to conformal and Painlev\'e descriptions of their universal scaling~\cite{StephanEFP2014,Rajabpour2015,Rajabpour2016,AresViti2020}.  Exact formulas for general configurations and their Shannon information were developed in Refs.~\cite{NajafiRajabpour2016,NajafiRajabpour2020,AresRajabpourViti2020}; closely related full-counting-statistics results appear in Refs.~\cite{Groha2018,AresRajabpourViti2021}, and Pfaffian representations in noncomputational bases in Ref.~\cite{TarighiKhassehRajabpour2024}.  Gaussianity nevertheless does not reduce the Shannon entropy to a single determinant because the entropy is a nonlinear functional of the complete probability distribution.  The recent sum-of-powers-of-principal-minors (SPPM) formulation provides a field-theory representation for integer R\'enyi moments; its explicit TFI calculation focuses on $n=2$ rather than on a direct evaluation at $n=1$~\cite{NajafiRamezanpourRajabpour2025}.

The full-system Shannon and Shannon--R\'enyi entropies of critical $c=1$ chains, including the $XX/XXZ$ family, were studied earlier in Refs.~\cite{Stephan2009,Stephan2011}.  Related subsystem quantities were subsequently analyzed through Shannon and R\'enyi mutual information, including both free-boson and Ising critical chains, in Ref.~\cite{Stephan2014}.  More recently, a structurally different exact correspondence has been established: for Gaussian eigenstates, the stabilizer R\'enyi entropy can be mapped exactly to the computational-basis Shannon--R\'enyi entropy of a number-conserving free-fermion state on a doubled system; for the critical TFI ground state this becomes the half-filled $XX$-chain ground state of length $2L$~\cite{RamirezTrinoRajabpour2026}.  A subsequent fugacity-resolved formulation shows that the associated TFI all-minors problem is a common finite-size building block for these stabilizer and $XX$-chain Shannon--R\'enyi quantities and maps it to a checkerboard-weighted discrete Selberg ensemble; at positive integer R\'enyi index it reduces to a finite aliased Dyson constant term, with determinant or Pfaffian product formulas at special indices~\cite{KhassehRajabpour2026}.  These results demonstrate strong exact solvability in neighboring information measures, but they do not directly evaluate the computational-basis TFI Shannon entropy studied here.  For the computational-basis TFI probability law considered here, the relevant Cauchy principal minors instead factorize into graph-cut characteristic functions.  At $n=1$, the ordinary Shannon chain rule then leaves a conditional entropy inside weighted cycle-space fibers.  Controlling that residual information is what necessitates the linked-support and ghost-vertex constructions developed below.

The singularity of the Shannon point is therefore a physical part of the problem rather than merely a technical inconvenience.  Boundary-CFT analyses show that R\'enyi reweighting can change the effective boundary condition and produce boundary transitions~\cite{Zaletel2011,Stephan2011}.  Mathematical work on related Book-Ising geometries has established long-range order in regimes motivated in part by these R\'enyi constructions~\cite{DuminilCopinGarbanTassion2023}.  Most directly relevant here, Ref.~\cite{RajabpourBoundaryPhase2026} established an exact finite-size Cauchy/parity representation for the critical Ising measurement law and used its weak inverse-square classical interaction to derive the real-$n>1$ interval term $\frac{c n}{8(n-1)}\log\ell$, with $c=1/2$ for Ising, directly from the interactions removed by a spatial cut.  That construction is valid for real $n>0$ and shows that the explicit long-distance sectors become marginal at the endpoint $n=1$.  The present work develops that endpoint directly: the ordinary Shannon chain rule isolates the remaining conditional information stored in graph cycle space.  Selective measurements and post-measurement entanglement in critical quantum chains were studied earlier in Ref.~\cite{NajafiRajabpourMeasurement2016}.  More recently, representative studies have investigated measurement-altered critical states and measurement-induced boundary criticality~\cite{SunYaoJian2023,Murciano2023,LiuMurcianoMrossAlicea2025,KhannaMurcianoVasseur2026}.  A complementary snapshot-level viewpoint was developed for the critical $XX$ chain, where partial occupation outcomes encode a compact conformal boundary coordinate with exact free-fermion determinant moments~\cite{RajabpourSnapshots2026}.  Here the state is not modified by a measurement protocol, and the object of interest is the Born distribution of the unmeasured critical ground state.

Building on the exact finite-size Cauchy/parity structure of Ref.~\cite{RajabpourBoundaryPhase2026}, the starting point here is its direct Shannon-point probability law.  Every Walsh coefficient of the Born distribution is a principal minor of a trigonometric Cauchy matrix, and Cauchy factorization turns each minor into a product over a graph cut.  That cut product is exactly the characteristic function of the odd-degree boundary of a complete graph with independent Bernoulli edges.  Identifying a spin configuration with its set $D$ of down spins therefore gives, at finite size,
\begin{align}
 p_L(D)&=\PP\{\partial F=D\},\label{eq:intro-parity-law}\\
 H_1&=H(F)-H(F\mid\partial F).\label{eq:intro-entropy-split}
\end{align}
Here $F$ is the random edge set and $\partial F$ its set of odd-degree vertices.  The first identity is an equality of complete probability laws; the second follows from the ordinary Shannon chain rule.  It separates an explicit independent-edge entropy from the information hidden in weighted cosets of the graph cycle space.  This is the key structural distinction from the discrete-Selberg all-minors organization described above: after the explicit edge contribution is removed, the remaining Shannon information is uncertainty within weighted boundary fibers rather than another independent-minor sum.  The Bernoulli weights and the determinant-to-cut derivation are given in Sec.~\ref{sec:probability}.

For the periodic chain, the explicit edge contribution contains the universal anomaly $-\pi^2/16$.  The remaining conditional entropy admits an exact finite-size M\"obius expansion by vertex support.  In the macroscopic even sectors its leading ambiguity is carried by perfect matchings, and a Pfaffian--Hafnian identity reduces their matching sum to determinants.  Conditional on the finite-part matching conjecture and the stated forest closure of the odd sectors, every nonvanishing linked intercept is represented by an explicit finite-dimensional linked finite-part integral, yielding an all-orders analytic hierarchy for the periodic constant.  The four-vertex continuum contribution is evaluated in closed form, while the six- and eight-vertex coefficients are extracted from exact finite-size linked sectors.  Already through eight vertices the resulting estimate is $0.25453$, within about $2\times10^{-4}$ of the high-precision benchmark.  Independently, reconstruction of the complete finite-size distribution gives the entropy density and the periodic constant, in agreement with Ref.~\cite{LauGrassberger2013}.

For an interval, marginalizing the infinite exterior has an exact reduction: all exterior edges incident on a physical site collapse to one parity bit, and these parity bits are independent across sites.  They may therefore be represented by edges to a single ghost vertex.  The internal-edge and ghost-edge entropies separately contain squared logarithms whose coefficients cancel, leaving the explicit edge contribution $\gamma_E=\log2/4-1/16$.  The remaining boundary cycle entropy admits an exact finite-size M\"obius decomposition at every support order into physical-line and ghost-linked sectors.  At arbitrary ghost support the primitive common-scale fiber is a monomer--dimer ensemble; conditional on closure of the boundary forest subtraction, the line and ghost sectors give an all-orders hierarchy of renormalized periods for $\gamma_C$ and hence for $\gamma_1=\gamma_E-\gamma_C$.  Nested endpoint and collision scales make this hierarchy less direct than the periodic one, and its all-order renormalization and summation are not proved.  At total augmented support three, the physical-line and ghost boundary-layer anomalies cancel exactly, so the first complete linked coefficient is obtained in closed form as $\kappa_3^{+}=2I_2$; the same cancellation persists throughout the auxiliary Cauchy family.  A direct reconstruction of the complete interval distribution from exact infinite-line Walsh coefficients independently recovers $\gamma_1=0.060020(3)$, providing an internal check against $b_1/8$ from Ref.~\cite{Stephan2014} without using that value in the fit.  The exact Cauchy deformation additionally identifies the hyperbolic coordinate $\operatorname{arcosh}2=\log(2+\sqrt3)$ as a natural constant of the boundary problem.

The same algebra also extends the interval result beyond the Ising chain.  Time-reversal-invariant critical quadratic chains can be organized through a complex symbol and the zeros that control their critical structure~\cite{VerresenJonesPollmann2018,JonesVerresen2019,Tarighi2022}.  For the sequence $f_m(z)=z^m-1$, the computational-basis correlation matrix decomposes exactly into residue classes modulo $m$, so the interval probability distribution factorizes into independent TFI blocks.  The corresponding entropy identity is proved in Appendix~\ref{app:decimation}.  On the experimental side, subleading Shannon--R\'enyi terms have recently been used to extract central charge on a quantum processor~\cite{Koyluoglu2026}; complementary field-theory work develops stabilizer entropy as a criticality diagnostic from a distinct boundary-CFT perspective~\cite{Hoshino2026}.

The paper is organized as follows.  Section~\ref{sec:model} derives the determinant, Walsh and Bernoulli descriptions of the critical probability law.  Sections~\ref{sec:periodic-edge} and~\ref{sec:subsystem-setup} treat the periodic and interval geometries separately, including their linked all-support organizations.  Section~\ref{sec:cauchy-family} develops the Cauchy deformation and a formal grand organization of the periodic hierarchy.  The appendices contain the algebraic proofs, asymptotic calculations, boundary-forest construction and numerical procedures used in the finite-size analysis.

\section{Exact measurement law and entropy decomposition}
\label{sec:model}

\subsection{Model and measurement conventions}

We use
\begin{equation}
\mathcal H(h)=-\frac12\sum_{j=0}^{L-1}
\left(\sigma_j^x\sigma_{j+1}^x+h\sigma_j^z\right),
\qquad \sigma_L^x=\sigma_0^x,
\label{eq:tfi-ham}
\end{equation}
for the periodic chain and the corresponding infinite-volume Hamiltonian on $\mathbb Z$ for the subsystem problem~\cite{Pfeuty1970}.  We work at $h=1$ and take $L$ even.  All measurement entropies are in the $\sigma^z$ basis unless stated otherwise.

We write $[L]=\{0,1,\ldots,L-1\}$ and denote the even critical ground state by $|G\rangle$.  A configuration is $\bm\sigma=(\sigma_0^z,\ldots,\sigma_{L-1}^z)$ with $\sigma_j^z=\pm1$ and probability
\begin{equation}
 p_{\bm\sigma}=|\langle\bm\sigma|G\rangle|^2.
\end{equation}
For a probability distribution $p_{\bm\sigma}$ we use natural logarithms and define the local-basis Shannon--R\'enyi entropies by
\begin{align}
 H_n&=\frac{1}{1-n}\log\sum_{\bm\sigma}p_{\bm\sigma}^{\,n},
 \qquad n>0,\quad n\neq1,\label{eq:def-renyi}\\
 H_1&=-\sum_{\bm\sigma}p_{\bm\sigma}\log p_{\bm\sigma}.\label{eq:def-shannon}
\end{align}
These are the classical Shannon and R\'enyi entropies of the Born probabilities~\cite{Shannon1948,Renyi1961}; they are not reduced-density-matrix entanglement entropies.  For any discrete random variable or vector $X$, $H(X)$ denotes its ordinary Shannon entropy, and $H(X\mid Y)$ the corresponding conditional entropy.
The basis dependence is physical: a change of local measurement basis changes the probability distribution and therefore changes $H_n$.  At the self-dual point the $x$- and $z$-basis entropies of the even periodic chain differ by the exact duality shift \cite{Stephan2010},
\begin{equation}
H_n^{(x)}(L)=H_n^{(z)}(L)+\log2,
\label{eq:KW}
\end{equation}
which affects only the $O(1)$ term.  The finite-size scaling appropriate to the full system and to the subsystem will be introduced separately in \Cref{sec:periodic-edge,sec:subsystem-setup}.

\subsection{Exact Born distribution as a Bernoulli parity graph}
\label{sec:probability}

At criticality the finite-ring probability distribution can be written as
\begin{equation}
p_{\bm\sigma}=2^{-L}\det(I+D_{\bm\sigma}G_L),
\label{eq:det-born}
\end{equation}
where $D_{\bm\sigma}=\diag(\sigma_0^z,\ldots,\sigma_{L-1}^z)$ and
\begin{equation}
(G_L)_{j\ell}=\frac{(-1)^{j-\ell}}{L\sin\!\left[\frac\pi L(j-\ell+\frac12)\right]}.
\label{eq:GL}
\end{equation}
Equation~\eqref{eq:det-born} is the standard Gaussian formation-probability representation specialized to the critical TFI ground state.  Appendix~\ref{app:det-walsh} explains how it follows from the local spin projectors and Wick's theorem and then derives the determinant-to-Walsh step explicitly.  Determinant formulas for computational-basis formation probabilities in quadratic chains are developed in Refs.~\cite{NajafiRajabpour2016,NajafiRajabpour2020,AresRajabpourViti2020}, while Pfaffian extensions to noncomputational bases are given in Ref.~\cite{TarighiKhassehRajabpour2024}.  The feature special to the present problem is the trigonometric Cauchy structure of \eqref{eq:GL}, which lets us factor every principal minor rather than only selected formation probabilities.  The same finite-size Cauchy probability kernel underlies the real-$n$ boundary-phase construction of Ref.~\cite{RajabpourBoundaryPhase2026}.

Expanding \eqref{eq:det-born} in Walsh characters gives
\begin{equation}
p_{\bm\sigma}=2^{-L}\sum_{S\subseteq[L]}d_L(S)\chi_S(\bm\sigma),
\qquad
\chi_S(\bm\sigma)=\prod_{j\in S}\sigma_j^z,
\label{eq:walsh}
\end{equation}
with coefficients equal to principal minors,
\begin{equation}
d_L(S)=\det G_L[S].
\label{eq:principal-minors}
\end{equation}
Trigonometric Cauchy factorization yields (see Appendix~\ref{app:cauchy})
\begin{equation}
 d_L(S)=A_L^{|S|}\prod_{\substack{i<j\\i,j\in S}}g_{j-i}^{(L)},
\qquad
A_L=\frac{1}{L\sin(\pi/2L)},
\label{eq:cauchy-factor}
\end{equation}
where
\begin{equation}
 g_m^{(L)}=\frac{\sin^2(\pi m/L)}{
 \sin[\pi(m-\frac12)/L]\sin[\pi(m+\frac12)/L]}
 =\frac{1}{1-z_m^{(L)}},
\label{eq:g-z}
\end{equation}
with
\begin{equation}
 z_m^{(L)}=\frac{\sin^2(\pi/2L)}{\sin^2(\pi m/L)}.
\end{equation}

Now introduce independent Bernoulli edges $\eta_{ij}\in\{0,1\}$ on the complete graph with vertex set $V=[L]$.  This long-range complete graph is an auxiliary probability representation, not the physical interaction graph of the TFI Hamiltonian: the Hamiltonian remains the nearest-neighbour spin chain of Eq.~\eqref{eq:tfi-ham}, while the Bernoulli edges encode the measurement probability law.  We write $\PP$ and $\EE$ for probability and expectation.  Their occupation probabilities are
\begin{equation}
\beta_{ij}^{(L)}\equiv\PP(\eta_{ij}=1)
=\frac12\left[1-\sqrt{1-z_{|i-j|}^{(L)}}\right].
\label{eq:beta-ring}
\end{equation}
Then
\begin{equation}
1-2\beta_{ij}^{(L)}=\sqrt{1-z_{|i-j|}^{(L)}}.
\end{equation}
For an edge configuration $F$, let $\partial F\subseteq V$ denote the set of vertices with odd incident edge parity.  Its Walsh transform is
\begin{equation}
\EE\,\chi_S(\partial F)=
\prod_{\substack{i\in S\\j\notin S}}(1-2\beta_{ij}^{(L)}).
\label{eq:cut-transform}
\end{equation}
The cut-product identity for the Cauchy minors gives exactly
\begin{equation}
\prod_{\substack{i\in S\\j\notin S}}(1-2\beta_{ij}^{(L)})=d_L(S).
\end{equation}
Identify a spin configuration with the subset $D=\{j\in[L]:\sigma_j^z=-1\}$, so that $\chi_S(\bm\sigma)=(-1)^{|S\cap D|}$.  The Walsh characters form a complete orthogonal basis on the Boolean hypercube, so equality of all Walsh coefficients determines the probability law uniquely.  Since the Born law and the Bernoulli boundary law have the same coefficient $d_L(S)$ for every $S\subseteq[L]$, Walsh inversion gives the central finite-size statement below.  The orthogonality and inversion argument is given in Appendix~\ref{app:det-walsh}, while the cut-product step is proved in Appendix~\ref{app:bernoulli-proof}:
\begin{equation}
p_L(D)=\PP\{\partial F=D\}.
\label{eq:parity-law}
\end{equation}
Appendix~\ref{app:exact-checkpoints} also gives an equivalent high-temperature-current representation of the same parity law, in which the auxiliary long-range coupling has the inverse-square tail $J_m=1/(16m^2)+O(m^{-4})$.

Equation~\eqref{eq:parity-law} immediately gives the entropy identity
\begin{align}
H_1&=H(\partial F)=H(F)-H(F\mid\partial F)\nonumber\\
&\equiv H_{\edge}-H_{\cyc}.
\label{eq:edge-cycle}
\end{align}

The boundary-map interpretation of the conditional term and its decomposition into weighted cycle-space fibers are summarized in \Cref{fig:cycle-fibers}.

\begin{figure}[t]
\centering
\includegraphics[width=\textwidth]{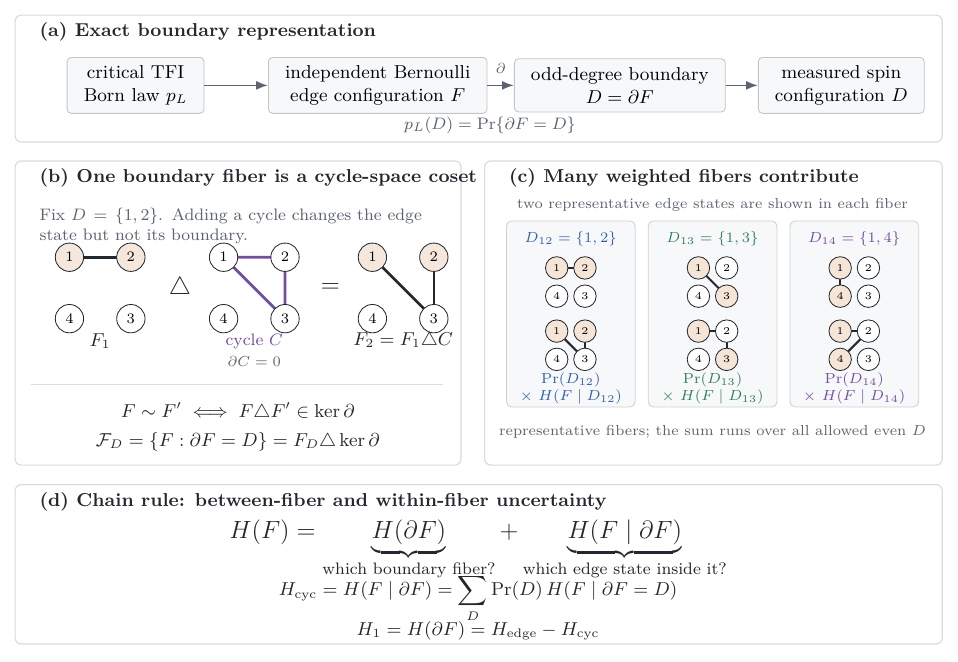}
\caption{Cycle entropy as uncertainty inside weighted boundary fibers.  (a) The critical TFI Born law is represented exactly as the odd-degree boundary $D=\partial F$ of an independent Bernoulli edge configuration $F$.  (b) For fixed boundary $D$, compatible edge configurations form an affine coset of the binary cycle space: adding a cycle $C$, which is an even-degree subgraph with $\partial C=0$, changes $F$ but leaves $\partial F$ unchanged.  (c) Different allowed boundaries label distinct weighted fibers; two representative edge states are shown in each displayed fiber, while the conditional entropy averages over all edge configurations and all allowed even $D$.  (d) The Shannon chain rule separates uncertainty between fibers, $H(\partial F)=H_1$, from uncertainty within fibers, $H(F\mid\partial F)=H_{\cyc}$, giving $H_1=H_{\edge}-H_{\cyc}$.}
\label{fig:cycle-fibers}
\end{figure}

Since the edges are independent,
\begin{equation}
H_{\edge}=\sum_{i<j}h(\beta_{ij}),
\qquad
h(p)=-p\log p-(1-p)\log(1-p).
\end{equation}
The second term is a weighted cycle-code entropy: two edge configurations have the same boundary iff their symmetric difference is an even subgraph, i.e. an element of the binary cycle space, as illustrated explicitly in panel~(b) of \Cref{fig:cycle-fibers}.  This interpretation is used throughout the rest of the paper.  The probability-law representation itself is independent of the entropy functional and may also be used to evaluate local-basis R\'enyi entropies.  What is special to the Shannon point is the ordinary chain rule leading to Eq.~\eqref{eq:edge-cycle}; there is no identical edge-minus-conditional-entropy decomposition at generic R\'enyi index.  Appendix~\ref{app:renyi-moment} gives an exact normalized-moment identity whose derivative at $n=1$ reproduces the cycle entropy, while Appendix~\ref{app:exact-checkpoints} records the exact support check $H_0^{(z)}=(L-1)\log2$.

The determinant-to-Walsh, Cauchy-factorization and cut-product steps are given respectively in Appendices~\ref{app:det-walsh}, \ref{app:cauchy} and \ref{app:bernoulli-proof}.  With the common probability law established, we now keep the two geometries separate and begin with the complete periodic chain.

\section{Full periodic system: universal Shannon constant}
\label{sec:periodic-edge}

For an even periodic chain of length $L$ we write
\begin{equation}
H_1^{(z)}(L)=s_1L+r_1^{(z)}+o(1),
\label{eq:pbc-asympt}
\end{equation}
with $H_1^{(x)}(L)=H_1^{(z)}(L)+\log2$ from \eqref{eq:KW}.  The target of this section is the universal $O(1)$ constant, not the subsystem logarithmic coefficient studied in \Cref{sec:subsystem-setup}.  Writing
\begin{equation}
H_{\cyc}(L)=s_{\cyc}L+c_{\cyc}+o(1),
\end{equation}
the asymptotic coefficient bookkeeping implied by \eqref{eq:edge-cycle} and the derived edge anomaly below is
\begin{equation}
s_1=s_{\edge}-s_{\cyc},\qquad
r_1^{(z)}=-\frac{\pi^2}{16}-c_{\cyc},\qquad
r_1^{(x)}=\log2-\frac{\pi^2}{16}-c_{\cyc}.
\label{eq:periodic-bookkeeping}
\end{equation}
Thus the edge contribution is obtained analytically, whereas the complete $c_{\cyc}$ (and hence the periodic constants in either basis) is not claimed here in closed form: it is determined numerically and organized by the linked expansion developed below.

\subsection{Independent-edge contribution}

For the ring the edge entropy is
\begin{equation}
H_{\edge}(L)=\frac{L}{2}\sum_{m=1}^{L-1}h\!\left[
\frac12\left(1-\sqrt{1-z_m^{(L)}}\right)\right].
\label{eq:hedge-ring}
\end{equation}
The bulk density is the absolutely convergent infinite-line series
\begin{equation}
s_{\edge}=\sum_{m=1}^{\infty}h\!\left[
\frac12\left(1-\sqrt{1-\frac1{4m^2}}\right)\right]
=0.5163350281651047631\ldots .
\label{eq:sedge}
\end{equation}
The universal ring correction comes solely from the nonanalytic $-z\log z/4$ term in
\begin{equation}
h\!\left(\frac{1-\sqrt{1-z}}{2}\right)
=\frac z4\,[1+\log4-\log z]+O(z^2|\log z|).
\end{equation}
A singular Euler--Maclaurin calculation, given in Appendix~\ref{app:edge-anomaly}, yields
\begin{equation}
H_{\edge}(L)=s_{\edge}L-\frac{\pi^2}{16}+o(1).
\label{eq:periodic-edge-result}
\end{equation}
The closed analytic value $-\pi^2/16$ is a derived asymptotic coefficient.  It isolates a contribution that is fully determined analytically from the remaining part of the periodic constant.

\subsection{Vertex-linked expansion and matching kernels}
\label{sec:periodic-linked}

The conditional entropy $H_{\cyc}$ is extensive, so expanding directly in graph cycles is badly redundant.  We instead reorganize it by vertex support using M\"obius inversion on the subset lattice.  This is a connected-cluster expansion in support space rather than a literal enumeration of graph cycles, and it does not require choosing a cycle basis.

For a finite vertex set $S$ let
\begin{equation}
C_L(S)=H(F_S\mid\partial F_S)
\end{equation}
be the conditional entropy of the induced graph on $S$.  Define the vertex M\"obius cumulant
\begin{equation}
\Psi_L(S)=\sum_{T\subseteq S}(-1)^{|S|-|T|}C_L(T).
\label{eq:psi}
\end{equation}
The quantity $\Psi_L(S)$ is a signed linked coefficient, not itself a Shannon entropy: it is the inclusion--exclusion remainder assigned to the complete support $S$ after contributions associated with proper sub-supports have been removed, and it need not be nonnegative.  The inverse M\"obius relation, together with explicit three- and four-vertex examples, is given in Appendix~\ref{app:mobius}.  Applying that inverse relation to the full vertex set gives the exact finite-size identity
\begin{equation}
H_{\cyc}(L)=\sum_{k=3}^{L}A_{k,L},\qquad
A_{k,L}=\sum_{|S|=k}\Psi_L(S).
\label{eq:linked-expansion}
\end{equation}
The expansion begins at three vertices because a graph on one or two vertices has no cycle.

\paragraph{From the exact lattice sector to a macroscopic support.}
For an ordered $k$-set on the ring introduce positive cyclic gaps $g_1+\cdots+g_k=L$.  Translation invariance gives the exact integer-gap formula
\begin{equation}
A_{k,L}=\frac{L}{k}\sum_{g_1+\cdots+g_k=L}\Psi_{k,L}(g_1,\ldots,g_k).
\label{eq:gap-sum}
\end{equation}
Equation~\eqref{eq:gap-sum}, rather than a continuum integral, is the starting point for the fixed-support asymptotics.  The question is which regions of this discrete sum can contribute to the $L^0$ intercept of $A_{k,L}$.

First consider its fully macroscopic interior.  Write $g_a=L u_a$, with $u_a>0$ and $\sum_a u_a=1$, and define the associated ordered points on the unit circle by
\begin{equation}
x_1=0,\qquad x_j=\sum_{a=1}^{j-1}u_a\pmod 1\quad (j=2,\ldots,k).
\label{eq:gap-to-points}
\end{equation}
For clarity one may temporarily restrict to a truncated simplex
\begin{equation}
\Delta_{k-1}^{(\epsilon)}=
\left\{u_a\ge \epsilon,\ \sum_{a=1}^{k}u_a=1\right\},
\qquad 0<\epsilon<1/k,
\label{eq:truncated-simplex}
\end{equation}
so that every pairwise separation remains uniformly macroscopic.  At a separation $m=xL$ one edge has
\begin{equation}
\beta_{xL}^{(L)}=\frac{\pi^2}{16L^2}\csc^2(\pi x)+O(L^{-4}).
\label{eq:macro-edge}
\end{equation}
It is convenient to write
\begin{equation}
K_{ij}=\frac{\pi}{4\sin\pi(x_i-x_j)}\quad(i\ne j),\qquad K_{ji}=-K_{ij},
\label{eq:K-macro}
\end{equation}
so that $\beta_{ij}^{(L)}=L^{-2}K_{ij}^2+O(L^{-4})$.  For fixed $k$ and $u\in\Delta_{k-1}^{(\epsilon)}$, an edge configuration $E$ containing $m$ occupied edges therefore has
\begin{equation}
\PP\{F_S=E\}
=L^{-2m}\prod_{(ij)\in E}K_{ij}^2\,[1+O_\epsilon(L^{-2})].
\label{eq:macro-edge-configuration}
\end{equation}
The absent-edge factors contribute only the bracket in Eq.~\eqref{eq:macro-edge-configuration}.  Thus, in the macroscopic interior, the rare-edge expansion is ordered simply by the number of occupied long edges.

The role of the M\"obius transform is particularly transparent in this expansion.  Each term in the rare-edge expansion has a vertex support: the set of vertices incident on the edge factors appearing in that term.  A term carried by a proper subset $R\subsetneq S$ appears in every induced graph $T$ with $R\subseteq T\subseteq S$, and its coefficient in $\Psi_L(S)$ is multiplied by
\begin{equation}
\sum_{R\subseteq T\subseteq S}(-1)^{|S|-|T|}
=(1-1)^{|S|-|R|}=0.
\label{eq:mobius-support-filter}
\end{equation}
Thus, order by order in the macroscopic rare-edge expansion, the M\"obius combination acts as a support projector: only terms whose edge factors require all vertices of $S$ can survive in $\Psi_L(S)$.  This is the precise asymptotic meaning of the ``full-support'' terminology used below.

\paragraph{Why perfect matchings are the leading even-support ambiguity.}
Let $|S|=2p$ and focus first on the boundary $D=S$, for which every vertex is odd.  Any edge configuration with this boundary has at least $p$ occupied edges.  Indeed, the sum of the degrees is at least $2p$, while it equals twice the number of edges.  If exactly $p$ edges are present, every vertex must have degree one, so the configuration is a perfect matching.  Conversely every perfect matching has boundary $S$.  Therefore the minimal realizations of the all-odd boundary are \emph{exactly} the perfect matchings of the $2p$ vertices.

Let $\cM_{2p}$ denote this set of perfect matchings and define
\begin{equation}
w_M=\prod_{(ij)\in M}K_{ij}^2,
\qquad
Z_p=\sum_{M\in\cM_{2p}}w_M.
\label{eq:matching-weights}
\end{equation}
Equation~\eqref{eq:macro-edge-configuration} then gives, uniformly on $\Delta_{2p-1}^{(\epsilon)}$,
\begin{align}
\PP\{F_S=M\}&=L^{-2p}w_M\,[1+O_\epsilon(L^{-2})],
\label{eq:matching-probability}\\
\PP\{\partial F_S=S\}&=L^{-2p}Z_p\,[1+O_\epsilon(L^{-2})].
\label{eq:allodd-probability}
\end{align}
Configurations in the same boundary fiber with more than $p$ edges are down by at least another factor $L^{-2}$.  Consequently, conditioned on the all-odd boundary,
\begin{equation}
\PP\{F_S=M\mid\partial F_S=S\}
=\frac{w_M}{Z_p}+O_\epsilon(L^{-2}).
\label{eq:matching-conditional-law}
\end{equation}
The leading conditional uncertainty in this fiber is therefore the Shannon entropy of the normalized matching weights.  If $\mu_M=w_M/Z_p$, its contribution to $C_L(S)$ is
\begin{align}
&\PP\{\partial F_S=S\}
H(F_S\mid\partial F_S=S)\nonumber\\
&\hspace{1cm}=L^{-2p}Z_pH(\mu)
+O_\epsilon(L^{-2p-2}\log L)\nonumber\\
&\hspace{1cm}=L^{-2p}\Gamma_{2p}(u)
+O_\epsilon(L^{-2p-2}\log L),
\label{eq:allodd-entropy-scaling}
\end{align}
where
\begin{equation}
\Gamma_{2p}=Z_p\log Z_p-\sum_{M}w_M\log w_M.
\label{eq:Gamma-matching}
\end{equation}
Thus $\Gamma_{2p}$ is not introduced as an independent continuum ansatz: it is the leading conditional Shannon ambiguity inside the rare all-odd boundary fiber of the exact Bernoulli graph.

It remains to connect this fiber calculation to the linked coefficient $\Psi_L(S)$.  The support filter in Eq.~\eqref{eq:mobius-support-filter} removes lower-order ambiguities that can already be realized on proper vertex subsets.  At order $p$ in the macroscopic edge activities, any graph using all $2p$ vertices must have exactly $p$ edges and hence degree one at every vertex; it is therefore one of the perfect matchings above and has the all-odd boundary.  A proper boundary can involve all $2p$ vertices only with more than $p$ edges and is suppressed by at least $L^{-2}$ relative to Eq.~\eqref{eq:allodd-entropy-scaling}.  Accordingly, away from collision faces the first primitive full-support term is
\begin{equation}
\Psi_{2p,L}(Lu_1,\ldots,Lu_{2p})
=L^{-2p}\Gamma_{2p}(u)
+O_\epsilon(L^{-2p-2}\log L).
\label{eq:Psi-macro-matching}
\end{equation}
Equation~\eqref{eq:Psi-macro-matching} is the missing bridge between the exact M\"obius sector and the continuum matching kernel: the M\"obius transform selects full vertex support, while minimal full-support parity configurations at even order are perfect matchings.

For four vertices this mechanism can be seen without any abstraction.  There are precisely three two-edge configurations using all four vertices,
\begin{equation}
M_1=(12)(34),\qquad
M_2=(13)(24),\qquad
M_3=(14)(23),
\end{equation}
and all three have boundary $S=\{1,2,3,4\}$.  Writing their weights as $w_1,w_2,w_3$, the all-odd fiber contributes
\begin{equation}
L^{-4}\left[
(w_1+w_2+w_3)\log(w_1+w_2+w_3)
-\sum_{a=1}^{3}w_a\log w_a
\right]
+O_\epsilon(L^{-6}\log L).
\label{eq:Gamma4-explicit-fiber}
\end{equation}
All order-$L^{-2}$ ambiguities are carried by smaller supports and disappear from the four-vertex M\"obius coefficient; Eq.~\eqref{eq:Gamma4-explicit-fiber} is the first term requiring the complete four-point support.

\paragraph{Why the discrete gap sum produces an $O(1)$ continuum integral.}
The passage to an integral is now a counting statement rather than a separate assumption.  Restrict Eq.~\eqref{eq:gap-sum} to gap vectors with $g_a/L\in\Delta_{2p-1}^{(\epsilon)}$ and denote this interior contribution by $A_{2p,L}^{(\epsilon)}$.  The integer compositions of $L$ form a Riemann mesh of dimension $2p-1$, so for fixed $\epsilon$
\begin{equation}
L^{-(2p-1)}
\sum_{\substack{g_1+\cdots+g_{2p}=L\\ g/L\in\Delta_{2p-1}^{(\epsilon)}}}
F(g/L)
\longrightarrow
\int_{\Delta_{2p-1}^{(\epsilon)}}F(u)\,\dd^{2p-1}u.
\label{eq:gap-riemann}
\end{equation}
Combining Eqs.~\eqref{eq:gap-sum}, \eqref{eq:Psi-macro-matching} and \eqref{eq:gap-riemann} gives
\begin{equation}
A_{2p,L}^{(\epsilon)}
\longrightarrow
\frac1{2p}
\int_{\Delta_{2p-1}^{(\epsilon)}}
\Gamma_{2p}(u)\,\dd^{2p-1}u.
\label{eq:interior-period}
\end{equation}
The origin of the constant scaling is visible directly in the powers of $L$:
\begin{equation}
\underbrace{L}_{\text{translation}}
\times
\underbrace{L^{2p-1}}_{\text{gap configurations}}
\times
\underbrace{L^{-2p}}_{\text{matching ambiguity}}
=L^0.
\label{eq:constant-power-balance}
\end{equation}
This balance is the reason a continuum period can govern an $O(1)$ term of the discrete entropy sector.

The fixed-$\epsilon$ result \eqref{eq:interior-period} is not yet a formula for the lattice intercept $c_{2p}$.  The expansion \eqref{eq:Psi-macro-matching} is uniform only away from the boundary of the simplex.  When one or more $u_a\to0$, some lattice separations cease to be macroscopic, $K_{ij}$ becomes singular, and the continuum limit mixes with microscopic or mesoscopic collision regions.  In the exact lattice sum the smallest gap is one, corresponding to a moving continuum cutoff of order $1/L$.  These collision layers contain the extensive pieces already assigned to lower supports, logarithmic odd-cluster residues, and finite face anomalies.  Therefore replacing the whole integer-gap sum by the ordinary integral of $\Gamma_{2p}$ would double count lower-support information and in general produces a singular integral.

This is why the continuum object entering the constant is a \emph{linked finite part}, not the naive integral.  The finite-part prescription is required to implement, on the collision faces of the continuum simplex, the same hierarchy of lower-support subtraction that is exact on the lattice.  Appendix~\ref{app:finite-part} gives the counterterm organization.  What is controlled directly by the argument above is the macroscopic interior kernel and its $L^0$ power counting; the identification of the fully renormalized continuum finite part with the $L^0$ coefficient of the exact integer-gap sum is an additional lattice--continuum matching statement.

\begin{conjecture}[Finite-part matching]
\label{conj:FPmatch}
For every fixed $p$, the $L^0$ coefficient of the exact lattice sector $A_{2p,L}$ equals the linked finite part obtained from the macroscopic matching kernel with the cutoff inherited from positive integer gaps.
\end{conjecture}

The geometry of this sequence of steps is summarized in \Cref{fig:linked-geometry}: panels~(a) and~(b) show the exact support and M\"obius organization, panel~(c) illustrates the minimal perfect-matching fiber at four vertices, and panel~(d) indicates why collision faces require an additional linked subtraction.

\begin{figure}[t]
\centering
\includegraphics[width=\textwidth]{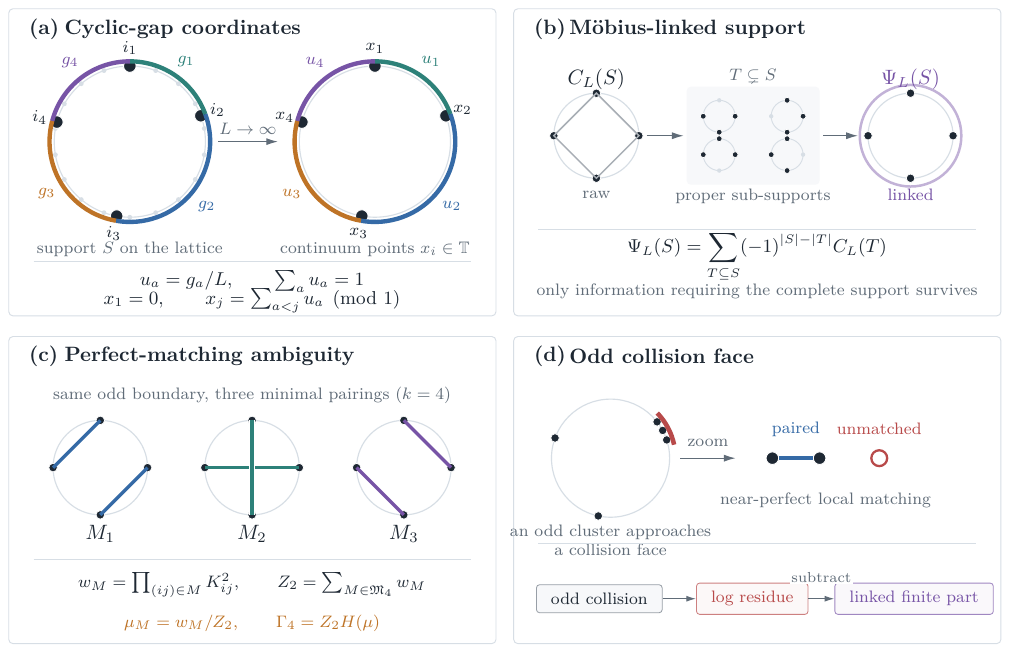}
\caption{Geometry of the periodic vertex-linked expansion.  (a) An ordered support $S$ on the periodic lattice is parametrized by positive cyclic gaps $g_a$, with $\sum_a g_a=L$.  In the macroscopic limit $g_a=L u_a$, the cumulative gaps define continuum points $x_i\in\mathbb T$.  (b) The vertex M\"obius cumulant $\Psi_L(S)$ subtracts contributions carried by proper sub-supports $T\subsetneq S$, leaving the signed linked remainder assigned by inclusion--exclusion to the complete support; the support diagrams in panel~(b) are schematic and do not represent particular edge configurations.  (c) For an even support the leading all-odd boundary ambiguity is realized by perfect matchings.  For $k=4$ there are three minimal pairings $M_1,M_2,M_3$, weighted by $w_M$, whose normalized ambiguity gives $\Gamma_4=Z_2H(\mu)$.  (d) When an odd cluster approaches a collision face, its minimal local realization is a near-perfect matching with one unmatched vertex.  The resulting logarithmic residue is subtracted in defining the linked finite part.}
\label{fig:linked-geometry}
\end{figure}

\paragraph{Cauchy reduction of the matching sum.}
Having identified the probabilistic origin of $\Gamma_{2p}$, we now exploit the special Cauchy kernel to evaluate it efficiently.  The identity
\begin{equation}
\Hf(K_{ij}^2)=\Pf(K)^2=\det K
\label{eq:pfhaf}
\end{equation}
turns the unsigned perfect-matching partition function into a determinant.  Here $\Pf$ is the Pfaffian of an antisymmetric matrix, characterized by $\Pf(K)^2=\det K$, while $\Hf$ is the Hafnian of a symmetric matrix, i.e. the unsigned sum over perfect matchings of products of matrix entries.  Equation~\eqref{eq:pfhaf} is a Cauchy-type Pfaffian--Hafnian identity.  Related Pfaffian--Hafnian analogues of Borchardt's identity were proved in Ref.~\cite{IshikawaKawamukoOkada2005}; Appendix~\ref{app:pfhaf-proof} gives a self-contained proof for the present rational and trigonometric Cauchy kernels.  Below $K_{\widehat i\widehat j}$ denotes the matrix obtained by deleting rows and columns $i,j$.  Moreover,
\begin{equation}
\sum_Mw_M\log w_M
=\sum_{i<j}K_{ij}^2\det K_{\widehat i\widehat j}\log K_{ij}^2,
\end{equation}
so
\begin{equation}
\Gamma_{2p}(K)=\det K\left[\log\det K-
\sum_{i<j}(K_{ij}K^{-1}_{ij})^2\log K_{ij}^2\right].
\label{eq:Gamma-det}
\end{equation}
The first equality follows by exchanging the sum over matchings with the sum over edges and applying Eq.~\eqref{eq:pfhaf} to the complementary minor; Jacobi's complementary-minor identity then gives $\det K_{\widehat i\widehat j}=\det K\,(K^{-1}_{ij})^2$.  Thus the factorial perfect-matching sum is reduced to $O(p^3)$ linear algebra; Appendix~\ref{app:pfhaf-proof} records the details.

Whenever a fixed-support sector admits an expansion
\begin{equation}
A_{k,L}=s_kL+c_k+o(1),
\label{eq:ck-definition}
\end{equation}
we call $c_k$ its linked $L^0$ intercept; it should not be confused with the complete cycle constant $c_{\cyc}$.  The macroscopic support argument also distinguishes even and odd $k$ before any collision analysis.  A graph using all $2p$ vertices can do so with $p$ edges, and the resulting minimal configurations are the perfect matchings above.  For $2p+1$ vertices, by contrast, at least $p+1$ edges are needed to use every vertex.  Its primitive macroscopic weight is therefore $L^{-2p-2}$, while the number of placements is $O(L^{2p+1})$, leaving only $O(L^{-1})$.  The possible loophole is again a collision face on which part of the support collapses to a lower-support structure.  Appendix~\ref{app:odd-selection} analyzes the first cases $k=3,5$ and the general collision-forest power counting.  Conditional on recursive linked collision subtraction closing compatibly with the exact vertex M\"obius expansion at every proper support, the renormalized odd sectors obey
\begin{equation}
A_{2p+1,L}=s_{2p+1}L+O\!\left(\frac{\log^mL}{L}\right),
\qquad
c_{2p+1}=0,
\label{eq:odd-selection-main}
\end{equation}
for some finite $m$ depending on the sector.  We refer to \eqref{eq:odd-selection-main} as the \emph{odd-support selection rule}.  It is a conditional physics-level statement, not an all-order theorem: a rigorous proof would require a uniform factorization/forest theorem and a bound on the renormalized odd remainder.  Accordingly, the even-sector partial sums below remain finite-order comparisons rather than a proved all-orders reconstruction of $c_{\cyc}$.

\subsection{Universal linked coefficients and reconstruction of the constant}
\label{sec:periodic-coeffs}

Conditional on Conjecture~\ref{conj:FPmatch}, the continuum representation of the $2p$-vertex constant is
\begin{equation}
c_{2p}=\frac1{2p}\FP_{\rm linked}
\int_{\Delta_{2p-1}}\Gamma_{2p}(u)\,\dd^{2p-1}u,
\label{eq:c2p-period}
\end{equation}
where $\Delta_{2p-1}$ is the positive cyclic-gap simplex and $\Gamma_{2p}(u)$ means the matching kernel evaluated at the cumulative point coordinates $x(u)$ defined in Eq.~\eqref{eq:gap-to-points}.  The corresponding exact finite-lattice sector $A_{2p,L}$ is defined independently of this conjecture by the M\"obius expansion.

These statements have distinct logical status.  Equation~\eqref{eq:interior-period} follows from the exact gap sum and the controlled macroscopic rare-edge expansion on every truncated simplex: it identifies the primitive interior $L^0$ contribution.  Equation~\eqref{eq:c2p-period} is stronger because it asserts that, after all collision faces are renormalized with the linked prescription, the resulting finite part equals the intercept of the exact lattice sector.  Finally, neither statement by itself gives the complete cycle constant $c_{\cyc}$: reconstructing that quantity requires combining all support orders, controlling the odd sectors, and controlling the higher-even-support tail.  Thus the chain of objects is
\begin{equation}
\Gamma_{2p}\quad\longrightarrow\quad c_{2p}\quad\longrightarrow\quad c_{\cyc}\quad\longrightarrow\quad r_1,
\label{eq:linked-hierarchy}
\end{equation}
with each step requiring a distinct mathematical input.

For four vertices conformal covariance reduces the renormalized continuum problem to one modulus.  This reduction acts on the already identified matching kernel $\Gamma_4$; it is not the step that produces the matching kernel from the lattice model.  The remaining elementary integral is evaluated in Appendix~\ref{app:c4}.  The continuum finite-part calculation gives the value displayed below, while its identification with the exact lattice intercept still uses Conjecture~\ref{conj:FPmatch}:
\begin{equation}
 c_4=-\frac{\pi^3\sqrt3}{256}.
\label{eq:c4-main}
\end{equation}
The same value is independently supported by direct exact-lattice extrapolation.

For six and eight vertices we evaluate the exact finite-$L$ linked sectors and extrapolate
\begin{equation}
A_{2p,L}=s_{2p}L+c_{2p}+\frac{a_1}{L}+\frac{a_2}{L^2}+\cdots.
\end{equation}
The continuum construction is nevertheless explicit beyond four vertices: Appendix~\ref{app:c6-setup} works out the next periodic case $2p=6$ in a projective gauge, lists its fifteen reduced matching weights, writes the resulting three-modulus integral, and exhibits representative pair- and three-point collision counterterms.  That construction is an analytic setup for $c_6$, not an independent evaluation of its finite part or a replacement for Conjecture~\ref{conj:FPmatch}.
The available sizes support
\begin{equation}
 c_6\simeq 0.0346,\qquad c_8\simeq-0.00305.
 \label{eq:coeffs-main}
\end{equation}
At the present sizes the ten-vertex sector is much smaller and its inferred constant changes appreciably when the correction order is varied.    The finite-size sequences and representative fit variations are given in Appendix~\ref{app:periodic-numerics}; the purpose of Eq.~\eqref{eq:coeffs-main} is to establish the scale and sign of the next two linked corrections rather than to claim more digits than the displayed data support.

Keeping $c_4,c_6,c_8$ gives the even-sector estimate
\begin{equation}
 \log2-\frac{\pi^2}{16}-(c_4+c_6+c_8)\simeq0.25453,
 \label{eq:periodic-eight-sector}
\end{equation}
already within about $2\times10^{-4}$ of the high-precision value $0.254392505(10)$ reported in Ref.~\cite{LauGrassberger2013}.  This comparison is useful as evidence that the support expansion is organizing the constant efficiently, but it is not a convergence proof: the odd-sector statement is conditional on the forest-closure assumption in Eq.~\eqref{eq:odd-selection-main}, and a bound on the higher-support tail remains open.

Independently of the linked-support expansion, the cycle entropy can be extracted directly from the exact Walsh distribution (Appendix~\ref{app:finite-data}).  Exact sizes through $L=24$ give
\begin{equation}
H_{\cyc}(L)=s_{\cyc}L+c_{\cyc}+O(L^{-1}),
\end{equation}
with
\begin{equation}
 s_{\cyc}=0.09306147(2),\qquad
 c_{\cyc}=-0.1780956(3),
\label{eq:ccyc-direct}
\end{equation}
and hence
\begin{equation}
s_1=s_{\edge}-s_{\cyc}=0.42327356(2).
\end{equation}
Together with \eqref{eq:periodic-bookkeeping}, the same finite-size analysis gives $r_1^{(x)}=0.2543925(3)$.  The uncertainty reflects the spread under changes of the finite-size window and correction order, not sampling noise.  The more precise external value $0.254392505(10)$ from Ref.~\cite{LauGrassberger2013} is quoted only for comparison.  This determination does not use the linked coefficients $c_{2p}$.

The periodic linked program is therefore reduced, conditional on the finite-part matching conjecture, to a sequence of explicit finite-dimensional conformal periods.  What remains open is the general lattice--continuum matching theorem, a rigorous all-order proof of the conditional odd-support selection rule, and a controlled all-orders sum or tail bound.

\FloatBarrier
\section{Interval subsystem: universal Shannon logarithm}
\label{sec:subsystem-setup}

For the subsystem we retain only the interval $A_\ell=\{1,\ldots,\ell\}$ of the infinite critical chain and marginalize all exterior spins,
\begin{equation}
p_A(\sigma_A)=\sum_{\sigma_{\bar A}}p(\sigma_A,\sigma_{\bar A})
=\langle\sigma_A|\rho_A|\sigma_A\rangle,
\end{equation}
where $\rho_A$ is the reduced density matrix of the interval.  Its Shannon entropy scales as
\begin{equation}
H_1(A_\ell)=s_1\ell+\gamma_1\log\ell+\kappa_1+o(1),
\label{eq:sub-asympt}
\end{equation}
where the universal target of this section is $\gamma_1$.  This is distinct from the $O(1)$ full-system constant studied in \Cref{sec:periodic-edge}.

There are two exact ways of using the free-fermion structure for the interval.  One may reconstruct the marginal probabilities directly from the Gaussian/Walsh data restricted to $A_\ell$; we use that route later for an independent finite-size determination of the complete entropy.  For the analytic treatment, however, it is more useful to preserve the Bernoulli boundary representation of the full measurement law.  We first show that marginalizing the infinite exterior has an exact finite-graph representation obtained by replacing all crossing-edge information with independent edges to a single ghost vertex.  The ordinary Shannon chain rule then separates the interval entropy into an explicit independent-edge contribution and a conditional cycle contribution.  The former can be treated analytically, while the latter is organized by the same vertex M\"obius transform used for the periodic system.

The distinguished ghost vertex makes the conditional problem less homogeneous than on the ring: linked supports split into purely physical-line supports and supports containing the ghost.  The decomposition is nevertheless all-support from the outset.  At arbitrary ghost support the leading common-scale ambiguity is described by a monomer--dimer kernel, while the physical-line sectors inherit the linked kernels of the infinite-line Bernoulli graph.  Conditional on a compatible subtraction of lower-support collision and endpoint forests, these sectors define an all-orders hierarchy for the conditional logarithmic coefficient $\gamma_C$ and hence for $\gamma_1=\gamma_E-\gamma_C$.  At the first complete augmented support, the line coefficient $\alpha_3^{\linesys}$ and the ghost coefficient $\delta_2^{\partial}$ are both controlled analytically and their finite boundary-layer anomalies cancel, leaving the closed coefficient $\kappa_3^{+}=\alpha_3^{\linesys}+\delta_2^{\partial}=2I_2$.  Beginning at higher support, nested endpoint scales and finite lattice face anomalies make the renormalization less direct than in the periodic problem.  The direct interval distribution is used later as an independent finite-size check of the complete $\gamma_1$ and of this line/ghost bookkeeping.  No exact closed form, all-order renormalization theorem, or summability result for the total $\gamma_1$ is claimed here.

\subsection{Exact reduction of the subsystem to an augmented graph}
\label{sec:subsystem-ghost}

The interval distribution can be reconstructed directly from the exact Gaussian/Walsh data restricted to $A_\ell$, and we use that route later for the complete finite-size entropy in Sec.~\ref{sec:subsystem-direct}.  For the analytic decomposition of the Shannon entropy, however, it is more useful to preserve the Bernoulli boundary representation of the full measurement law from Sec.~\ref{sec:probability}.  In that language, taking the subsystem marginal becomes a question of how the exterior Bernoulli edges modify the parity observed on the vertices of $A_\ell$.

The key simplification is that, at each physical site $i$, the entire exterior enters only through the parity of the crossing edges incident on $i$.  We now show that these exterior parity variables can be represented exactly by independent edges to a single ghost vertex.

In the thermodynamic limit the physical edge law is
\begin{equation}
\beta_m=\frac12\left(1-\sqrt{1-\frac1{4m^2}}\right),
\qquad
q_m\equiv1-2\beta_m=\sqrt{1-\frac1{4m^2}}.
\label{eq:beta-infinite}
\end{equation}
Since $\beta_m=O(m^{-2})$, the total occupation probability of crossing edges incident on any fixed site is summable.  Split the infinite Bernoulli graph into edges wholly inside $A_\ell$, edges crossing the two cuts, and edges entirely outside.  Exterior--exterior edges never affect the parity observed inside the interval.  If $F_A$ denotes the internal graph, define for each $i\in A_\ell$ the exterior crossing parity
\begin{equation}
\xi_i=\bigoplus_{j\notin A_\ell}\eta_{ij}.
\end{equation}
Then the measured parity is
\begin{equation}
D_i=(\partial F_A)_i\oplus\xi_i.
\label{eq:sub-parity}
\end{equation}
A crossing edge has exactly one endpoint in $A_\ell$, so the sets of Bernoulli variables entering $\xi_i$ and $\xi_j$ are disjoint for $i\neq j$.  Hence the variables $\xi_i$ are independent.  A finite-range truncation establishing that these infinite parities are well defined and that their independence survives the infinite-volume limit is given in Appendix~\ref{app:ghost-details}.

Let
\begin{equation}
\lambda_{i,\ell}\equiv\PP(\xi_i=1)
\end{equation}
be the probability that the exterior contributes odd parity at site $i$.  Because the $\xi_i$ are independent, their joint law is reproduced exactly by independent Bernoulli variables $\zeta_i$ with $\PP(\zeta_i=1)=\lambda_{i,\ell}$.  Represent $\zeta_i$ by an edge from a new ghost vertex $0$ to $i$ and define
\begin{equation}
F^+=F_A\cup\{(0,i):\zeta_i=1\},
\qquad V_\ell^+=\{0\}\cup A_\ell.
\label{eq:Fplus}
\end{equation}

Thus the infinite exterior can be compressed, for purposes of the interval measurement law, to one independent parity variable per physical site; representing those variables as ghost edges preserves the marginal law while restoring the boundary form of a finite graph.

\begin{proposition}[Exact ghost reduction]
The marginal measurement vector on $A_\ell$ is equal in distribution to the physical components of the augmented boundary $\partial F^+$.  The ghost component is fixed by those physical components, so $D$ and $\partial F^+$ are in one-to-one correspondence and
\begin{equation}
H_1(A_\ell)=H(D)=H(\partial F^+)=H(F^+)-H(F^+\mid\partial F^+).
\label{eq:ghost-entropy-chain}
\end{equation}
\end{proposition}

Indeed, for every physical vertex,
\begin{equation}
(\partial F^+)_i=(\partial F_A)_i\oplus\zeta_i,
\end{equation}
so, jointly over all $i\in A_\ell$,
\begin{equation}
\bigl((\partial F^+)_i\bigr)_{i\in A_\ell}
\stackrel{d}{=}D.
\label{eq:ghost-boundary-physical}
\end{equation}
The ghost boundary bit contains no additional information.  Since the boundary of every finite graph has even total parity,
\begin{equation}
(\partial F^+)_0=\bigoplus_{i=1}^{\ell}(\partial F^+)_i
=\bigoplus_{i=1}^{\ell}D_i.
\label{eq:ghost-bit-determined}
\end{equation}
Thus $D$ and the full augmented boundary $\partial F^+$ determine one another, so $H(D)=H(\partial F^+)$.  Since $\partial F^+$ is a deterministic function of $F^+$, the Shannon chain rule gives Eq.~\eqref{eq:ghost-entropy-chain}.

The reduction is illustrated in \Cref{fig:ghost}.  It is an equality in distribution of the parity variables relevant to the subsystem measurement law, not an identification of the microscopic exterior edge configuration.  The same ghost representation extends to the real-$n>1$ interval problem~\cite{RajabpourBoundaryPhase2026}.

\begin{figure}[t]
\centering
\includegraphics[width=\textwidth]{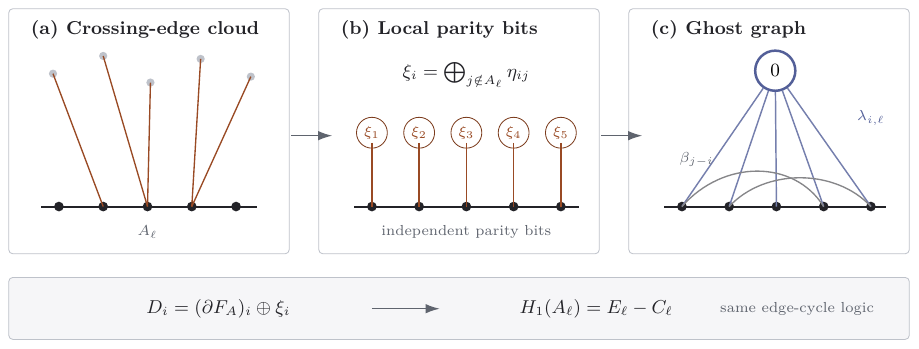}
\caption{Exact interval reduction.  (a) The exterior of $A_\ell$ contributes through the cloud of crossing Bernoulli edges.  (b) For each physical site $i$, those exterior edges enter the measured parity only through the single bit $\xi_i=\bigoplus_{j\notin A_\ell}\eta_{ij}$, with $\Pr(\xi_i=1)=\lambda_{i,\ell}$.  (c) The independent parity bits $\xi_i$ are represented exactly by edges from a single ghost vertex $0$ to the physical sites, while the original internal edges of probabilities $\beta_{j-i}$ remain present.  The physical components of $\partial F^+$ are the measured bits, and the ghost component is fixed by their total parity; the augmented graph therefore reproduces the complete interval measurement law.}
\label{fig:ghost}
\end{figure}

It remains to determine the ghost-edge probabilities $\lambda_{i,\ell}$ explicitly.  Since $\xi_i\in\{0,1\}$, the nontrivial $\Ztwo$ character gives
\begin{equation}
\EE(-1)^{\xi_i}=1-2\lambda_{i,\ell}.
\end{equation}
Parity addition becomes multiplication under this character,
\begin{equation}
(-1)^{\xi_i}
=\prod_{j\notin A_\ell}(-1)^{\eta_{ij}},
\qquad
\EE(-1)^{\eta_{ij}}=1-2\beta_{|i-j|}=q_{|i-j|}.
\end{equation}
Independence of the crossing edges therefore factorizes the expectation.  The exterior vertices to the left of $A_\ell$ occur at distances $m=i,i+1,\ldots$ from $i$, while those to the right occur at distances $m=\ell+1-i,\ell+2-i,\ldots$.  Hence
\begin{equation}
1-2\lambda_{i,\ell}
=\prod_{m=i}^{\infty}q_m
 \prod_{m=\ell+1-i}^{\infty}q_m.
\end{equation}
To evaluate the product, define the squared tail
\begin{equation}
P_n\equiv\prod_{m=n}^{\infty}q_m^2
=\prod_{m=n}^{\infty}\left(1-\frac1{4m^2}\right)
=\frac{\Gamma(n)^2}{\Gamma(n-\frac12)\Gamma(n+\frac12)}.
\end{equation}
Since $\prod_{m=n}^{\infty}q_m=\sqrt{P_n}$, the exact ghost-edge probability is
\begin{equation}
1-2\lambda_{i,\ell}=\sqrt{P_iP_{\ell+1-i}},\qquad
\lambda_{i,\ell}=\frac12\left[1-\sqrt{P_iP_{\ell+1-i}}\right].
\label{eq:lambda-main}
\end{equation}
The finite-truncation derivation of the infinite product and the Gamma-function tail identity are recorded in Appendix~\ref{app:ghost-details}.

With the ghost probabilities now explicit, the consequence of the exact ghost-reduction proposition can be written in the same edge--cycle form as for the full system.  Define the entropy of the independent augmented edges by
\begin{equation}
E_\ell\equiv H(F^+)
=\sum_{m=1}^{\ell-1}(\ell-m)h(\beta_m)
+\sum_{i=1}^{\ell}h(\lambda_{i,\ell}),
\label{eq:Eell}
\end{equation}
and the residual conditional uncertainty by
\begin{equation}
C_\ell\equiv H(F^+\mid\partial F^+)
=H(F_A,\zeta\mid\partial F_A\oplus\zeta).
\label{eq:Cell-augmented}
\end{equation}
Then the exact interval entropy is
\begin{equation}
H_1(A_\ell)=E_\ell-C_\ell.
\label{eq:sub-split}
\end{equation}
Thus the ghost construction does more than reproduce the marginal probabilities: it restores a finite independent-edge graph for which the Shannon entropy again separates into an explicit edge entropy and a conditional cycle entropy.  The next subsection determines the logarithmic part of $E_\ell$; the remaining term $C_\ell$ is organized by linked supports afterwards.

\subsection{Independent contribution and cancellation of squared logarithms}
\label{sec:subsystem-log}

The exact split \eqref{eq:sub-split} reduces the first part of the subsystem problem to the independent-edge entropy $E_\ell$.  The augmented graph contains two distinct classes of independent edges, so it is useful to separate
\begin{equation}
E_\ell=E_\ell^{\rm int}+E_\ell^{\rm ghost},
\qquad
E_\ell^{\rm int}=\sum_{m=1}^{\ell-1}(\ell-m)h(\beta_m),
\qquad
E_\ell^{\rm ghost}=\sum_{i=1}^{\ell}h(\lambda_{i,\ell}).
\label{eq:Eell-two-pieces}
\end{equation}
The purpose of this subsection is to determine the subleading growth of these two terms.  Neither piece separately has the single-logarithm form expected for the full interval entropy: each contains a $(\log\ell)^2$ contribution.  We show that these squared logarithms arise from two complementary boundary mechanisms and cancel exactly in their sum, leaving a single logarithm with an analytic coefficient.

\paragraph{Internal physical edges.}
Put $a=1/16$.  At large separation,
\begin{align}
h(\beta_m)
&=\frac{a}{m^2}\left[1-\log a+2\log m\right]
+O\!\left(\frac{\log m}{m^4}\right)\nonumber\\
&=\frac1{m^2}\left[\frac{\log m}{8}+\frac1{16}+\frac{\log2}{4}\right]
+O\!\left(\frac{\log m}{m^4}\right).
\label{eq:hbeta-tail}
\end{align}
An edge of length $m$ has $\ell-m$ possible placements inside the interval.  Relative to the infinite bulk edge entropy $s_{\edge}\ell=\ell\sum_{m\ge1}h(\beta_m)$, restricting the graph to the interval removes edge placements that would cross either cut.  Exactly,
\begin{equation}
E_\ell^{\rm int}
=s_{\edge}\ell-
\left[
\ell\sum_{m\ge\ell}h(\beta_m)
+\sum_{m<\ell}m h(\beta_m)
\right].
\label{eq:internal-split}
\end{equation}
Thus every subextensive term in $E_\ell^{\rm int}$ is a cut deficit rather than a bulk contribution.  The origin of the squared logarithm is already visible in the second term inside the brackets: Eq.~\eqref{eq:hbeta-tail} gives
\begin{equation}
-\sum_{m<\ell}m h(\beta_m)
=-2a\sum_{m<\ell}\frac{\log m}{m}+O(\log\ell)
=-a(\log\ell)^2+O(\log\ell).
\end{equation}
Keeping also the single-logarithmic pieces from this sum and from the missing tail $m\ge\ell$ gives
\begin{equation}
E_\ell^{\rm int}
=s_{\edge}\ell-a(\log\ell)^2-a(3-\log a)\log\ell+O(1).
\label{eq:internal-asympt}
\end{equation}
Thus the deficit of physical-edge placements created by the two interval boundaries produces a negative squared logarithm.

\paragraph{Ghost edges.}
The exterior degrees of freedom that were removed from the physical graph re-enter through the ghost probabilities.  Stirling expansion of the exact formula \eqref{eq:lambda-main}, with $N=\ell+1$, gives
\begin{equation}
\lambda_{i,\ell}
=a\left(\frac1i+\frac1{N-i}\right)
+O\!\left(i^{-2}+(N-i)^{-2}\right).
\label{eq:lambda-two-endpoints}
\end{equation}
Near either endpoint this has the slower tail $\lambda\sim a/R$.  Since a rare Bernoulli variable has
\begin{equation}
h(\lambda)=\lambda(1-\log\lambda)+O(\lambda^2),
\end{equation}
one endpoint contributes asymptotically $a\sum_{R<\ell}(\log R)/R$, and the two endpoints together therefore generate $+a(\log\ell)^2$.  Carrying the same expansion to single-logarithmic order gives
\begin{equation}
E_\ell^{\rm ghost}
=+a(\log\ell)^2+2a(1-\log a)\log\ell+O(1).
\label{eq:ghost-asympt}
\end{equation}
The detailed harmonic sums, the control of the remainder in \eqref{eq:lambda-two-endpoints}, and the corresponding endpoint summation are given in Appendix~\ref{app:sub-independent}.

\paragraph{Cancellation and the coefficient $\gamma_E$.}
The two squared logarithms have equal magnitude and opposite sign:
\begin{equation}
E_\ell^{\rm int}\supset-a(\log\ell)^2,
\qquad
E_\ell^{\rm ghost}\supset+a(\log\ell)^2.
\end{equation}
Their cancellation is the finite-interval counterpart of keeping the information carried by the exterior crossing edges: removing those edges from the physical interval creates the first term, while encoding their parity by ghost edges supplies the compensating second term.  Adding Eqs.~\eqref{eq:internal-asympt} and \eqref{eq:ghost-asympt} leaves
\begin{equation}
E_\ell=s_{\edge}\ell+
\left(\frac{\log2}{4}-\frac1{16}\right)\log\ell+O(1),
\label{eq:Eell-asympt}
\end{equation}
so the independent-edge logarithmic coefficient is
\begin{equation}
\gamma_E=\frac{\log2}{4}-\frac1{16}=0.1107867951399\ldots .
\label{eq:gammaE-main}
\end{equation}
The logarithmic accounting is therefore explicit: the internal physical edges contribute the cut-induced terms $-a(\log\ell)^2-a(3-\log a)\log\ell$, while the ghost edges contribute $+a(\log\ell)^2+2a(1-\log a)\log\ell$.  The squared logarithms cancel, whereas the two single-logarithmic pieces combine to \eqref{eq:gammaE-main}.

\paragraph{Relation to the real-$n>1$ cut picture.}
This decomposition should not be read as assigning separate ``interior'' and ``exterior'' origins to the universal interval logarithm.  The subleading part of $E_\ell^{\rm int}$ is itself induced by the cut, as Eq.~\eqref{eq:internal-split} makes explicit: it is the entropy deficit from physical-edge placements that are present in the infinite bulk but cannot lie wholly inside $A_\ell$.  The ghost term is the parity-compressed remnant of the crossing edges.  In the real-$n>1$ construction of Ref.~\cite{RajabpourBoundaryPhase2026}, the same boundary geometry is organized directly through the long-range interactions removed by the spatial cut.  At the Shannon point the ordinary chain rule instead reorganizes the problem into the independent-edge contribution $E_\ell$ and the conditional cycle contribution $C_\ell$.  Consequently, the $n=1$ decomposition is not a term-by-term continuation of the real-$n>1$ one, even though both are controlled by the same inverse-square Cauchy structure.

This relation now fixes the asymptotic bookkeeping.  Since both the total interval entropy \eqref{eq:sub-asympt} and the independent contribution \eqref{eq:Eell-asympt} contain only a single logarithm, the exact identity $H_1(A_\ell)=E_\ell-C_\ell$ implies that the conditional term may be written as
\begin{equation}
C_\ell=s_C\ell+\gamma_C\log\ell+O(1).
\end{equation}
Comparing coefficients gives
\begin{equation}
s_1=s_{\edge}-s_C,\qquad
\gamma_1=\gamma_E-\gamma_C.
\label{eq:subsystem-bookkeeping}
\end{equation}
Thus the independent-edge part of the subsystem logarithm is completely fixed analytically.  The remaining task is to understand the conditional cycle contribution $C_\ell$, to which we now apply the vertex-linked decomposition.

\subsection{Exact vertex-linked decomposition: line and ghost sectors}
\label{sec:subsystem-linked}

Having determined the independent-edge contribution $E_\ell$, the unresolved part of the subsystem logarithm is entirely contained in
\begin{equation}
C_\ell=H(F^+\mid\partial F^+).
\end{equation}
Unlike $E_\ell$, this conditional entropy is not additive over edges, so its logarithmic contribution cannot be assigned directly to individual physical or ghost edges.  The role of this subsection is therefore organizational: we resolve $C_\ell$ exactly into linked vertex supports by the same M\"obius transform used for the periodic problem.  The augmented interval introduces one essential distinction absent on the ring: a linked support either lies entirely on the physical interval or contains the ghost vertex.  This finite-size decomposition identifies the geometrically distinct sectors from which the coefficient $\gamma_C$ must be assembled; their endpoint asymptotics are analyzed only afterwards.

The augmented vertex set is
\begin{equation}
V_\ell^+=\{0\}\cup A_\ell.
\end{equation}
For every $U\subseteq V_\ell^+$ let
\begin{equation}
C_\ell^{\rm aug}(U)=H(F_U\mid\partial F_U)
\end{equation}
be the cycle entropy of the induced augmented graph.  When $U=S\subseteq A_\ell$ contains only physical vertices, write
\begin{equation}
C_\ell^{\rm aug}(S)=C_\ell^{\linesys}(S),
\end{equation}
and for the full physical interval set $C_\ell^{\linesys}\equiv C_\ell^{\linesys}(A_\ell)$.  For a fixed physical support $S$, this local line quantity uses the same infinite-line edge probabilities $\beta_{|i-j|}$ as in the bulk problem and therefore depends only on the relative physical positions; the explicit $\ell$-dependence of the summed line sectors below enters through the finite interval over which those supports are placed.  For a support containing the ghost write
\begin{equation}
C_\ell^{\rm aug}(\{0\}\cup S)=C_\ell^+(0\cup S).
\end{equation}

The ordinary vertex M\"obius cumulant on the augmented graph is
\begin{equation}
\widetilde\Psi_\ell(U)=
\sum_{R\subseteq U}(-1)^{|U|-|R|}C_\ell^{\rm aug}(R).
\label{eq:augmented-mobius}
\end{equation}
For a purely physical support this is the line cumulant
\begin{equation}
\Psi_\ell^{\linesys}(S)=\widetilde\Psi_\ell(S),
\qquad
A_{k,\ell}^{\linesys}=\sum_{|S|=k}\Psi_\ell^{\linesys}(S),
\qquad k\ge3.
\label{eq:line-sectors}
\end{equation}
These are the open-interval counterparts of the periodic linked sectors.  Their local kernels are inherited from the same physical Bernoulli law, but their sums are over an interval rather than a ring; consequently they possess endpoint asymptotics of their own and may contribute to the subsystem logarithm.

For a support $\{0\}\cup S$ containing the ghost, split the subsets $R\subseteq\{0\}\cup S$ according to whether they contain $0$.  Equation~\eqref{eq:augmented-mobius} then gives
\begin{align}
\widetilde\Psi_\ell(\{0\}\cup S)
&=\sum_{T\subseteq S}(-1)^{|S|-|T|}C_\ell^+(0\cup T)\nonumber\\
&\quad-\sum_{T\subseteq S}(-1)^{|S|-|T|}C_\ell^{\linesys}(T)\nonumber\\
&=\sum_{T\subseteq S}(-1)^{|S|-|T|}
\left[C_\ell^+(0\cup T)-C_\ell^{\linesys}(T)\right].
\label{eq:ghost-cumulant}
\end{align}
Thus the apparent ``ghost subtraction'' is not an additional rule: it is simply the ordinary M\"obius cumulant of an augmented support containing the distinguished vertex.  We denote it by
\begin{equation}
\Psi_\ell^{\partial}(S)\equiv\widetilde\Psi_\ell(\{0\}\cup S),
\qquad
B_{k,\ell}=\sum_{|S|=k}\Psi_\ell^{\partial}(S),
\qquad k\ge2.
\label{eq:Bk}
\end{equation}
The index $k$ counts physical vertices only.  Hence $B_{k,\ell}$ is a linked sector on $k+1$ total augmented vertices, and the first possible cycle is $B_{2,\ell}$, corresponding to the three-vertex support $\{0,i,j\}$.

M\"obius inversion on the full augmented graph now gives the exact finite-size decomposition
\begin{equation}
C_\ell=C_\ell^{\linesys}+C_\ell^\partial,
\qquad
C_\ell^{\linesys}=\sum_{k\ge3}A_{k,\ell}^{\linesys},
\qquad
C_\ell^\partial\equiv\sum_{k\ge2}B_{k,\ell},
\label{eq:Cell-line-ghost}
\end{equation}
or, equivalently,
\begin{equation}
C_\ell=\sum_{k\ge3}A_{k,\ell}^{\linesys}
+\sum_{k\ge2}B_{k,\ell}.
\label{eq:Cell-full-linked}
\end{equation}
This is the exact structural split needed for the conditional entropy: the first sum contains linked supports entirely inside the physical interval, while the second contains every linked support through which the exterior parity enters via the ghost.  It is a decomposition of the full $C_\ell$, not an approximation scheme and not yet an evaluation of $\gamma_C$.

For comparison with the periodic support counting, define
\begin{equation}
\widetilde A_{q,\ell}^{+}=\sum_{\substack{U\subseteq V_\ell^+\\|U|=q}}\widetilde\Psi_\ell(U).
\end{equation}
Then, for $3\le q\le\ell+1$,
\begin{equation}
\widetilde A_{q,\ell}^{+}=A_{q,\ell}^{\linesys}+B_{q-1,\ell},
\qquad A_{\ell+1,\ell}^{\linesys}\equiv0.
\label{eq:augmented-sector-split}
\end{equation}
Thus a periodic $q$-vertex sector is homogeneous in vertex type, whereas the augmented interval sector of the same total size splits into a line piece and a ghost-containing piece; only their sum reconstructs the corresponding augmented M\"obius sector.  Combining this with the independent-edge split gives the compact exact hierarchy
\begin{equation}
H_1(A_\ell)
=\underbrace{\bigl(E_\ell^{\rm int}+E_\ell^{\rm ghost}\bigr)}_{E_\ell}
-\underbrace{\sum_{q=3}^{\ell+1}\widetilde A_{q,\ell}^{+}}_{C_\ell}.
\label{eq:interval-nested-hierarchy}
\end{equation}
The first bracket contains the explicit cancellation of the squared logarithms, while the second is the complete augmented linked hierarchy whose renormalized single-log coefficients are organized below.  Appendix~\ref{app:distinguished-mobius} records the distinguished-vertex M\"obius identity in abstract form and its interval implementation, using the finite-size boundary transform developed in Appendix~\ref{app:mobius}.

For the analysis below it is useful to isolate the first ghost sector,
\begin{equation}
C_\ell^{\partial,\ge3}\equiv\sum_{k\ge3}B_{k,\ell}
=C_\ell^\partial-B_{2,\ell},
\end{equation}
so that
\begin{equation}
C_\ell=C_\ell^{\linesys}+B_{2,\ell}+C_\ell^{\partial,\ge3}.
\label{eq:Cell-three-piece}
\end{equation}
Equation~\eqref{eq:Cell-three-piece} is the bookkeeping identity used in the remainder of the interval analysis.  The three terms have different geometric origins: $C_\ell^{\linesys}$ is linked uncertainty supported entirely inside the physical interval, $B_{2,\ell}$ is the first linked ambiguity involving the exterior parity through the ghost, and $C_\ell^{\partial,\ge3}$ contains all higher ghost-linked supports.  The next subsection evaluates $B_{2,\ell}$ because it is the first analytically closed ghost sector; it should not be interpreted as an approximation to the complete $C_\ell$ or to $\gamma_C$.

\subsection{The first ghost-linked logarithm}
\label{sec:subsystem-B2}

The decomposition in the previous subsection isolates $B_{2,\ell}$ as the first ghost-containing piece of the conditional entropy.  We now evaluate its logarithmic growth completely.  This sector is the simplest place where the boundary mechanism can be seen without the nested endpoint subdivergences that appear at higher support: a linked support contains the ghost and exactly two physical vertices, so its induced graph is the triangle $\{0,i,j\}$.

Every proper induced subgraph of a triangle is a forest and therefore has zero cycle entropy.  The M\"obius cumulant consequently reduces exactly to the conditional entropy of that triangle,
\begin{equation}
\Psi_\ell^{\partial}(\{i,j\})
=C_\triangle(\lambda_{i,\ell},\lambda_{j,\ell},\beta_{j-i}),
\qquad
B_{2,\ell}=\sum_{1\le i<j\le\ell}C_\triangle.
\label{eq:B2-triangle}
\end{equation}
The three possible edges are the two ghost edges $(0,i)$ and $(0,j)$, with probabilities $\lambda_{i,\ell}$ and $\lambda_{j,\ell}$, and the physical edge $(i,j)$, with probability $\beta_{j-i}$.  Since the triangle has a single cycle, fixing its boundary can leave two complementary edge configurations.  The quantity $C_\triangle$ is precisely the Shannon ambiguity inside those boundary fibers.  The exact finite-lattice expression contains four complementary fiber pairs and is recorded in Appendix~\ref{app:triangle-I2}.  The local triangle ambiguity is illustrated in \Cref{fig:B2-two-scale}; the two asymptotic mechanisms responsible for the logarithm are analyzed separately below.

\begin{figure}[t]
\centering
\includegraphics[width=0.72\textwidth]{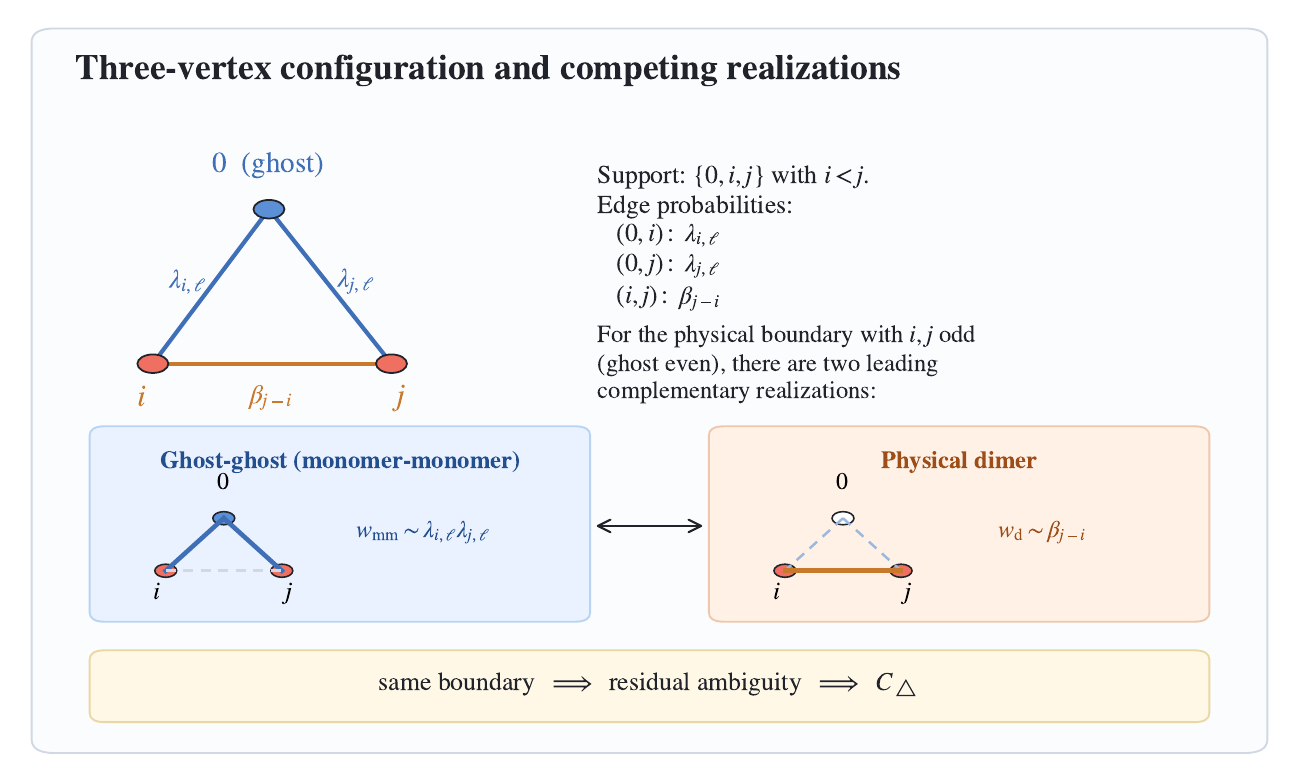}
\caption{Three-vertex ambiguity underlying the first ghost-linked sector.  On the support $\{0,i,j\}$, the unique triangle cycle allows two complementary realizations of the same physical boundary with $i$ and $j$ odd and the ghost even.  The leading competing configurations are the two ghost edges, with weight proportional to $\lambda_{i,\ell}\lambda_{j,\ell}$, and the single physical edge, with weight proportional to $\beta_{j-i}$.  The conditional triangle entropy $C_\triangle$ measures the residual Shannon ambiguity between these realizations.  The common-scale and fixed-separation regions that generate the logarithm in $B_{2,\ell}$ are analyzed in the text.}
\label{fig:B2-two-scale}
\end{figure}

These two asymptotic mechanisms must be kept distinct.  The common-scale sector gives the continuum period, whereas the fixed-separation sector is a nonuniform lattice boundary layer of that continuum limit.  Both produce a harmonic outer-scale sum and therefore contribute to the same coefficient of $\log\ell$.

\paragraph{Common-scale endpoint region.}
Consider the left endpoint and write
\begin{equation}
i=tR,\qquad j=R,\qquad 0<t<1,
\qquad 1\ll R\ll\ell.
\end{equation}
At this scale
\begin{equation}
\lambda_{i,\ell}=\frac{1}{16tR}+o(R^{-1}),\qquad
\lambda_{j,\ell}=\frac{1}{16R}+o(R^{-1}),\qquad
\beta_{j-i}=\frac{1}{16(1-t)^2R^2}+o(R^{-2}).
\end{equation}
For the physical boundary in which both $i$ and $j$ are odd, the two minimally suppressed realizations are therefore the two ghost edges and the single physical edge.  Removing their common factor $R^{-2}$ leaves the scaled weights
\begin{equation}
u(t)=\frac1{256t},\qquad
v(t)=\frac1{16(1-t)^2}.
\end{equation}
Their conditional Shannon ambiguity gives
\begin{align}
C_\triangle(tR,R)
&=\frac1{R^2}\Gamma_2^\partial(t)+o(R^{-2}),\nonumber\\
\Gamma_2^\partial(t)
&=(u+v)\log(u+v)-u\log u-v\log v.
\label{eq:B2-macro-scaling}
\end{align}
For two physical vertices the line-only subfiber contains only the single dimer and hence has zero conditional entropy, so no additional line subtraction appears at this order.

The logarithm now follows from scale counting.  To keep the common-scale approximation separate from the nonuniform boundary layers, fix $\epsilon>0$.  At fixed outer scale $R$ there are $O(R)$ allowed positions of the inner point, and the Riemann sum over $t=i/R$ gives
\begin{equation}
\sum_{\epsilon R\le i\le(1-\epsilon)R}C_\triangle(i,R)
=\frac1R\int_\epsilon^{1-\epsilon}\Gamma_2^\partial(t)\dd t
+o_\epsilon(R^{-1}),
\qquad
I_2=\int_0^1\Gamma_2^\partial(t)\dd t.
\label{eq:B2-macro-sum}
\end{equation}
Taking $R\to\infty$ first and then $\epsilon\to0$ defines the common-scale coefficient $I_2/R$; the microscopic near-diagonal strip is not included in this limit and is matched separately below.  Thus the remaining scale sum $\sum_RR^{-1}$ produces $I_2\log\ell$ from the left endpoint.  Reflection gives the same contribution from the right endpoint, so the common-scale regions contribute $2I_2\log\ell$ in total.  The elementary evaluation of the period is given in Appendix~\ref{app:triangle-I2} and yields
\begin{equation}
I_2=\frac1{256}\left[
\frac{\pi^2}{2}+2\log^2(2+\sqrt3)+32\log2-8\sqrt3\log(2+\sqrt3)
\right].
\label{eq:I2}
\end{equation}

\paragraph{Microscopic fixed-separation boundary layer.}
The common-scale parametrization is not uniform as $t\to1$.  Indeed, $t$ fixed assumes that the physical separation $r=j-i=(1-t)R$ grows with $R$; it therefore misses the near-diagonal strip in which $r$ remains microscopic while the pair is translated far from an endpoint.  Put
\begin{equation}
i=R,\qquad j=R+r,\qquad r=O(1),\qquad R\to\infty.
\end{equation}
Then
\begin{equation}
\lambda_{R,\ell}=\lambda_{R+r,\ell}=\frac1{16R}+O(R^{-2}),
\qquad
\beta_{j-i}=\beta_r,
\end{equation}
so the physical edge is no longer of order $R^{-2}$.  The exact triangle-fiber decomposition gives instead
\begin{equation}
C_\triangle
=\frac{h(\beta_r)}{8R}+O(R^{-2}\log R).
\label{eq:triangle-micro-tail}
\end{equation}
Each fixed separation therefore contributes $h(\beta_r)\log\ell/8$ per endpoint.  Summing over $r$ and using $s_{\edge}=\sum_{r\ge1}h(\beta_r)$ gives $s_{\edge}/8$ from each endpoint, or the total boundary-layer contribution $s_{\edge}/4$.  This term is not an independent double count of the continuum period: the two descriptions are matched by an intermediate separation cutoff, and the resulting coefficient is cutoff independent.  Appendix~\ref{app:triangle-I2} gives the explicit four-fiber derivation and the matching argument.

\paragraph{Result and interpretation.}
Combining the two endpoint periods with the microscopic boundary layer gives the structural formula
\begin{equation}
\delta_2^{\partial}=2I_2+\frac{s_{\edge}}4.
\label{eq:delta2-structural}
\end{equation}
Equivalently,
\begin{equation}
\begin{aligned}
 \delta_2^{\partial}={}&\frac{s_{\edge}}4+
 \frac1{128}\left[
 \frac{\pi^2}{2}+2\log^2(2+\sqrt3)+32\log2
 -8\sqrt3\log(2+\sqrt3)\right]\\
={}&0.2254584779146003\ldots .
\end{aligned}
\label{eq:delta2-boundary}
\end{equation}
Hence
\begin{equation}
B_{2,\ell}=\delta_2^{\partial}\log\ell+O(1).
\label{eq:B2-asymptotic}
\end{equation}
The decomposition in Eq.~\eqref{eq:delta2-structural} is the main structural result of this subsection: the first ghost-linked logarithm receives one continuum contribution from configurations whose two physical vertices scale together away from an endpoint and one lattice boundary-layer contribution from configurations that remain microscopically separated.

The coefficient $\delta_2^{\partial}$ is much larger than the logarithmic coefficient of the complete conditional entropy,
\begin{equation}
\gamma_C=\gamma_E-\gamma_1\simeq0.05077.
\end{equation}
This is not a discrepancy.  Equation~\eqref{eq:Cell-three-piece} shows that $B_{2,\ell}$ is only one linked component of $C_\ell$, and substantial cancellations occur against both the open-line sector and higher ghost sectors.  As a finite-size diagnostic, after subtracting the known bulk term $s_{\cyc}\ell$ the two-scale slope of the complete $C_\ell$ is about $0.0521$ between $\ell=8$ and $16$, while the same diagnostic decomposes approximately as
\begin{equation}
0.0521\simeq -0.2880+0.2405+0.0996,
\label{eq:conditional-slope-cancellation}
\end{equation}
coming respectively from $C_\ell^{\linesys}$, $B_{2,\ell}$, and $C_\ell^{\partial,\ge3}$.  These are finite-size effective slopes rather than separate universal coefficients---the higher ghost sectors contain nested logarithms---but they make the cancellation mechanism visible.  Appendix~\ref{app:finite-data} gives the corresponding data.

\subsection{First complete augmented linked coefficient}
\label{sec:subsystem-kappa3}

The first ghost sector $B_{2,\ell}$ is only one part of the complete three-vertex augmented sector.  Equation~\eqref{eq:augmented-sector-split} gives
\begin{equation}
\widetilde A_{3,\ell}^{+}=A_{3,\ell}^{\linesys}+B_{2,\ell}.
\label{eq:q3-augmented-split}
\end{equation}
The physical-line term is also analytically accessible.  For three physical vertices every proper induced subgraph is a forest, so the three-vertex M\"obius cumulant is simply the triangle conditional entropy.  Resolving the sum by its outer span $R$ gives the exact identity
\begin{align}
A_{3,\ell}^{\linesys}
&=\sum_{R=2}^{\ell-1}(\ell-R)\Phi_3(R),\label{eq:A3-span}\\
\Phi_3(R)
&=\sum_{r=1}^{R-1}
C_\triangle(\beta_r,\beta_{R-r},\beta_R).
\label{eq:Phi3-span}
\end{align}
For fixed $r$ and $R\to\infty$, the two long physical edges are both $1/(16R^2)+O(R^{-3})$, while the short edge retains probability $\beta_r$.  The exact triangle-fiber formula of Appendix~\ref{app:triangle-I2} therefore gives
\begin{equation}
C_\triangle(\beta_r,\beta_{R-r},\beta_R)
=\frac{h(\beta_r)}{8R^2}+O(R^{-3}\log R).
\label{eq:line-triangle-tail}
\end{equation}
The two collision ends $r=O(1)$ and $R-r=O(1)$ contribute equally, whereas the fully macroscopic interior is subleading.  Matching the fixed-separation and common-scale regions as in Appendix~\ref{app:triangle-I2} yields
\begin{equation}
\Phi_3(R)=\frac{s_{\edge}}{4R^2}+o(R^{-2}).
\label{eq:Phi3-tail}
\end{equation}
Since $\sum_{R<\ell}(\ell-R)R^{-2}$ has logarithmic part $-\log\ell$, the line coefficient is
\begin{equation}
\boxed{\alpha_3^{\linesys}=-\frac{s_{\edge}}4.}
\label{eq:alpha3-line}
\end{equation}
Combining this with Eq.~\eqref{eq:delta2-structural} shows that the finite boundary-layer anomaly cancels inside the complete augmented support:
\begin{equation}
\boxed{
\kappa_3^{+}=\alpha_3^{\linesys}+\delta_2^{\partial}=2I_2.}
\label{eq:kappa3-main}
\end{equation}
Thus the first complete augmented linked logarithmic coefficient is known in closed form,
\begin{equation}
\boxed{
\kappa_3^{+}=\frac1{128}\left[
\frac{\pi^2}{2}+2\log^2(2+\sqrt3)+32\log2
-8\sqrt3\log(2+\sqrt3)\right]
=0.0963747208733\ldots .}
\label{eq:kappa3-closed}
\end{equation}
The cancellation is internal to one exact M\"obius sector: the $+s_{\edge}/4$ term in the ghost contribution is precisely compensated by the $-s_{\edge}/4$ endpoint anomaly of the physical-line contribution.  This makes $\widetilde A_{3,\ell}^{+}$, rather than $B_{2,\ell}$ alone, the natural first object in the interval linked hierarchy.

\subsection{Higher boundary sectors and nested endpoint scales}
\label{sec:subsystem-higher-boundary}

The solved $B_{2,\ell}$ sector suggests a general boundary mechanism: at a common endpoint scale, ghost attachments and physical pairings can contribute at the same total scaling degree.  We now identify the leading boundary fiber for arbitrary $k$ and explain why, beginning at $k=3$, nested endpoint and collision scales obstruct a direct extraction of a single logarithmic coefficient.  Consider a ghost-containing support $\{0\}\cup S$ with $k$ physical vertices in a common left-endpoint scaling region
\begin{equation}
x_i=Rt_i,\qquad 0<t_1<\cdots<t_k=1,
\qquad 1\ll R\ll\ell.
\end{equation}
The edge probabilities have the asymptotic form
\begin{equation}
\lambda_{x_i,\ell}=\frac{1}{R}\,m_i+o(R^{-1}),
\qquad
\beta_{x_j-x_i}=\frac{1}{R^2}\,a_{ij}+o(R^{-2}),
\end{equation}
with
\begin{equation}
m_i=\frac{1}{16t_i},
\qquad
a_{ij}=\frac{1}{16(t_i-t_j)^2}.
\label{eq:boundary-monomer-dimer-weights}
\end{equation}
The appearance of monomer--dimer coverings can be derived directly from the exact boundary fibers, rather than taken as a continuum ansatz.  Consider any augmented edge configuration on the support $\{0\}\cup S$ and let $n_{\rm g}$ be the number of occupied ghost edges and $n_{\rm p}$ the number of occupied physical edges.  In the common-scale regime its probability is suppressed by
\begin{equation}
R^{-(n_{\rm g}+2n_{\rm p})},
\end{equation}
because a ghost edge costs one power of $R^{-1}$ whereas a physical edge costs two.  The same exponent has a simple graph-theoretic meaning:
\begin{equation}
n_{\rm g}+2n_{\rm p}=\sum_{i\in S}\deg_F(i),
\label{eq:boundary-degree-exponent}
\end{equation}
namely it is the total degree carried by the physical vertices.

This identity singles out the primitive boundary fiber.  Let $D\subseteq S$ be the set of physical vertices with odd degree and write $r=|D|$.  If an edge configuration genuinely uses the complete physical support $S$, every vertex in $D$ has positive odd degree and therefore degree at least one, while every vertex in $S\setminus D$ has positive even degree and therefore degree at least two.  Hence every full-support realization of the physical boundary $D$ obeys
\begin{equation}
n_{\rm g}+2n_{\rm p}=\sum_{i\in S}\deg_F(i)
\ge r+2(k-r)=2k-r.
\label{eq:boundary-full-support-bound}
\end{equation}
For the all-odd physical boundary $D=S$ one has $r=k$, so the lower bound is $k$.  Equality is possible only when every physical vertex has degree exactly one.  Each vertex must then either be paired once with another physical vertex or be attached once to the ghost.  The minimally suppressed full-support realizations are therefore in one-to-one correspondence with \emph{monomer--dimer coverings} of the $k$ physical vertices: physical edges are dimers and ghost attachments are monomers.

The same degree count also shows why no other physical boundary competes at the primitive logarithmic order.  If $D\subsetneq S$, then $r\le k-1$ and Eq.~\eqref{eq:boundary-full-support-bound} gives $n_{\rm g}+2n_{\rm p}\ge k+1$.  Thus every genuinely full-support realization of a proper physical boundary is suppressed by at least one additional power of $R$.  Configurations that evade this bound by leaving one or more physical vertices unused are precisely lower-support configurations; their leading contribution is already present on a proper subset and is removed by the vertex M\"obius support filter.  Consequently the all-odd fiber is the unique common-scale boundary fiber capable of producing a primitive $R^{-k}$ full-support term.

Let $M$ be a partial matching of the physical vertices and let $U(M)$ be its unmatched set.  The corresponding augmented covering is
\begin{equation}
F_M=M\cup\{(0,i):i\in U(M)\},
\end{equation}
so that
\begin{equation}
n_{\rm p}=|M|,\qquad n_{\rm g}=|U(M)|,
\qquad n_{\rm g}+2n_{\rm p}=k.
\end{equation}
Its scaled weight is
\begin{equation}
w_M=\prod_{(ij)\in M}a_{ij}\prod_{i\in U(M)}m_i.
\label{eq:boundary-covering-weight}
\end{equation}
All such coverings have the same overall scale,
\begin{equation}
\Pr(F=F_M)=R^{-k}w_M+o(R^{-k}),
\end{equation}
whereas full-support configurations in every proper physical-boundary fiber are $o(R^{-k})$.  If $\mathfrak M_k^\partial$ denotes the set of all partial matchings, then
\begin{equation}
Z_k^{\rm aug}=\sum_{M\in\mathfrak M_k^\partial}w_M,
\end{equation}
and conditioning on the all-odd physical boundary gives
\begin{equation}
\Pr(F=F_M\mid D_{\rm phys}=S)
=\frac{w_M}{Z_k^{\rm aug}}+o(1).
\end{equation}
The common factor $R^{-k}$ cancels in the conditional probabilities.  Multiplying the resulting Shannon entropy by the probability of the all-odd boundary therefore gives
\begin{equation}
\Pr(D_{\rm phys}=S)H(F\mid D_{\rm phys}=S)
=R^{-k}\Gamma_k^{\rm aug}+o(R^{-k}),
\end{equation}
with
\begin{equation}
\Gamma_k^{\rm aug}
=Z_k^{\rm aug}\log Z_k^{\rm aug}
-\sum_{M\in\mathfrak M_k^\partial}w_M\log w_M.
\label{eq:boundary-raw-md}
\end{equation}
Thus the monomer--dimer kernel is the leading conditional Shannon ambiguity of a specific rare boundary fiber of the exact augmented Bernoulli graph, exactly as the perfect-matching kernel arises in the periodic problem.

There is one additional feature specific to the \emph{ghost-linked} M\"obius coefficient.  Equation~\eqref{eq:ghost-cumulant} subtracts the line entropy on the same physical support.  For odd $k$ every monomer--dimer covering necessarily contains at least one monomer, so the all-odd fiber has no line-only counterpart at the same scaling order.  For even $k$, however, the zero-monomer subfamily consists of perfect matchings of the physical vertices and is already present in $C_\ell^{\linesys}(S)$.  Denote this subfamily by $\mathfrak M_k^0$ and define
\begin{equation}
Z_k^{\rm line}=\sum_{M\in\mathfrak M_k^0}w_M,
\qquad
\Gamma_k^{\rm line}
=Z_k^{\rm line}\log Z_k^{\rm line}
-\sum_{M\in\mathfrak M_k^0}w_M\log w_M,
\end{equation}
with $\Gamma_k^{\rm line}=0$ for odd $k$.  The leading common-scale kernel belonging to the ghost-linked cumulant is consequently
\begin{equation}
\Gamma_k^\partial=
\begin{cases}
\Gamma_k^{\rm aug},&k\ \text{odd},\\[2mm]
\Gamma_k^{\rm aug}-\Gamma_k^{\rm line},&k\ \text{even}.
\end{cases}
\label{eq:monomer-dimer}
\end{equation}
The subtraction is invisible at $k=2$, because the line-only subfamily contains a single physical dimer and hence has zero entropy.  Equation~\eqref{eq:monomer-dimer} therefore reduces exactly to the two-weight kernel $\Gamma_2^\partial$ derived in \Cref{sec:subsystem-B2}.  For higher even $k$ the subtraction is essential: it removes the perfect-matching ambiguity that belongs to the physical-line sector rather than to new ghost-linked information.

We can now state the discrete-to-continuum bridge at fixed support.  Let $x_i/R\to t_i$ with $0<t_1<\cdots<t_k=1$, and keep the shape variables a fixed positive distance from all collision and endpoint faces while taking $R\to\infty$ with $R/\ell\to0$.  The support-filter argument above and the asymptotic edge weights then give
\begin{equation}
\Psi_\ell^\partial(\{x_1,\ldots,x_k\})
=R^{-k}\Gamma_k^\partial(t_1,\ldots,t_k)+o(R^{-k}).
\label{eq:boundary-Psi-md}
\end{equation}
Equation~\eqref{eq:boundary-Psi-md} is the subsystem counterpart of Eq.~\eqref{eq:Psi-macro-matching}: the M\"obius transform selects complete physical support, the degree bound selects the all-odd boundary as the unique marginal fiber, and its minimal configurations give the monomer--dimer entropy kernel.

The origin of the logarithm is then a direct Riemann-sum count.  Fix the outermost left-endpoint vertex at $x_k=R$.  There are $k-1$ internal ordered positions, so a compact region of shape space contains $O(R^{k-1})$ lattice supports.  Combining this measure factor with Eq.~\eqref{eq:boundary-Psi-md} gives
\begin{equation}
R^{k-1}R^{-k}=R^{-1}.
\label{eq:boundary-superficial-degree}
\end{equation}
To make the statement precise away from singular faces, let $\mathcal D_{k-1}(\epsilon)$ be the ordered shape simplex with $t_1$, all successive separations, and $1-t_{k-1}$ bounded below by $\epsilon>0$, and define
\begin{equation}
\mathcal I_k^\partial(\epsilon)
=\int_{\mathcal D_{k-1}(\epsilon)}
\Gamma_k^\partial(t_1,\ldots,t_{k-1},1)\,\dd^{k-1}t.
\label{eq:boundary-shape-integral}
\end{equation}
Then the common-scale contribution from the two interval endpoints has the asymptotic form
\begin{equation}
B_{k,\ell}^{\rm common}(\epsilon)
=2\,\mathcal I_k^\partial(\epsilon)\log\ell+O_\epsilon(1).
\label{eq:boundary-common-log}
\end{equation}
The factor of two is the contribution of the left and right endpoints.  Thus every primitive boundary support order is superficially marginal, in sharp contrast with the periodic problem where only even primitive supports contribute at order $L^0$.

The exact linked decomposition therefore already has the structure of an all-orders expansion for the interval cycle logarithm.  To state it explicitly, let $\alpha_k^{\linesys}$ denote the renormalized single-log coefficient generated by the physical-line sector $A_{k,\ell}^{\linesys}$ after its lower-support collision and endpoint subtractions.  For the ghost sectors, conditional on closure of the compatible boundary-forest subtraction, define
\begin{equation}
\delta_k^{\partial}
=2\,\FP_{\partial}
\int_{0<t_1<\cdots<t_{k-1}<1}
R_{\partial}\Gamma_k^{\partial}(t_1,\ldots,t_{k-1},1)\,\dd^{k-1}t,
\label{eq:boundary-delta-allorder}
\end{equation}
where $R_{\partial}$ subtracts nested endpoint/collision subclusters and $\FP_{\partial}$ includes the finite lattice face anomalies required by matching to the discrete problem.  It is then natural to combine the line and ghost coefficients belonging to one complete augmented support,
\begin{equation}
\kappa_q^{+}\equiv\alpha_q^{\linesys}+\delta_{q-1}^{\partial},
\qquad q\ge3.
\label{eq:kappa-q-definition}
\end{equation}
The superscript $+$ distinguishes these augmented linked coefficients from the nonuniversal interval constant $\kappa_1$ in Eq.~\eqref{eq:sub-asympt}.  Under the corresponding all-order closure and summability assumptions, Eqs.~\eqref{eq:Cell-full-linked} and~\eqref{eq:subsystem-bookkeeping} become
\begin{align}
\gamma_C
&=\sum_{k\ge3}\alpha_k^{\linesys}
 +\sum_{k\ge2}\delta_k^{\partial}
 =\sum_{q\ge3}\kappa_q^{+},
\label{eq:interval-cycle-allorder}\\
\gamma_1
&=\frac{\log2}{4}-\frac1{16}
 -\sum_{q\ge3}\kappa_q^{+}.
\label{eq:interval-gamma-allorder}
\end{align}
The second form groups the line and ghost contributions into the complete augmented M\"obius sector at fixed total support $q$, as in Eq.~\eqref{eq:kappa-q-definition}; in particular $\kappa_3^{+}$ is the closed coefficient in Eq.~\eqref{eq:kappa3-main}.  This is the interval counterpart of the periodic linked series.  Its complication is not the absence of an all-support organization, but the fact that every boundary order is marginal and the renormalized primitive coefficient must be disentangled from nested endpoint scales before the series can be evaluated or controlled.

Equations~\eqref{eq:boundary-Psi-md}--\eqref{eq:boundary-common-log} hold only with the collision and nested-endpoint faces kept at finite distance.  Removing the shape cutoff is precisely where the nonuniform regions enter.  For $k=2$ the continuum shape period is integrable, but the fixed-separation lattice face still produces the finite boundary-layer anomaly $s_{\edge}/4$ derived above.  For $k\ge3$ the shape integral itself contains endpoint and collision subdivergences, so the coefficient of the physical single logarithm cannot be obtained by simply sending $\epsilon\to0$ in Eq.~\eqref{eq:boundary-common-log}; the proper subscales must first be removed by the boundary forest subtraction, together with their finite lattice face anomalies.

The first higher case makes the construction concrete.  For $k=3$ there is no line-only perfect matching, so $\Gamma_3^\partial=\Gamma_3^{\rm aug}$.  With $0<t_1<t_2<t_3=1$, the four coverings are the three-ghost-monomer state and the three states containing one physical dimer and one ghost monomer.  Their weights are
\begin{align}
w_0&=m_1m_2m_3=\frac1{4096t_1t_2},&
 w_{12}&=a_{12}m_3=\frac1{256(t_2-t_1)^2},\nonumber\\
w_{13}&=a_{13}m_2=\frac1{256t_2(1-t_1)^2},&
 w_{23}&=a_{23}m_1=\frac1{256t_1(1-t_2)^2}.
\end{align}
Hence
\begin{equation}
Z_3=w_0+w_{12}+w_{13}+w_{23},
\qquad
\Gamma_3^\partial=Z_3\log Z_3-\sum_{\alpha}w_\alpha\log w_\alpha.
\label{eq:Gamma3-boundary-main}
\end{equation}
The common-scale power counting by itself would suggest a single logarithm.  The obstruction is that the shape integral of $\Gamma_3^\partial$ has proper endpoint subscales.  For example, if the first point approaches the endpoint while $t_2=t$ remains fixed,
\begin{equation}
\Gamma_3^\partial(t_1,t,1)
=\frac{1}{16t_1}\Gamma_2^\partial(t)+O(|\log t_1|),
\qquad t_1\to0^+.
\label{eq:k3-factorization-main}
\end{equation}
The factorization has a simple covering interpretation: the first vertex is forced into its dominant ghost-monomer factor $1/(16t_1)$, while the remaining two vertices retain exactly the two-covering ambiguity of the solved $k=2$ problem.  Integrating the inner scale therefore creates an additional logarithm which multiplies the outer $R^{-1}$ sum.  A common-scale collapse of two points near the endpoint gives another representation of the same subdivergence, and the overlap of these descriptions must be removed by inclusion--exclusion on compatible endpoint/collision forests.

This is why the naive integral of $\Gamma_k^\partial$ is not, for $k\ge3$, an invariant single-log coefficient.  Raw higher-support sectors can contain $(\log\ell)^2$ and higher nested logarithms before proper subscales are subtracted.  Appendix~\ref{app:boundary-forest} gives the detailed $k=3$ collision and endpoint factorizations and formulates the corresponding forest prescription.  Appendix~\ref{app:k3-setup} then turns the same $k=3$ example into an explicit two-variable integral on the unit square, displays an analytic regulator, and shows how the two lower-order $\Gamma_2^\partial$ residues appear as boundary poles before the overlap subtraction.  At present Eqs.~\eqref{eq:boundary-delta-allorder}--\eqref{eq:interval-gamma-allorder} should therefore be read as a controlled all-orders renormalization organization rather than as a proved expansion: closing the overlap subtraction, matching the finite lattice face anomalies at arbitrary support, and controlling the support sum remain open.  We quote the analytically controlled $k=2$ coefficient, use the higher sectors to expose the recursive structure and cancellations of the conditional entropy, and determine the complete subsystem logarithm independently from the exact interval distribution in the next subsection.

\subsection{Direct reconstruction of the complete subsystem logarithm}
\label{sec:subsystem-direct}

The complete $\gamma_1$ can be obtained without the linked expansion.  The infinite-line Walsh coefficient for a finite set $S$ is
\begin{equation}
d_\infty(S)=\left(\frac2\pi\right)^{|S|}
\prod_{i<j\in S}\frac{4(i-j)^2}{4(i-j)^2-1}.
\label{eq:d-infty}
\end{equation}
Encode a spin configuration by $x\in\{0,1\}^\ell$ through $\sigma_i=(-1)^{x_i}$ and write $S\cdot x=\sum_{i\in S}x_i\pmod 2$.  The interval distribution follows from a $2^\ell$-component Walsh--Hadamard transform,
\begin{equation}
p_A(x)=2^{-\ell}\sum_{S\subseteq A}(-1)^{S\cdot x}d_\infty(S).
\label{eq:FWHT-paper}
\end{equation}
The reconstruction uses only the exact infinite-line Walsh coefficients in Eq.~\eqref{eq:d-infty}; there is no stochastic sampling error.  Representative exact values of $H_1(A_\ell)$, together with a decomposition of the corresponding augmented cycle entropy, are listed in Appendix~\ref{app:finite-data}.

As an internal check of the exact interval distribution, we fit the finite-size entropies to
\begin{equation}
H_1(A_\ell)=s\ell+\gamma\log\ell+c+
\frac{a_1}{\ell}+\frac{a_2}{\ell^2}+\frac{a_3}{\ell^3}.
\label{eq:sub-fit}
\end{equation}
The literature value of the logarithmic coefficient is not used as a fit input.  We instead vary the lower size cutoff and correction order, and also repeat the extrapolation with the extensive coefficient fixed to the independently determined periodic value $s=0.42327356$.  The latter removes the strongest $s$--$\gamma$ correlation.  Representative fit windows, their degrees of freedom, the fixed-$s$ sequence, and a diagnostic including a logarithmically enhanced correction are collected in Appendix~\ref{app:fit-stability}; see in particular \Cref{tab:interval-fit,tab:interval-fit-fixed-s}.  The stable higher-cutoff windows give the internal reconstruction
\begin{equation}
s_1\simeq0.42327356,\qquad \gamma_1=0.060020(3).
\label{eq:gamma-direct-main}
\end{equation}
The parenthetical uncertainty is a deterministic extrapolation envelope over the tested finite-size descriptions, not a statistical standard error and not a claim of greater precision than the existing large-size literature.  The purpose of this calculation is to verify that the exact Walsh distribution and the subsystem bookkeeping reproduce the established universal logarithm from first principles at the finite sizes accessible to the complete distribution.

The normalization relating our single-interval coefficient to the published mutual-information result is worth making explicit.  On a periodic ring define the chord length
\begin{equation}
d(\ell,L)=\frac{L}{\pi}\sin\frac{\pi\ell}{L}.
\end{equation}
If a single-interval Shannon entropy has the universal form
\begin{equation}
H_1(A_\ell)=s_1\ell+\gamma_1\log d(\ell,L)+O(1),
\end{equation}
then the complementary interval has the same logarithm, while the entropy of the complete periodic system has no $\log L$ term.  Therefore
\begin{equation}
I_1(A:B)=H_1(A)+H_1(B)-H_1(AB)
=2\gamma_1\log d(\ell,L)+O(1).
\end{equation}
Ref.~\cite{Stephan2014} uses the convention
\begin{equation}
I_1(A:B)=\frac{b_1}{4}\log d(\ell,L)+O(1),
\end{equation}
so that
\begin{equation}
\gamma_1=\frac{b_1}{8}.
\label{eq:gamma-b1-relation}
\end{equation}
The published value $b_1=0.4801629(2)$ thus implies
\begin{equation}
\frac{b_1}{8}=0.06002036(3),
\end{equation}
in agreement with Eq.~\eqref{eq:gamma-direct-main}.  This agreement should be read as a validation of the present exact-distribution reconstruction, not as a new high-precision determination of a previously unknown coefficient.  The same non-$U(1)$ $n\le1$ logarithmic branch in critical quadratic fermions was independently identified in Ref.~\cite{Tarighi2022}.

Combining our internal reconstruction with the independently derived edge coefficient $\gamma_E$ in Eq.~\eqref{eq:gammaE-main} gives the logarithmic coefficient of the \emph{complete} conditional entropy,
\begin{equation}
\gamma_C\equiv\gamma_E-\gamma_1\simeq0.05077.
\label{eq:gammaC-direct}
\end{equation}
This number characterizes all of $C_\ell$ in Eq.~\eqref{eq:Cell-full-linked}; it is not the coefficient of $B_{2,\ell}$ or of any other individual raw support sector.

The analytic and numerical roles are therefore cleanly separated.  The ghost reduction, the independent coefficient $\gamma_E$, and the augmented M\"obius decomposition are exact; the arbitrary-support monomer--dimer kernels and boundary-forest prescription provide the conditional all-orders hierarchy in Eqs.~\eqref{eq:boundary-delta-allorder}--\eqref{eq:interval-gamma-allorder}, with the first complete augmented coefficient $\kappa_3^{+}=2I_2$ controlled analytically.  What remains open is to prove the all-order boundary renormalization and lattice matching, evaluate the higher renormalized periods, and control their sum.  The finite-size Walsh reconstruction independently verifies that this organization is compatible with the established subsystem logarithm.  The Cauchy deformation in the next section explains why a particular hyperbolic constant appears already in $I_2$; a brief numerical observation involving that constant is recorded there separately and is not used in the analysis.
\FloatBarrier
\section{Unifying Cauchy structure and all-order organization}
\label{sec:cauchy-family}

The previous two sections analyze different universal quantities with parallel all-support organizations: the ring concerns an $O(1)$ constant built from perfect-matching linked periods, while the interval concerns a boundary logarithm built from physical-line and monomer--dimer boundary sectors.  The interval hierarchy requires the additional subtraction of nested endpoint scales, so its all-order analytic control is less complete.  Their common origin becomes visible only after those problems have been treated separately.  Both descend from the same Cauchy principal-minor structure.  We now introduce an exact one-parameter deformation of that algebraic probability kernel.  The deformation has three roles: it preserves the exact Cauchy/parity structure away from the Ising point, exposes the interval quantities $\gamma_E(\nu)$ and the complete first augmented coefficient $\kappa_3^{+}(\nu)=2I_2(\nu)$ as explicit functions of $\nu$, and supplies the natural parameter $z=4\nu^2$ that grades the periodic linked hierarchy.  The grand construction in the final subsection is specific to the periodic all-order organization.

\subsection{Exact Cauchy deformation and interval boundary data}
\label{sec:cauchy-interval-deformation}

The status of this parameter should be kept distinct from the physical model generalizations discussed elsewhere in the paper.  At $\nu=1/2$ the kernel below is the critical TFI kernel.  For generic $0<\nu<1$, however, we use it as an exact Toeplitz--Cauchy deformation of the principal minors and of the induced parity-boundary probability law.  No local spin-chain parent Hamiltonian is assumed or required for generic $\nu$.  In particular, this auxiliary deformation should not be conflated with the separate family of genuine critical quadratic chains $f_m(z)=z^m-1$ treated in Appendix~\ref{app:decimation}.  On the infinite line define
\begin{equation}
(G_\nu)_{jk}=\frac{\sin\pi\nu}{\pi(j-k+\nu)},\qquad 0<\nu<1.
\label{eq:Gnu-main}
\end{equation}
It is Toeplitz (its entries depend only on $j-k$) and is generated by the pure-jump Fourier symbol $g_\nu(e^{i\theta})=e^{i\nu(\pi-\theta)}$.  Its principal minors are
\begin{equation}
d_\nu(S)=\left(\frac{\sin\pi\nu}{\pi\nu}\right)^{|S|}
\prod_{i<j\in S}\frac{(i-j)^2}{(i-j)^2-\nu^2}.
\label{eq:dnu-main}
\end{equation}
The corresponding independent-edge parity graph has
\begin{equation}
q_r(\nu)=\sqrt{1-\frac{\nu^2}{r^2}},
\qquad
\beta_r(\nu)=\frac12\left[1-\sqrt{1-\frac{\nu^2}{r^2}}\right].
\label{eq:beta-nu}
\end{equation}
Euler's sine product supplies the one-body factor
\begin{equation}
\prod_{r=1}^{\infty}\left(1-\frac{\nu^2}{r^2}\right)=\frac{\sin\pi\nu}{\pi\nu}.
\end{equation}
Thus the entire Bernoulli representation deforms with the Cauchy kernel, not just its asymptotic tail.

For an interval, the exterior product remains explicit:
\begin{equation}
P_m(\nu)=\prod_{r=m}^{\infty}\left(1-\frac{\nu^2}{r^2}\right)
=\frac{\Gamma(m)^2}{\Gamma(m-\nu)\Gamma(m+\nu)},
\end{equation}
and
\begin{equation}
1-2\lambda_{i,\ell}^{(\nu)}=\sqrt{P_i(\nu)P_{\ell+1-i}(\nu)}.
\end{equation}
The far-boundary activity is
\begin{equation}
\lambda_R^{(\nu)}=\frac{\nu^2}{4R}+O(R^{-2}).
\end{equation}
The calculation leading to \eqref{eq:Eell-asympt} can then be repeated for arbitrary $\nu$ and gives
\begin{equation}
\gamma_E(\nu)=-\frac{\nu^2}{4}\left[1+\log\frac{\nu^2}{4}\right].
\label{eq:gammaEnu-main}
\end{equation}

The two-vertex macroscopic period also closes.  Set $a=\nu^2/4$, $u=a^2/t$ and $v=a/(1-t)^2$.  Then
\begin{equation}
I_2(\nu)=\int_0^1[(u+v)\log(u+v)-u\log u-v\log v]\dd t.
\end{equation}
Parametrize
\begin{equation}
\cosh\eta=\frac1\nu,
\end{equation}
so $\eta=\arcosh(1/\nu)$, with $\arcosh$ denoting the inverse hyperbolic cosine.
The exact result is
\begin{equation}
I_2(\nu)=\frac{\nu^4}{16}\left[
\frac{\pi^2}{2}+\frac4{\nu^2}\log\frac2\nu
-\frac{4\sqrt{1-\nu^2}}{\nu^2}\arcosh\frac1\nu
+2\arcosh^2\frac1\nu\right].
\label{eq:I2nu}
\end{equation}
At $\nu=1/2$ this reduces to \eqref{eq:I2}.  The appearance of $\log(2+\sqrt3)$ is therefore fixed by the hyperbolic parametrization of the exact Cauchy family: at the Ising point the canonical coordinate is $\eta=\arcosh2=\log(2+\sqrt3)$.

The cancellation in the first complete augmented sector also extends throughout this family.  Define
\begin{equation}
s_{\edge}(\nu)=\sum_{r\ge1}h(\beta_r(\nu)).
\end{equation}
Since $\beta_R(\nu)=\nu^2/(4R^2)+O(R^{-4})$, the fixed-separation physical-line triangle gives
\begin{equation}
\alpha_3^{\linesys}(\nu)=-\nu^2s_{\edge}(\nu).
\label{eq:alpha3nu}
\end{equation}
On the ghost side $\lambda_R^{(\nu)}=\nu^2/(4R)+O(R^{-2})$, so the same exact triangle-fiber expansion gives
\begin{equation}
\delta_2^{\partial}(\nu)=2I_2(\nu)+\nu^2s_{\edge}(\nu).
\label{eq:delta2nu}
\end{equation}
The boundary-layer terms therefore cancel identically at fixed augmented support,
\begin{equation}
\boxed{\kappa_3^{+}(\nu)=\alpha_3^{\linesys}(\nu)+\delta_2^{\partial}(\nu)=2I_2(\nu).}
\label{eq:kappa3nu}
\end{equation}
Thus the first complete interval linked coefficient, not only its ghost component, is known in closed form throughout the auxiliary Cauchy family.

\subsection{Periodic deformation and coefficient homogeneity}
\label{sec:periodic-deformation}

For the periodic ring the same auxiliary Cauchy deformation has the finite-size realization
\begin{equation}
(G_{\nu,L})_{jk}=\frac{\sin\pi\nu}{L\sin[\pi(j-k+\nu)/L]}.
\label{eq:finite-cauchy-nu}
\end{equation}
Its macroscopic physical-edge activity is multiplied by $4\nu^2$ relative to the Ising point.  Since a $2p$-vertex matching contains $p$ edges, the renormalized even-sector coefficients are expected, within the same finite-part matching prescription, to obey the homogeneity relation
\begin{equation}
 c_{2p}(\nu)=(4\nu^2)^p c_{2p}(1/2),
\label{eq:homogeneity}
\end{equation}
This homogeneity is exact for the full-support kernel and its local counterterms; its identification with the lattice constants is conditional on Conjecture~\ref{conj:FPmatch}.  It is therefore natural to introduce
\begin{equation}
C(z)=\sum_{p\ge2}c_{2p}z^p,\qquad z=4\nu^2,
\label{eq:Cz}
\end{equation}
with the physical periodic Ising point at $z=1$.  Thus the same deformation that probes the interval boundary geometry also supplies a natural bookkeeping parameter for the periodic linked hierarchy.  Direct finite-size calculations at $z<1$ provide an independent check of this generating interpretation because higher linked orders are suppressed; the finite-ring construction and numerical comparison are given in Appendix~\ref{app:deformation-test}.

\subsection{Formal grand determinant/Hafnian organization}
\label{sec:grand}

The rapid periodic coefficient hierarchy suggests a useful all-order generating framework, but its logical status is different from the fixed-support construction developed above.  It is therefore important to separate three levels: an exact identity at each fixed support, an exactly summable Hafnian object in its convergence domain, and a formal determinant grand object whose renormalization and resummation are not yet established.  None of the finite-support results of Sec.~\ref{sec:periodic-linked} depends on this grand construction.

\paragraph{Exact fixed-support identity.}
For $n=2p$ distinct points on the circle define
\begin{equation}
D_{2p}(q;x)=[\det K_{2p}(x)]^q,
\qquad
H_{2p}(q;x)=\Hf[(K_{ij}(x)^2)^q].
\label{eq:grand-fixed-DH}
\end{equation}
At $q=1$ the Pfaffian--Hafnian identity gives $D_{2p}(1;x)=H_{2p}(1;x)=Z_p(x)$, and differentiation is pointwise exact:
\begin{equation}
\Gamma_{2p}(x)
=\left.\partial_q\left[D_{2p}(q;x)-H_{2p}(q;x)\right]\right|_{q=1}.
\label{eq:grand-fixed-derivative}
\end{equation}
Thus, conditional only on the lattice--continuum finite-part matching of Conjecture~\ref{conj:FPmatch}, the fixed-support coefficient can be written in the controlled form
\begin{equation}
c_{2p}
=\frac1{(2p)!}\FP_{\rm linked}
\int_{\mathbb T^{2p}}
\left.\partial_q\left[D_{2p}(q;x)-H_{2p}(q;x)\right]\right|_{q=1}
\dd^{2p}x.
\label{eq:grand-fixed-controlled}
\end{equation}
The order of operations in Eq.~\eqref{eq:grand-fixed-controlled} is deliberate: the $q$ derivative is taken at fixed noncoincident configuration first, producing the Shannon kernel $\Gamma_{2p}$ for which the linked finite part has already been defined.  No $q$-dependent extension of the collision subtraction is needed for this statement.

\paragraph{Exactly summable Hafnian side.}
The matching side can be packaged independently in a genuine convergence domain.  For $\Re q<1/2$ define
\begin{equation}
J(q)=\int_0^1\dd x\int_0^1\dd y\,
\left[\frac{\pi^2}{16\sin^2\!\pi(x-y)}\right]^q
=\left(\frac\pi4\right)^{2q}
\frac{\Gamma(\frac12-q)}{\sqrt\pi\Gamma(1-q)}.
\label{eq:Jq-grand}
\end{equation}
In this half-plane the integrals are absolutely convergent.  Every perfect matching factorizes into $p$ copies of $J(q)$, and therefore
\begin{align}
\cB(q,z)
&=\sum_{p=0}^{\infty}\frac{z^p}{(2p)!}
\int_{\mathbb T^{2p}}H_{2p}(q;x)\,\dd^{2p}x
\label{eq:Bqz}\\
&=\exp\!\left[\frac{zJ(q)}2\right].
\label{eq:B-exact}
\end{align}
The second line supplies a meromorphic continuation in $q$; from this point onward $\cB(q,z)$ denotes that continuation when $q$ lies outside the convergence half-plane.  Expanding it about $q=1$ gives
\begin{equation}
\cB(1,z)=1,
\qquad
\left.\partial_q\cB(q,z)\right|_{q=1}=\frac{\pi^2}{16}z.
\label{eq:B-derivative-one}
\end{equation}
This continuation of the explicitly solvable matching side reproduces the magnitude needed to cancel the deformed periodic independent-edge anomaly.

\paragraph{Formal determinant side and the missing compatibility theorem.}
Writing an analogous determinant grand object near $q=1$ requires additional structure that has not been constructed in this work.  The linked finite part used elsewhere is defined for the Shannon combination $\Gamma_{2p}$ after the derivative in Eq.~\eqref{eq:grand-fixed-derivative}; it does not automatically define separate finite parts of $D_{2p}(q)$ and $H_{2p}(q)$ for $q$ in a neighborhood of one.  In particular, the individual $q$-dependent terms can have collision singularities that cancel only after the Shannon combination is formed.

To display the prospective all-order structure without hiding this issue, suppose that there exists a $q$-dependent linked renormalization $\operatorname{FP}_{q,\rm linked}$ for the determinant coefficients, compatible with the already defined finite part at $q=1$, and set formally
\begin{equation}
\cA(q,z)
=\sum_{p=0}^{\infty}\frac{z^p}{(2p)!}
\operatorname{FP}_{q,\rm linked}
\int_{\mathbb T^{2p}}D_{2p}(q;x)\,\dd^{2p}x.
\label{eq:Aqz}
\end{equation}
The required compatibility is precisely that differentiation of the separately renormalized objects reproduce the finite part of the differentiated difference,
\begin{multline}
\left.\partial_q\left[
\operatorname{FP}_{q,\rm linked}\int D_{2p}(q;x)\,\dd^{2p}x
-\operatorname{AC}\int H_{2p}(q;x)\,\dd^{2p}x
\right]\right|_{q=1}
\\
=\FP_{\rm linked}\int
\left.\partial_q\left[D_{2p}(q;x)-H_{2p}(q;x)\right]\right|_{q=1}
\dd^{2p}x,
\label{eq:q-ren-compatibility}
\end{multline}
where $\operatorname{AC}$ denotes the analytic continuation from the convergent Hafnian domain.  Equation~\eqref{eq:q-ren-compatibility} is not proved here; it is the extra renormalization statement needed before the two grand objects can be differentiated separately.

If this compatibility holds coefficient by coefficient, then the generating series introduced in Eq.~\eqref{eq:Cz} obeys, as a \emph{formal power-series identity in $z$},
\begin{equation}
C(z)\stackrel{\rm formal}{=}
\left.\partial_q[\cA(q,z)-\cB(q,z)]\right|_{q=1}.
\label{eq:C-grand}
\end{equation}
No interchange of an infinite particle-number sum with the finite-part operation is asserted by Eq.~\eqref{eq:C-grand}; it is only a coefficientwise packaging of the fixed-$p$ relations.

Passing from this formal series to the physical point $z=1$ requires still more.  One needs convergence or a specified summation of the even-support series at $z=1$, control of its higher-support tail, and the all-order odd-support selection rule in Eq.~\eqref{eq:odd-selection-main}.  Only under these additional assumptions does the cancellation in Eq.~\eqref{eq:B-derivative-one} lead, for the auxiliary periodic parity law, to the formal constant
\begin{equation}
r^{\rm aux}(z)\stackrel{\rm formal}{=}
-\left.\partial_q\cA(q,z)\right|_{q=1}.
\label{eq:grand-master-z}
\end{equation}
At $z=1$ this auxiliary law is the physical TFI $z$-basis law; using the exact duality shift then gives, still under the same all-order assumptions,
\begin{equation}
r_1^{(x)}\stackrel{\rm formal}{=}
\log2-\left.\partial_q\cA(q,1)\right|_{q=1}.
\label{eq:grand-master-x}
\end{equation}
These equations are therefore an all-order organizing target rather than an independent derivation of the periodic constant.  The exact content is the fixed-support derivative identity and the closed Hafnian continuation; the unresolved content is the $q$-dependent determinant renormalization, its compatibility with the linked finite part, and the resummation to $z=1$.  Appendix~\ref{app:grand-derivation} records the coefficient normalization and the Hafnian calculation in this hierarchy.

The formal object \eqref{eq:Aqz} is a continuum analogue of the sum of powers of principal minors (SPPM), for which a Berezin field-theory representation was introduced in Ref.~\cite{NajafiRamezanpourRajabpour2025}.  Related finite-lattice all-minors sums in the TFI-stabilizer/$XX$-chain Shannon--R\'enyi correspondence admit a discrete-Selberg organization and determinant/Pfaffian reductions at special indices~\cite{RamirezTrinoRajabpour2026,KhassehRajabpour2026}.  The present proposal is nevertheless different: it would require a linked renormalization of an antisymmetric Cauchy determinant power and its derivative at $q=1$, rather than the fugacity-resolved finite-lattice moment itself.  Constructing such a renormalization, and only then seeking an integrable or tau-function representation, is a natural route toward an all-orders analytic periodic constant.

\FloatBarrier
\section{Discussion and conclusions}
\label{sec:discussion}
\label{sec:conclusion}

The central result is the exact identification of the critical TFI measurement distribution with a classical parity-boundary distribution.  Formation probabilities in free-fermion chains are naturally expressed through determinant methods~\cite{Shiroishi2001,FranchiniAbanov2005,StephanEFP2014,NajafiRajabpour2016,NajafiRajabpour2020,AresViti2020,AresRajabpourViti2020,AresRajabpourViti2021}, with Pfaffian extensions available in noncomputational bases~\cite{TarighiKhassehRajabpour2024}; here the additional trigonometric Cauchy structure factorizes every principal minor and promotes those individual probability formulas to an identity for the complete Born law.  At $n=1$ the ordinary entropy chain rule then gives Eq.~\eqref{eq:edge-cycle}, so that the difficult many-body Shannon entropy is separated into a known independent-edge term and a conditional entropy on the graph cycle space.

For the periodic chain this separation isolates the universal edge contribution $-\pi^2/16$ analytically.  The remaining constant is organized by the exact vertex M\"obius expansion.  At even macroscopic support the leading ambiguity is carried by perfect matchings, and the Pfaffian--Hafnian identity reduces the corresponding matching partition function to determinants.  Odd support has a complementary structural suppression: unmatched-leg power counting, together with recursive removal of lower-support collision faces, gives the conditional selection rule \eqref{eq:odd-selection-main}; a rigorous all-order forest theorem remains open.  Conditional on this forest closure and on the finite-part matching conjecture of Sec.~\ref{sec:periodic-linked}, every nonvanishing linked intercept is given by the explicit finite-dimensional integral \eqref{eq:c2p-period}, providing an all-orders analytic hierarchy for the periodic constant.  The four-vertex continuum integral is evaluated in closed form.  The six- and eight-vertex lattice sectors have the signs and magnitudes shown in Eq.~\eqref{eq:coeffs-main}; retaining these first three linked sectors gives $0.25453$, already within about $2\times10^{-4}$ of the high-precision periodic constant.  This rapid finite-order saturation provides strong numerical evidence that the linked-support hierarchy organizes the universal term efficiently, without constituting a convergence proof.  The complete constant is obtained independently from the full finite-size distribution.

The interval geometry exposes a different aspect of the same probability law.  Marginalizing the exterior does not generate a complicated correlated bath: all exterior information relevant to the measured interval collapses to independent parity variables, represented by a single ghost vertex.  The augmented conditional entropy then has an exact vertex-linked decomposition at every support order into physical-line sectors and sectors whose support contains the ghost.  This makes the cancellation of the two $(\log\ell)^2$ independent-edge terms explicit and leaves the analytic edge coefficient $\gamma_E=\log2/4-1/16$.  The remaining cycle coefficient is itself organized to all orders: physical-line supports provide open-line linked periods, while an arbitrary $k$-vertex ghost support reduces at common scale to the monomer--dimer kernel $\Gamma_k^\partial$; conditional on closure of the boundary-forest subtraction and lattice face matching, these sectors give the series in Eqs.~\eqref{eq:interval-cycle-allorder}--\eqref{eq:interval-gamma-allorder}.  In contrast with the periodic hierarchy, every boundary order is marginal and nested endpoint scales must be subtracted before a primitive single-log coefficient is defined.  The first complete augmented coefficient is evaluated analytically: $\alpha_3^{\linesys}=-s_{\edge}/4$ cancels the boundary-layer term in $\delta_2^{\partial}$, leaving $\kappa_3^{+}=2I_2$ in closed form.  The same cancellation holds for the full auxiliary Cauchy family, $\kappa_3^{+}(\nu)=2I_2(\nu)$.  The complete Walsh reconstruction independently recovers $\gamma_1=0.060020(3)$, serving as an internal validation against the previously known mutual-information coefficient $b_1$~\cite{Stephan2014,Tarighi2022}.  The Cauchy deformation also singles out $\log(2+\sqrt3)$ as a natural boundary constant.  The established value is numerically compatible with $\gamma_1\approx[62+37\log2-64\log(2+\sqrt3)]/56$, but this relation is neither derived nor used elsewhere.

The relation to the R\'enyi boundary phase is also useful for locating the Shannon result.  Boundary-transition analyses show why a naive continuation to $n=1$ is not justified~\cite{Zaletel2011,Stephan2011}, while related Book-Ising geometries provide a rigorous statistical-mechanical setting for such boundary ordering~\cite{DuminilCopinGarbanTassion2023}.  At real $n>1$ the Cauchy probability kernel gives a microscopic derivation of the universal interval logarithm from the inverse-square pairs removed by the cut~\cite{RajabpourBoundaryPhase2026}.  At $n=1$, by contrast, the chain rule isolates the conditional cycle information evaluated here.  Related developments on weak measurement, decoherence, measurement-altered critical states and information-theoretic diagnostics include Refs.~\cite{LeeJianXu2023,YangMaoJian2023,PaviglianitiTurkeshiSchiroSilva2024,PatilLudwig2024,KhannaVasseur2026,PatilLudwig2025,KunoOritoIchinose2025,Heinrich2026}.

Finally, two logically distinct extensions indicate that the parity representation is not an isolated algebraic accident: the auxiliary Cauchy deformation provides an exact algebraic family of probability kernels around the Ising point, while the decimation of $f_m(z)=z^m-1$ is an exact statement for a genuine family of critical quadratic chains.  The symbol-based description of critical quadratic chains~\cite{VerresenJonesPollmann2018,JonesVerresen2019,Tarighi2022} supplies a natural setting in which to ask which other physical models admit similarly factorized principal minors.  Two analytic problems remain particularly clear: proving the lattice--continuum finite-part correspondence and the required all-order forest control for the periodic linked hierarchy, and proving closure and lattice matching of the nested boundary-forest hierarchy for the interval cycle entropy.  The all-orders organizations themselves are already explicit in both geometries; resolving these problems would place the corresponding renormalized series on firm analytic footing and permit controlled evaluation of their full sums.

\section*{Acknowledgements}
The author used OpenAI ChatGPT during the preparation of this work as an interactive tool for checking the consistency of analytical derivations and numerical results, auditing the LaTeX manuscript and cross-references, and improving language and presentation. All AI-assisted suggestions incorporated into the manuscript were critically reviewed by the author, who made all final scientific and editorial decisions and assumes full responsibility for the content.

\section*{Author contributions}
M.A.R. is the sole author of this work.

\section*{Funding information}
This work was partially supported by CNPq and FAPERJ (grant number E-26/210.062/\allowbreak 2023).

\section*{Data and code availability}
No external datasets were used.  The numerical results reported in this work are deterministic evaluations of the exact finite-size formulas given in the text and appendices.

\section*{Competing interests}
The author declares no competing interests.

\appendix
\section{Algebraic foundations of the probability law}
\label{app:algebra}

This appendix supplies the algebra behind the exact probability-law statements in Sec.~\ref{sec:probability}: the formation-probability determinant, its Cauchy factorization, the Bernoulli boundary law, and the infinite-volume ghost reduction.

\subsection{Formation-probability determinant and Walsh conventions}
\label{app:det-walsh}

For a fixed spin configuration $\bm\sigma$, its Born probability can first be written without any fermionic notation as the expectation value of the product of local projectors,
\begin{equation}
 p_{\bm\sigma}=\left\langle G\left|\prod_{j=0}^{L-1}\frac{1+\sigma_j^z\,\widehat\sigma_j^z}{2}\right|G\right\rangle .
 \label{eq:projector-prob}
\end{equation}
After the Jordan--Wigner transformation, which maps the spin chain to free fermions, the critical TFI ground state is Gaussian: all higher fermionic correlators are fixed by two-point contractions through Wick's theorem.  The projector expectation therefore reduces to the standard formation-probability determinant.  In the conventions used in the main text one obtains~\cite{NajafiRajabpour2016,NajafiRajabpour2020,AresRajabpourViti2020}
\begin{equation}
 p_{\bm\sigma}=\det M_{\bm\sigma},\qquad
 M_{\bm\sigma}=\frac12\left(I+D_{\bm\sigma}G_L\right),
 \label{eq:formation-M}
\end{equation}
with the critical finite-ring contraction matrix $G_L$ given in \eqref{eq:GL}.  Thus \eqref{eq:formation-M} is exactly \eqref{eq:det-born}.  This Gaussian/Wick reduction is standard free-fermion machinery; the nonstandard step used in this paper begins by exploiting the Cauchy form of $G_L$ and of all its principal minors.  We now spell out that determinant-to-Walsh algebra because it is essential for the parity-graph identification.

For $D_\sigma=\diag(\sigma_0,\ldots,\sigma_{L-1})$,
\begin{equation}
\det(I+D_\sigma G)=\sum_{S\subseteq[L]}\det[(D_\sigma G)[S]].
\end{equation}
Since $D_\sigma[S]$ is diagonal,
\begin{equation}
\det[(D_\sigma G)[S]]=\left(\prod_{j\in S}\sigma_j\right)\det G[S].
\end{equation}
Thus
\begin{equation}
2^{-L}\det(I+D_\sigma G)=2^{-L}\sum_S d(S)\chi_S(\sigma),
\qquad d(S)=\det G[S].
\end{equation}
The key point is that the Walsh characters form an orthogonal basis on the Boolean hypercube.  For $S,T\subseteq[L]$,
\begin{equation}
2^{-L}\sum_{\bm\sigma\in\{\pm1\}^L}
\chi_S(\bm\sigma)\chi_T(\bm\sigma)=\delta_{S,T}.
\label{eq:walsh-orthogonality-app}
\end{equation}
Indeed, $\chi_S\chi_T=\chi_{S\triangle T}$.  If $S\neq T$, choose $j\in S\triangle T$ and pair every spin configuration with the configuration obtained by flipping only $\sigma_j$; the two terms cancel.  If $S=T$, every term equals one.

Consequently, for any function $f$ on $\{\pm1\}^L$, define
\begin{equation}
\widehat f(S)=\sum_{\bm\sigma}f(\bm\sigma)\chi_S(\bm\sigma).
\end{equation}
Orthogonality gives the inverse transform
\begin{equation}
f(\bm\sigma)=2^{-L}\sum_{S\subseteq[L]}\widehat f(S)\chi_S(\bm\sigma).
\label{eq:walsh-inversion-app}
\end{equation}
Thus the complete set of Walsh coefficients determines a function, and in particular a probability law, uniquely.  Applied to the determinant expansion above, this shows that its coefficients $d(S)$ are precisely the Walsh coefficients of the Born distribution.  Normalization is the $S=\varnothing$ component: $d(\varnothing)=1$, while every nonempty Walsh character has zero average over the hypercube.

For the infinite interval calculation one simply takes $L\to\infty$ at fixed separations in the finite-ring Cauchy factorization.  Since
\begin{equation}
\frac{1}{L\sin(\pi/2L)}\to\frac2\pi,
\end{equation}
and
\begin{equation}
\frac{\sin^2(\pi m/L)}{\sin[\pi(m-1/2)/L]\sin[\pi(m+1/2)/L]}
\to\frac{4m^2}{4m^2-1},
\end{equation}
one obtains \eqref{eq:d-infty}.

\subsection{Trigonometric and ordinary Cauchy factorization}
\label{app:cauchy}

This subsection provides the factorization used in Sec.~\ref{sec:probability} to turn principal minors into graph-cut products.  The finite-ring matrix \eqref{eq:GL} is a trigonometric Cauchy matrix.  Up to row and column signs, write
\begin{equation}
G_{ab}=\frac{1}{L\sin(x_a-y_b)},
\qquad
x_a=\frac{\pi a}{L},\qquad y_b=\frac{\pi(b-1/2)}{L}.
\end{equation}
The trigonometric Cauchy determinant formula~\cite{Krattenthaler1999} is
\begin{equation}
\det\left[\frac1{\sin(x_a-y_b)}\right]
=\frac{\prod_{a<a'}\sin(x_a-x_{a'})\prod_{b<b'}\sin(y_{b'}-y_b)}
{\prod_{a,b}\sin(x_a-y_b)}
\end{equation}
up to an ordering sign.  Restricting to the same index set in rows and columns and grouping pairwise factors gives \eqref{eq:cauchy-factor}.  The pair factor follows from
\begin{equation}
\sin\!\left(x-\frac\pi{2L}\right)
\sin\!\left(x+\frac\pi{2L}\right)
=\sin^2x-\sin^2\frac\pi{2L}.
\end{equation}
The same sine products give the finite identity needed for the cut representation,
\begin{equation}
A_L^2\prod_{m=1}^{L-1}g_m^{(L)}=1.
\label{eq:finite-product-app}
\end{equation}
Indeed,
\begin{equation}
\prod_{m=1}^{L-1}\sin\frac{\pi m}{L}=\frac{L}{2^{L-1}},
\qquad
\prod_{m=0}^{L-1}\sin\frac{\pi(m+1/2)}{L}=2^{1-L},
\end{equation}
and substitution into the product of Eq.~\eqref{eq:g-z} gives Eq.~\eqref{eq:finite-product-app}.

For the infinite family \eqref{eq:Gnu-main}, ordinary Cauchy factorization gives
\begin{align}
\det_{a,b}\frac1{s_a-s_b+\nu}
&=\frac1{\nu^k}\prod_{a<b}\frac{(s_a-s_b)^2}{(s_a-s_b)^2-\nu^2},
\end{align}
which immediately yields \eqref{eq:dnu-main} after multiplying by $(\sin\pi\nu/\pi)^k$.

\subsection{Bernoulli parity-boundary law}
\label{app:bernoulli-proof}

We now complete the determinant-to-probability step used in Eq.~\eqref{eq:parity-law} by matching the cut products to the characteristic function of a Bernoulli graph boundary.  Let $\eta_e$ be independent Bernoulli variables with probabilities $\beta_e$.  For a vertex subset $S$,
\begin{equation}
\chi_S(\partial F)=(-1)^{|S\cap\partial F|}.
\end{equation}
Every edge with both endpoints in $S$ or both outside $S$ contributes an even number to $|S\cap\partial F|$ and cancels.  Let $\delta(S)$ denote the set of edges with exactly one endpoint in $S$.  Every cut edge $e\in\delta(S)$ contributes $(-1)^{\eta_e}$.  Therefore
\begin{equation}
\EE\,\chi_S(\partial F)=\prod_{e\in\delta(S)}\EE(-1)^{\eta_e}
=\prod_{e\in\delta(S)}(1-2\beta_e).
\end{equation}
To derive the TFI cut product rather than assume it, apply Eq.~\eqref{eq:finite-product-app} around every vertex and multiply it over $i\in S$.  Internal pairs in $S$ are then counted twice, while each cut pair is counted once:
\begin{equation}
1=A_L^{2|S|}
\left[\prod_{\substack{i<j\\i,j\in S}}(g_{j-i}^{(L)})^2\right]
\left[\prod_{e\in\delta(S)}g_e^{(L)}\right].
\end{equation}
The first two factors are exactly $d_L(S)^2$ by Eq.~\eqref{eq:cauchy-factor}.  Hence
\begin{equation}
d_L(S)^2=\prod_{e\in\delta(S)}\frac1{g_e^{(L)}}
=\prod_{e\in\delta(S)}(1-z_e).
\end{equation}
The principal minors are positive for the present kernel, so
\begin{equation}
d_L(S)=\prod_{e\in\delta(S)}\sqrt{1-z_e}.
\label{eq:cut-product-app}
\end{equation}
Choosing $1-2\beta_e=\sqrt{1-z_e}$ makes the Walsh coefficient of the Bernoulli boundary law equal to $d_L(S)$ for every $S\subseteq[L]$.  By the determinant expansion above, the Born law has the same coefficients.  The orthogonality and inversion formulae \eqref{eq:walsh-orthogonality-app}--\eqref{eq:walsh-inversion-app} therefore imply equality of the two laws pointwise, which proves \eqref{eq:parity-law}.

The entropy identity follows from the deterministic map $F\mapsto\partial F$:
\begin{equation}
H(F)=H(\partial F)+H(F\mid\partial F),
\end{equation}
because $H(\partial F\mid F)=0$.  Thus $H(\partial F)=H(F)-H(F\mid\partial F)$.

Finally, $\ker\partial$ is exactly the binary cycle space of the graph.  If $F$ and $F'$ have the same boundary then $\partial(F\triangle F')=0$, so their difference is an even subgraph.  Conditional entropy at fixed boundary is therefore entropy within a coset of the weighted cycle code.

\subsection{Exact support checkpoint and an equivalent current representation}
\label{app:exact-checkpoints}

The parity law immediately fixes the support of the finite periodic Born distribution.  A graph boundary always has even cardinality, and every even vertex subset occurs with positive probability because $0<\beta_e<1$.  Hence
\begin{equation}
|\operatorname{supp}p_L|=2^{L-1},
\qquad
H_0^{(z)}(L)=(L-1)\log2.
\label{eq:H0-check}
\end{equation}
This gives the exact endpoint values $s_0=\log2$, $r_0^{(z)}=-\log2$ and, using the duality shift in Eq.~\eqref{eq:KW}, $r_0^{(x)}=0$.  We use this only as an exact consistency check; the paper does not rely on an extrapolation in R\'enyi index.

The same Cauchy product also admits the exact auxiliary Ising/current form used in Ref.~\cite{RajabpourBoundaryPhase2026}.  With
\begin{equation}
J_m^{(L)}=\frac14\log g_m^{(L)}
\end{equation}
and $\tau_j=(-1)^{\mathbf 1(j\in S)}$, the principal minor can be written as a zero-field long-range Ising weight up to the normalization fixed by the finite product identity.  Expanding each factor
$e^{J_e\tau_i\tau_j}=\cosh J_e(1+\tanh J_e\,\tau_i\tau_j)$
gives a high-temperature current sum over edge sets with prescribed boundary.  Normalizing that sum makes each edge independently occupied with
\begin{equation}
\frac{\tanh J_e}{1+\tanh J_e}
=\frac12\left(1-\sqrt{1-z_e}\right)
=\beta_e,
\end{equation}
which is an independent derivation of the Bernoulli law used in the main text.  In the infinite-line limit
\begin{equation}
J_m=\frac14\log\frac{4m^2}{4m^2-1}
=\frac{1}{16m^2}+O(m^{-4}),
\end{equation}
making the inverse-square character of the auxiliary interaction explicit.

\subsection{Normalized R\'enyi moment and the exact Shannon derivative}
\label{app:renyi-moment}

This short derivation explains why the normalized R\'enyi moment has an especially simple derivative at the Shannon point and connects the neighboring moment formulation to the chain-rule decomposition used in the main text.  For a finite weighted graph $G=(V,E)$ with boundary law $P_G(D)=\PP\{\partial F=D\}$, define the R\'enyi moment of its boundary law and the corresponding one-edge moment by
\begin{equation}
\mathcal M_n(G)=\sum_D P_G(D)^n,
\qquad
b_e(n)=(1-\beta_e)^n+\beta_e^n,
\end{equation}
and normalize out the independent-edge contribution,
\begin{equation}
\widehat{\mathcal M}_n(G)=
\frac{\mathcal M_n(G)}{\prod_{e\in E}b_e(n)}.
\end{equation}
For a forest the boundary map is injective, hence $\widehat{\mathcal M}_n(G)=1$ for every $n$.  At $n=1$, $\mathcal M_1=b_e(1)=1$, while direct differentiation gives
\begin{equation}
\left.\partial_n\log\mathcal M_n(G)\right|_{n=1}=-H(\partial F),
\qquad
\left.\partial_n\log b_e(n)\right|_{n=1}=-h(\beta_e).
\end{equation}
Therefore the cycle conditional entropy has the exact finite-graph representation
\begin{equation}
H(F\mid\partial F)
=\left.\partial_n\log\widehat{\mathcal M}_n(G)\right|_{n=1}.
\label{eq:Hcyc-renyi-derivative}
\end{equation}
This identity does not assume a smooth thermodynamic continuation in $n$: it is simply the derivative of a finite-graph moment at $n=1$.  It also makes precise why the Shannon point is special even though the underlying Bernoulli probability law exists independently of $n$.

\subsection{Exterior product, infinite-volume limit and ghost independence}
\label{app:ghost-details}

For the infinite-line edge characteristic factor
\begin{equation}
q_m=\sqrt{1-\frac1{4m^2}},
\end{equation}
we have $\beta_m=O(m^{-2})$ and therefore $\sum_{m\ge1}\beta_m<\infty$.  A convenient way to make the exterior parity rigorous is to truncate the crossing edges at a maximal length $M$.  Let $\xi_i^{(M)}$ be the parity of all crossing edges incident on $i$ whose length is at most $M$.  For fixed $i$, the expected number of occupied crossing edges omitted beyond $M$ tends to zero as $M\to\infty$, and the full sum of their occupation probabilities is finite.  By the Borel--Cantelli lemma only finitely many crossing edges incident on $i$ are occupied almost surely.  Hence $\xi_i^{(M)}$ stabilizes almost surely to a well-defined limit $\xi_i$.

At finite $M$, $\xi_i^{(M)}$ and $\xi_j^{(M)}$ are functions of disjoint sets of Bernoulli edge variables whenever $i\neq j$: an edge crossing the interval boundary has exactly one endpoint in $A_\ell$.  Thus the truncated parity variables are independent.  Passing to the almost-sure limits preserves the factorized joint characteristic function, so the full exterior parity variables $\{\xi_i\}_{i\in A_\ell}$ are independent as asserted in the main text.

The exterior parity at site $i$ contains edges of lengths $m=i,i+1,\ldots$ on the left and $m=\ell+1-i,\ell+2-i,\ldots$ on the right.  For a parity of independent Bernoulli variables,
\begin{equation}
\EE(-1)^{\xi_i}=1-2\lambda_{i,\ell}
=\prod_{m=i}^{\infty}q_m\prod_{m=\ell+1-i}^{\infty}q_m.
\end{equation}
The squared tail product is
\begin{align}
P_n&=\prod_{m=n}^{\infty}\frac{(m-1/2)(m+1/2)}{m^2}\\
&=\frac{\Gamma(n)^2}{\Gamma(n-1/2)\Gamma(n+1/2)}.
\end{align}
This proves Eq.~\eqref{eq:lambda-main}.

Now let $\zeta_i$ be independent Bernoulli variables with the same one-site probabilities $\lambda_{i,\ell}$ and represent $\zeta_i$ by the ghost edge $(0,i)$.  Since the joint laws of $\xi$ and $\zeta$ are identical, replacing the entire exterior crossing-edge cloud by these $\ell$ ghost edges preserves the vector of physical parities exactly in distribution:
\begin{equation}
\partial F_A\oplus\xi\ \stackrel{d}{=}\ \partial F_A\oplus\zeta.
\end{equation}
This replacement does not assert equality of the microscopic exterior edge configurations; it identifies only the sufficient parity variables relevant to the subsystem measurement law.

Finally, for the augmented finite graph $F^+$ the total boundary has even parity.  Hence
\begin{equation}
(\partial F^+)_0=\bigoplus_{i=1}^{\ell}(\partial F^+)_i.
\end{equation}
The full augmented boundary is therefore a deterministic one-to-one extension of its physical components.  Consequently
\begin{equation}
H(D)=H(\partial F^+),
\end{equation}
and because $F^+\mapsto\partial F^+$ is deterministic,
\begin{equation}
H(D)=H(F^+)-H(F^+\mid\partial F^+),
\end{equation}
which proves Eq.~\eqref{eq:ghost-entropy-chain} and the subsystem split \eqref{eq:sub-split}.

\section{Asymptotics, Möbius inversion and finite-part prescriptions}
\label{app:asymptotics}

This appendix supplies the asymptotic and subtraction machinery used in the periodic and interval analyses: first the explicit edge terms, then support-space Möbius inversion, and finally the collision and boundary finite-part prescriptions.

\subsection{Periodic edge anomaly from singular Euler--Maclaurin}
\label{app:edge-anomaly}

This subsection derives the universal $-\pi^2/16$ independent-edge constant quoted in Sec.~\ref{sec:periodic-edge} and isolates the nonanalytic endpoint term responsible for it.  Set
\begin{equation}
f(z)=h\!\left(\frac{1-\sqrt{1-z}}2\right).
\end{equation}
The singular part near $z=0$ is
\begin{equation}
f_{\rm sing}(z)=-\frac14 z\log z.
\end{equation}
The analytic part does not generate an $O(1)$ ring anomaly.  This can be seen explicitly before treating the logarithm.  If
\begin{equation}
 f(z)=-\frac14z\log z+a_1z+a_2z^2+\cdots,
 \qquad a_1=\frac{1+\log4}{4},
\end{equation}
then the linear contribution obeys the exact cosecant sum
\begin{align}
 \frac{L}{2}a_1\sum_{m=1}^{L-1}z_m^{(L)}
 &=\frac{La_1}{2}\sin^2\!\delta\sum_{m=1}^{L-1}\csc^2\frac{\pi m}{L}\\
 &=\frac{La_1}{6}(L^2-1)\sin^2\!\delta
 =\frac{a_1\pi^2}{24}L+O(L^{-1}),
\end{align}
where $\sum_{m=1}^{L-1}\csc^2(\pi m/L)=(L^2-1)/3$.  There is no $L^0$ term.  The same conclusion holds term by term for every analytic power $z^k$: the finite cosecant-power sum is a polynomial in $L^2$ of degree $k$, while $\sin^{2k}\delta$ has an expansion in even powers of $L^{-1}$.  After multiplication by the outer factor $L/2$, only $L,L^{-1},L^{-3},\ldots$ can occur.  Since the physical range satisfies $z_m^{(L)}<1$ uniformly, the regular power series may be treated termwise.  Hence the constant anomaly comes from the nonanalytic $-z\log z/4$ term.  Define
\begin{equation}
S_L=\sum_{m=1}^{L-1}z_m^{(L)}\log z_m^{(L)},
\qquad z_m^{(L)}=\frac{\sin^2\delta}{\sin^2(\pi m/L)},\quad\delta=\frac\pi{2L}.
\end{equation}
Introduce
\begin{equation}
C_L(a)=\sum_{m=1}^{L-1}\sin^{-2a}\frac{\pi m}{L}.
\end{equation}
Then
\begin{equation}
S_L=\sin^2\delta\left[2\log(\sin\delta)C_L(1)+C'_L(1)\right].
\end{equation}
Near $a=1$ the singular Euler--Maclaurin form is
\begin{equation}
C_L(a)=2\left(\frac L\pi\right)^{2a}\zeta(2a)+L I(a)+O(1),
\end{equation}
where by analytic continuation
\begin{equation}
I(a)=\int_0^1\sin^{-2a}(\pi x)\dd x
=\frac{\Gamma(\frac12-a)}{\sqrt\pi\Gamma(1-a)}.
\end{equation}
At $a=1$, $I(1)=0$ and $I'(1)=2$.  Since $\sin^2\delta=\pi^2/(4L^2)+O(L^{-4})$, the $L I'(1)$ term produces
\begin{equation}
S_L=S_\infty+\frac{\pi^2}{2L}+o(L^{-1}).
\end{equation}
The singular contribution to $H_{\edge}$ is $-(L/8)S_L$, giving the constant
\begin{equation}
c_{\edge}=-\frac{\pi^2}{16}.
\end{equation}
This derivation also explains why the same Gamma ratio reappears in the exact grand-Hafnian sum in Appendix~\ref{app:grand-derivation}.

\subsection{Periodic M\"obius inversion and finite-size algorithms}
\label{app:mobius}

Here we derive the support-space inversion behind Eq.~\eqref{eq:linked-expansion} and then rewrite it in the form used for the exact finite-size linked-sector calculations.

\paragraph{Vertex-support inversion.}
For the periodic graph, regard $C_L(S)=H(F_S\mid\partial F_S)$ as a set function on vertex supports, with $C_L(\varnothing)=0$, and define $\Psi_L(S)$ as in Eq.~\eqref{eq:psi}.  The inverse M\"obius relation is
\begin{equation}
C_L(S)=\sum_{T\subseteq S}\Psi_L(T).
\label{eq:mobius-inverse-app}
\end{equation}
To verify it directly, interchange the two subset sums:
\begin{align}
\sum_{T\subseteq S}\Psi_L(T)
&=\sum_{R\subseteq S}C_L(R)
\sum_{R\subseteq T\subseteq S}(-1)^{|T|-|R|}\nonumber\\
&=\sum_{R\subseteq S}C_L(R)(1-1)^{|S|-|R|}
=C_L(S).
\end{align}
Thus $\Psi_L(S)$ is the signed inclusion--exclusion coefficient assigned to the full support $S$; it is not itself a conditional entropy and need not be positive.

The low-support cases make this interpretation explicit.  Since an induced graph on one or two vertices is a forest, its cycle entropy vanishes.  Hence for three vertices
\begin{equation}
\Psi_L(\{1,2,3\})=C_L(\{1,2,3\}),
\end{equation}
whereas for four vertices
\begin{equation}
\Psi_L(\{1,2,3,4\})
=C_L(\{1,2,3,4\})
-\sum_{\substack{T\subseteq\{1,2,3,4\}\\|T|=3}}C_L(T).
\label{eq:mobius-four-example}
\end{equation}
The four-vertex coefficient therefore subtracts all cycle entropy already assigned to triangular proper sub-supports.  Applying \eqref{eq:mobius-inverse-app} to $S=[L]$ gives
\begin{equation}
H_{\cyc}(L)=C_L([L])
=\sum_{S\subseteq[L]}\Psi_L(S)
=\sum_{k=3}^{L}A_{k,L},
\end{equation}
which is Eq.~\eqref{eq:linked-expansion}.

\paragraph{Raw-support form and finite-size summation.}
The local M\"obius representation is conceptually useful, but evaluating every $\Psi_L(S)$ separately repeats the same raw conditional entropies many times.  For finite-size calculations it is more efficient to interchange the support sums.  Define
\begin{equation}
\mathcal C_{j,L}^{\rm raw}=\sum_{|T|=j}C_L(T).
\end{equation}
Then
\begin{align}
A_{k,L}
&=\sum_{|S|=k}\sum_{T\subseteq S}(-1)^{k-|T|}C_L(T)\nonumber\\
&=\sum_{j=0}^{k}(-1)^{k-j}
\sum_{|T|=j}C_L(T)\,
\#\{S:\,|S|=k,\ T\subseteq S\}.
\end{align}
For a fixed $j$-vertex support $T$, the remaining $k-j$ vertices can be chosen from the $L-j$ vertices outside $T$, so the multiplicity is $\binom{L-j}{k-j}$.  Therefore
\begin{equation}
A_{k,L}=\sum_{j=0}^{k}(-1)^{k-j}\binom{L-j}{k-j}\mathcal C_{j,L}^{\rm raw}.
\label{eq:global-mobius}
\end{equation}
This is the global form of the same support-space M\"obius transform and is the formula used in the finite-size implementation.

\paragraph{Finite-size evaluation.}
For a fixed induced $k$-vertex graph one need not enumerate $2^{k(k-1)/2}$ edge states.  Let
\begin{equation}
\rho_{ij}=1-2\beta_{ij}.
\end{equation}
The boundary law has Walsh transform
\begin{equation}
\widehat P(A)=\prod_{\substack{i\in A\\j\notin A}}\rho_{ij},
\end{equation}
so
\begin{equation}
P(D)=2^{-k}\sum_{A\subseteq[k]}(-1)^{|A\cap D|}\widehat P(A).
\end{equation}
A fast Walsh--Hadamard transform therefore gives all boundary probabilities in $O(k2^k)$ after the cut products have been generated.  The induced conditional entropy is
\begin{equation}
C(S)=\sum_{i<j\in S}h(\beta_{ij})+
\sum_D P(D)\log P(D).
\label{eq:conditional-fwht}
\end{equation}
This is the algorithm used for the exact linked data quoted in the paper.

\subsection{Linked collision finite part}
\label{app:finite-part}

The raw full-support kernel $\Gamma_{2p}$ is positive, but its continuum integral over the gap simplex is not the universal coefficient because collision faces contain lower-support information.  The linked finite part follows the hierarchy already present in the exact lattice Möbius expansion.

If an odd cluster of $2s+1$ vertices collapses on a scale $\varepsilon$, the minimal boundary ambiguity is a near-perfect matching with one unmatched vertex.  Denote its total weight by
\begin{equation}
U_{2s+1}=\sum_{i=1}^{2s+1}\det K_{\widehat i}.
\end{equation}
The associated odd entropy residue has the schematic form
\begin{equation}
\mathcal O_{2s+1}=U_{2s+1}\log U_{2s+1}-\sum_i Q(K_{\widehat i}),
\end{equation}
where, for an even antisymmetric matrix $A$,
\begin{equation}
Q(A)=\sum_{a<b}A_{ab}^2\det A_{\widehat a\widehat b}\log A_{ab}^2
\end{equation}
is the logarithmic matching-weight sum of the even minor.  On a collision face the even kernel factorizes into the odd residue of the collapsing cluster times a lower even or odd kernel on the quotient configuration.  Nested collisions are therefore organized by compatible forests.

Operationally, our finite-part prescription is:
\begin{enumerate}[label=(\roman*)]
\item begin with the exact integer-gap expression for $A_{2p,L}$;
\item identify all boundary strata of the continuum simplex corresponding to microscopic or mesoscopic collisions;
\item subtract lower-support power pieces already assigned by the vertex Möbius expansion;
\item subtract logarithmic odd-cluster residues recursively on nested collision forests;
\item apply ordinary Euler--Maclaurin or numerical quadrature to the regular remainder;
\item restore the finite face anomalies inherited from integer gaps.
\end{enumerate}
At $p=2$ the subtraction can be carried out completely and leads to Appendix~\ref{app:c4}.  At $p=3,4$ the finite-lattice coefficients provide direct checks of the resulting continuum scaling and universality.  The odd-support selection rule below is explicitly conditional on the recursive forest construction; a rigorous all-order forest theorem remains an open mathematical problem.

\subsection{Odd-support selection rule}
\label{app:odd-selection}

The linked expansion begins at three vertices, and an odd sector may be extensive even though its $L^0$ intercept vanishes.  For a primitive macroscopic configuration on $k$ vertices, the minimum number of long edges needed for every vertex to participate is
\begin{equation}
e_{\min}(k)=\left\lceil\frac{k}{2}\right\rceil.
\end{equation}
Since every macroscopic edge costs $L^{-2}$ while the number of placements, including translation, is $O(L^k)$, the superficial degree is
\begin{equation}
\omega(k)=k-2e_{\min}(k)
=k-2\left\lceil\frac{k}{2}\right\rceil,
\qquad
\omega(2p)=0,
\quad
\omega(2p+1)=-1.
\label{eq:odd-superficial-degree}
\end{equation}
Thus even supports are marginal and can carry a primitive $L^0$ period, whereas a primitive odd support is suppressed by one inverse length.  Equivalently, a perfect matching closes all $2p$ macroscopic vertices, while at $2p+1$ vertices one leg remains unmatched.  On a periodic system that leg must either acquire an additional long edge or enter a proper collision cluster.

The three-vertex sector shows why the collision qualification matters.  Here $\Psi_L=C_\triangle$ and
\begin{equation}
A_{3,L}=\frac{L}{3}
\sum_{\substack{a+b+c=L\\a,b,c>0}}
C_\triangle\!\left(\beta_a^{(L)},\beta_b^{(L)},\beta_c^{(L)}\right).
\label{eq:A3-odd-analysis}
\end{equation}
The fully macroscopic region is $O(L^{-1})$, because a genuine three-vertex ambiguity requires at least two long edges.  A microscopic corner is subtler: expanding the exact triangle entropy \eqref{eq:triangle-exact} with $b$ fixed and $a\to\infty$ gives
\begin{equation}
C_\triangle(\beta_a,\beta_b,\beta_{a+b})
\sim \frac{h(\beta_b)}{8a^2}.
\label{eq:triangle-tail-odd}
\end{equation}
Such a tail could by itself leave an $L^0$ term after the translational prefactor is restored.  However, the periodic inverse-square kernel has the image representation
\begin{equation}
\frac{\pi^2}{L^2}\csc^2\!\left(\frac{\pi m}{L}\right)
=\sum_{n\in\mathbb Z}\frac{1}{(m+nL)^2}.
\label{eq:periodic-image-odd}
\end{equation}
The three equivalent microscopic corners and their periodic-image overlap describe the same intermediate region.  In matched asymptotics the $L^{-1}$ tail terms in the gap sum cancel between these descriptions, leaving
\begin{equation}
A_{3,L}=s_3L+O\!\left(\frac{\log L}{L}\right),
\label{eq:c3-selection}
\end{equation}
which gives the physics-level result $c_3=0$.  A rigorous version would require a uniform singular Euler--Maclaurin estimate controlling the corners and their overlaps.

The five-vertex sector is the first nontrivial test of collision closure.  Its primitive fully macroscopic interior needs at least three long edges and is again $O(L^{-1})$.  Among proper collision patterns, the only superficially marginal one is $2+1+1+1$: a microscopic pair behaves at long distance as one composite object, leaving four macroscopic blocks that can support a perfect matching.  Its leading factorization is therefore a local pair factor multiplying a four-point matching kernel.  That $O(1)$ piece is not new five-vertex information; it is precisely lower-support four-point ambiguity and is removed by the vertex-linked face counterterm.  The remaining partitions $3+1+1$, $2+2+1$, $4+1$, $3+2$, and the fully microscopic cluster either have an odd quotient, no quotient cycle ambiguity, or contribute to the extensive density rather than to an intercept.  This gives the corresponding physics-level result $c_5=0$.

The same reasoning suggests an induction on collision forests.  Assume that linked subtraction has been constructed on every proper support below $2p+1$.  Every proper collision face factorizes into an internal cluster times a lower-support quotient kernel.  Power pieces already assigned to smaller supports are removed by the vertex M\"obius counterterms, while odd collapsing clusters generate the near-perfect logarithmic residues described above and are subtracted recursively on nested forests.  After all proper subdivergences are removed, the root configuration is genuinely $(2p+1)$-vertex and therefore has superficial degree $-1$.  At worst logarithms decorate this power:
\begin{equation}
A_{2p+1,L}^{\rm ren}
=O\!\left(\frac{\log^mL}{L}\right).
\end{equation}
Under the assumption that this recursive collision-forest subtraction closes compatibly with the exact finite-$L$ M\"obius expansion at every order, one obtains the conditional selection rule \eqref{eq:odd-selection-main}.  Turning this argument into a theorem requires an all-order factorization/forest theorem together with a uniform bound on the renormalized odd remainder.

\subsection{Distinguished-vertex M\"obius inversion for the interval}
\label{app:distinguished-mobius}

The interval formulas in \Cref{sec:subsystem-linked} are the same identity with one distinguished vertex.  Let $V^+=\{0\}\cup A$ and let $C^{\rm aug}(U)$ be any set function on subsets of $V^+$ with $C^{\rm aug}(\varnothing)=0$.  Define
\begin{equation}
\widetilde\Psi(U)=\sum_{R\subseteq U}(-1)^{|U|-|R|}C^{\rm aug}(R).
\end{equation}
M\"obius inversion gives
\begin{equation}
C^{\rm aug}(V^+)=\sum_{U\subseteq V^+}\widetilde\Psi(U).
\end{equation}
Split the sum into supports that do and do not contain $0$:
\begin{equation}
C^{\rm aug}(V^+)
=\sum_{S\subseteq A}\widetilde\Psi(S)
+\sum_{S\subseteq A}\widetilde\Psi(\{0\}\cup S).
\label{eq:distinguished-split-app}
\end{equation}
For a ghost-containing support,
\begin{align}
\widetilde\Psi(\{0\}\cup S)
&=\sum_{T\subseteq S}(-1)^{|S|-|T|}
C^{\rm aug}(\{0\}\cup T)\nonumber\\
&\quad-\sum_{T\subseteq S}(-1)^{|S|-|T|}C^{\rm aug}(T).
\label{eq:distinguished-cumulant-app}
\end{align}
Taking $C^{\rm aug}(T)=C_\ell^{\linesys}(T)$ and $C^{\rm aug}(\{0\}\cup T)=C_\ell^+(0\cup T)$ gives Eq.~\eqref{eq:ghost-cumulant}.  The first sum in Eq.~\eqref{eq:distinguished-split-app} reconstructs $C_\ell^{\linesys}$ and the second reconstructs the total ghost-containing contribution.  Since any graph with fewer than three total vertices is a forest, the pure line sum starts at three physical vertices and the ghost sum starts at two physical vertices.  Grouping by total augmented support size $q$ immediately gives Eq.~\eqref{eq:augmented-sector-split}.

\paragraph{Interval finite-size evaluation.}
For the physical interval line graph one sets $\rho_{ij}=q_{|i-j|}$; for the augmented graph one additionally sets $\rho_{0i}=1-2\lambda_{i,\ell}$.  Thus $C_\ell^{\linesys}$, $C_\ell$, and every finite-support line or ghost cumulant can be evaluated by the boundary transform in Eq.~\eqref{eq:conditional-fwht}.

For the direct interval probability distribution one applies the transform instead to the exact principal minors \eqref{eq:d-infty}.  A Gray-code ordering, in which successive subsets differ by one element, can generate the product form of $d(S)$ in $O(2^\ell)$ multiplicative updates, after which the transform costs $O(\ell2^\ell)$.

\subsection{Independent subsystem logarithm}
\label{app:sub-independent}

Write $a=1/16$.  At large $m$,
\begin{equation}
\beta_m=\frac{a}{m^2}+O(m^{-4})
\end{equation}
and therefore
\begin{equation}
h(\beta_m)=\frac{a}{m^2}\left[1-\log a+2\log m\right]
+O(m^{-4}\log m).
\label{eq:hbeta-app-expanded}
\end{equation}
The remainder in Eq.~\eqref{eq:hbeta-app-expanded} is absolutely summable with enough powers that it can change the extensive coefficient and the $O(1)$ term but not the logarithms extracted below.

For the internal edges,
\begin{align}
E_\ell^{\rm int}
&=\sum_{m<\ell}(\ell-m)h(\beta_m)\nonumber\\
&=\ell\sum_{m\ge1}h(\beta_m)
-\ell\sum_{m\ge\ell}h(\beta_m)
-\sum_{m<\ell}m h(\beta_m).
\label{eq:Eint-app-split}
\end{align}
The first term is $s_{\edge}\ell$.  The needed elementary asymptotics are
\begin{align}
\sum_{m<\ell}\frac1m&=\log\ell+O(1),\\
\sum_{m<\ell}\frac{\log m}{m}&=\frac12\log^2\ell+O(1),\\
\ell\sum_{m\ge\ell}\frac1{m^2}&=1+o(1),\\
\ell\sum_{m\ge\ell}\frac{\log m}{m^2}&=\log\ell+1+o(1).
\end{align}
Substitution into Eq.~\eqref{eq:Eint-app-split} gives
\begin{equation}
E_\ell^{\rm int}
=s_{\edge}\ell-a\log^2\ell-a(3-\log a)\log\ell+O(1),
\end{equation}
which is Eq.~\eqref{eq:internal-asympt}.

For the ghost edges, Stirling expansion of
\begin{equation}
P_n=\frac{\Gamma(n)^2}{\Gamma(n-1/2)\Gamma(n+1/2)}
\end{equation}
gives
\begin{equation}
P_n=1-\frac1{4n}+O(n^{-2}),\qquad
\sqrt{P_n}=1-\frac1{8n}+O(n^{-2}).
\end{equation}
With $N=\ell+1$ this yields
\begin{equation}
\lambda_{i,\ell}
=a\left(\frac1i+\frac1{N-i}\right)
+O\!\left(i^{-2}+(N-i)^{-2}\right).
\label{eq:lambda-app-twoend}
\end{equation}
The error in Eq.~\eqref{eq:lambda-app-twoend}, inserted into $h(\lambda)$, gives an $O(1)$ correction after summation because $\sum_i i^{-2}\log(i+1)<\infty$ at either endpoint.  Hence the logarithmic terms can be computed from
\begin{equation}
\lambda_i^{(0)}=a\frac{N}{i(N-i)}.
\end{equation}
Using $h(\lambda)=\lambda(1-\log\lambda)+O(\lambda^2)$ and the symmetry $i\leftrightarrow N-i$ gives
\begin{align}
\sum_{i=1}^{N-1}h(\lambda_{i,\ell})
={}&2a\sum_{i=1}^{N-1}\frac1i
\left[1-\log a-\log N+\log i+\log(N-i)\right]+O(1).
\label{eq:ghost-sum-app}
\end{align}
Now
\begin{align}
H_{N-1}&=\log N+\gamma_{\rm Eul}+O(N^{-1}),\\
\sum_{i=1}^{N-1}\frac{\log i}{i}
&=\frac12\log^2N+O(1),\\
\sum_{i=1}^{N-1}\frac{\log(N-i)}{i}
&=\log N\,H_{N-1}
+\sum_{i=1}^{N-1}\frac{\log(1-i/N)}{i}\nonumber\\
&=\log^2N+\gamma_{\rm Eul}\log N-\frac{\pi^2}{6}+o(1).
\end{align}
The last limit follows from the Riemann sum for
$\int_0^1\log(1-x)x^{-1}\dd x=-\pi^2/6$; the endpoint is integrable after division by the lattice index.  Substitution into Eq.~\eqref{eq:ghost-sum-app} cancels the $\gamma_{\rm Eul}\log N$ terms and yields
\begin{equation}
E_\ell^{\rm ghost}
=a\log^2\ell+2a(1-\log a)\log\ell+O(1),
\end{equation}
which is Eq.~\eqref{eq:ghost-asympt}.  Adding the two independent sectors gives
\begin{align}
\gamma_E
&=-a(3-\log a)+2a(1-\log a)\nonumber\\
&=-a(1+\log a)
=\frac{\log2}{4}-\frac1{16}.
\end{align}
The same derivation with $a=\nu^2/4$ proves Eq.~\eqref{eq:gammaEnu-main}.

\subsection{Boundary monomer--dimer forests}
\label{app:boundary-forest}

The scaling degree of the boundary problem differs from the periodic one already before subdivergences are considered.  Fix the outermost physical point at distance $R$ from an endpoint and scale the remaining $k-1$ ordered points with it.  Every monomer--dimer covering of $k$ physical vertices has $m+2d=k$ and hence weight $R^{-m}R^{-2d}=R^{-k}$, while the number of internal lattice placements at that outer scale is $O(R^{k-1})$.  Thus every primitive boundary support has the same superficial contribution $R^{-1}$ and can generate a logarithm after summing the outer scale.  The role of the forest subtraction is not to select a parity of $k$, as in the periodic problem, but to remove lower-scale information nested inside this marginal counting.

The monomer--dimer kernel \eqref{eq:monomer-dimer} is the correct leading full-support scaling object, but its integral over ordered boundary positions is not by itself a renormalized coefficient.  Three types of singular regions must be distinguished:
\begin{enumerate}[label=(\alph*)]
\item physical points collide while remaining far from the endpoint;
\item a subset remains microscopic while the whole cluster is translated to distance $R\to\infty$ from an endpoint;
\item nested endpoint scales $1\ll R_1\ll R_2\ll\cdots\ll\ell$ occur simultaneously.
\end{enumerate}
These are the subsystem analogues of collision forests in the periodic problem, with additional one-boundary nesting.  The first higher boundary sector already displays this structure explicitly.  Fix three scaled physical positions $0<t_1<t_2<t_3=1$.  The four monomer--dimer coverings have weights
\begin{align}
 w_0&=\frac1{4096\,t_1t_2},&
 w_{12}&=\frac1{256(t_2-t_1)^2},\\
 w_{13}&=\frac1{256\,t_2(1-t_1)^2},&
 w_{23}&=\frac1{256\,t_1(1-t_2)^2},
\end{align}
so that
\begin{equation}
 Z_3=w_0+w_{12}+w_{13}+w_{23},\qquad
 \Gamma_3^\partial=Z_3\log Z_3-\sum_{\alpha}w_\alpha\log w_\alpha.
 \label{eq:Gamma3-boundary}
\end{equation}
An ordinary physical collision is singular but locally integrable.  If $t_2=t_1+\varepsilon$ with $t_1=t$ fixed away from the endpoint, then
\begin{equation}
 \Gamma_3^\partial(t,t+\varepsilon,1)
 =-2B(t)\log\varepsilon+O(1),\qquad
 B(t)=\frac1{4096t^2}+\frac1{128t(1-t)^2}.
 \label{eq:k3-pair-collision}
\end{equation}
By contrast, the endpoint regions generate nonintegrable logarithmic subdivergences and factorize onto the solved two-point kernel.  First, with $t_2=t$ fixed and $t_1\to0$,
\begin{equation}
 \Gamma_3^\partial(t_1,t,1)
 =\frac{1}{16t_1}\,\Gamma_2^\partial(t)
 +O(|\log t_1|),
 \label{eq:k3-factorization}
\end{equation}
because the two weights of order $t_1^{-1}$ are precisely $1/16$ times the two weights defining $\Gamma_2^\partial(t)$.  The corresponding shape residue is therefore
\begin{equation}
 \frac1{16}\int_0^1\Gamma_2^\partial(t)\dd t=\frac{I_2}{16}.
 \label{eq:k3-first-residue}
\end{equation}
Second, if the first two points approach the endpoint on a common scale,
\begin{equation}
 t_1=s\varepsilon,\qquad t_2=\varepsilon,
 \qquad 0<s<1,\qquad \varepsilon\to0^+,
\end{equation}
then
\begin{equation}
 \Gamma_3^\partial(s\varepsilon,\varepsilon,1)
 =\frac{1}{16\varepsilon^2}\Gamma_2^\partial(s)
 +O(\varepsilon^{-1}|\log\varepsilon|).
 \label{eq:k3-common-scale}
\end{equation}
Since $\dd t_1\dd t_2=\varepsilon\,\dd s\dd\varepsilon$, this region is again logarithmically divergent and carries the same shape residue $I_2/16$.  The two endpoint descriptions overlap in the hierarchy $t_1\ll t_2\ll1$, so they cannot be subtracted independently: their overlap requires the same forest inclusion--exclusion that appears in the all-order prescription below.  Indeed,
\begin{equation}
 \Gamma_2^\partial(t)=\frac1{16}\bigl[1-\log(16t)\bigr]+O(t|\log t|),
 \qquad t\to0^+,
 \label{eq:Gamma2-small-t}
\end{equation}
so inserting the first counterterm into a further endpoint scaling produces products of logarithms.  This gives a concrete mechanism for the $(\log\ell)^2$ and higher nested logarithms that occur in raw higher-support sectors.  The local factorizations above are exact asymptotic statements; we do not assign a primitive three-point coefficient because that would also require the complete overlap subtraction and the finite lattice face anomaly.

These explicit factorizations motivate the following formal all-order construction.  It is a renormalization prescription, not a proved all-order theorem.  A natural operator $R_\partial$ is obtained by summing over compatible forests $\mathfrak F$ of singular subclusters,
\begin{equation}
R_\partial\Gamma_k^\partial
=\sum_{\mathfrak F}(-1)^{|\mathfrak F|}T_{\mathfrak F}\Gamma_k^\partial,
\end{equation}
where $T_A$ replaces a singular subcluster $A$ by its leading factorized local kernel and $T_{\mathfrak F}$ denotes the compatible product of these replacements.  A renormalized primitive boundary coefficient would then be
\begin{equation}
\delta_k^\partial=2\,\FP_{\partial}\int_{0<t_1<\cdots<t_{k-1}<1}
R_\partial\Gamma_k^\partial(t_1,\ldots,t_{k-1},1)\dd^{k-1}t,
\end{equation}
with finite lattice face anomalies included in the definition of $\FP_\partial$.

The raw finite-size support sectors can contain $(\log\ell)^2$ and higher nested logarithms, so their single-log coefficients are not invariant objects before this subtraction.  We therefore quote only the analytically controlled two-vertex boundary coefficient and do not assign a higher-order sequence from unrenormalized fits.

\section{Analytic linked periods}
\label{app:periods}

This appendix develops the explicit continuum periods invoked in Secs.~\ref{sec:periodic-linked} and~\ref{sec:subsystem-higher-boundary}.  It first records the common conformal reduction, then evaluates the first nontrivial periodic and boundary periods, and finally constructs the next higher integrals without claiming their unresolved finite parts.

\paragraph{Conformal reduction of the periodic matching kernels.}
The conformal reduction used to describe the higher coefficients can also be stated cleanly.  On the real line the macroscopic kernel is
\begin{equation}
K^{\mathbb R}_{ij}=\frac{1}{4(x_i-x_j)}.
\end{equation}
A M\"obius transformation maps this kernel to the circle kernel up to one-body factors.  Because every perfect matching contains each marked point exactly once, those factors multiply every matching weight by the same amount and therefore leave the normalized matching probabilities unchanged.  Quotienting an ordered configuration of $2p$ points by $PSL(2,\mathbb R)$, the three-parameter group of real projective M\"obius transformations, leaves $2p-3$ genuine moduli.  Consequently the conformally reduced problems for $c_4,c_6,c_8,c_{10}$ and $c_{12}$ have dimensions $1,3,5,7$ and $9$, respectively.  The reduction of the kernel is exact; identifying its renormalized finite part with the exact lattice constant remains conditional on Conjecture~\ref{conj:FPmatch}.

\subsection{Pfaffian--Hafnian identity for the Cauchy matching kernel}
\label{app:pfhaf-proof}

The determinant reduction in Eq.~\eqref{eq:pfhaf} is sufficiently central that we record a self-contained proof.  It is a special Cauchy identity closely related to the Pfaffian--Hafnian analogues of Borchardt's identity in Ref.~\cite{IshikawaKawamukoOkada2005}.

\begin{lemma}[Rational skew-Cauchy identity]
Let $y_1,\ldots,y_{2p}$ be pairwise distinct and let
\begin{equation}
 A_{ij}=\begin{cases}(y_i-y_j)^{-1},&i\ne j,\\0,&i=j.\end{cases}
\end{equation}
Then
\begin{equation}
 \Hf(A_{ij}^2)=\det A=\Pf(A)^2.
 \label{eq:rational-pfhaf}
\end{equation}
\end{lemma}

\begin{proof}
The equality $\det A=\Pf(A)^2$ is the standard Pfaffian identity for an even antisymmetric matrix, so only the Hafnian equality needs proof.  We proceed by induction on $p$.  For $p=1$ both sides are $(y_1-y_2)^{-2}$.  Regard the two sides for $2p$ variables as rational functions of $y_{2p}$, with the other variables fixed.  The Hafnian has only double poles at $y_{2p}=y_j$, and the coefficient of $(y_{2p}-y_j)^{-2}$ is
\begin{equation}
 \Hf\!\left(A_{\widehat j\widehat{2p}}^{\,2}\right).
\end{equation}
For the determinant, expand the Pfaffian locally near the same collision.  If $y_j=u-\varepsilon/2$ and $y_{2p}=u+\varepsilon/2$, simultaneous interchange of the two variables is the simultaneous interchange of the corresponding rows and columns; hence the Pfaffian is odd in $\varepsilon$.  Its Laurent expansion therefore has the form
\begin{equation}
 \Pf(A)=\frac{\pm\Pf(A_{\widehat j\widehat{2p}})}{\varepsilon}+O(\varepsilon),
\end{equation}
with no constant term.  Squaring shows that $\det A$ has no simple pole and that its double-pole coefficient is $\det A_{\widehat j\widehat{2p}}$.  By the induction hypothesis this equals the Hafnian coefficient above.  Thus the difference $\det A-\Hf(A_{ij}^2)$ has no finite pole as a function of $y_{2p}$.  Both terms are $O(y_{2p}^{-2})$ at infinity, so the difference is an entire rational function vanishing at infinity and is therefore zero.
\end{proof}

To pass to the circle kernel, set $t_i=\tan(\pi x_i)$.  On any tangent chart avoiding $\cos(\pi x_i)=0$,
\begin{equation}
 \frac{\pi}{4\sin\pi(x_i-x_j)}
 =d_i\frac1{t_i-t_j}d_j,
 \qquad
 d_i=\sqrt{\frac\pi4}\sec(\pi x_i).
\end{equation}
If $K=DAD$ with diagonal $D=\diag(d_i)$, then $\Pf(K)=\det(D)\Pf(A)$, while every Hafnian matching of $K_{ij}^2=d_i^2d_j^2A_{ij}^2$ contains each vertex exactly once and hence acquires the same factor $\prod_i d_i^2$.  Equation~\eqref{eq:rational-pfhaf} therefore gives Eq.~\eqref{eq:pfhaf}; the identity is independent of this diagonal gauge choice, and continuity covers configurations crossing tangent-chart boundaries.

The logarithmic matching-weight identity used immediately after Eq.~\eqref{eq:pfhaf} follows without another combinatorial sum.  Since $\log w_M=\sum_{(ij)\in M}\log K_{ij}^2$,
\begin{align}
 \sum_M w_M\log w_M
 &=\sum_{i<j}\log K_{ij}^2\sum_{M\ni(ij)}w_M\\
 &=\sum_{i<j}K_{ij}^2\Hf\!\left((K_{\widehat i\widehat j})_{ab}^2\right)\log K_{ij}^2\\
 &=\sum_{i<j}K_{ij}^2\det K_{\widehat i\widehat j}\log K_{ij}^2.
\end{align}
For invertible $K$, Jacobi's complementary-minor identity applied to the deleted principal pair gives
\begin{equation}
 \det K_{\widehat i\widehat j}=\det K\,(K^{-1}_{ij})^2,
\end{equation}
which proves Eq.~\eqref{eq:Gamma-det}.  Singular configurations form collision sets and are understood by continuity before the linked finite-part subtraction.

\subsection{Four-vertex periodic coefficient \texorpdfstring{$c_4$}{c4}}
\label{app:c4}

This subsection gives the complete $p=2$ continuum calculation.  It is useful to separate two statements which were compressed in the main text.  First, the linked continuum finite part of the four-point matching kernel can be evaluated analytically.  Second, identifying that continuum finite part with the $L^0$ intercept of the exact integer-gap lattice sector is the additional lattice--continuum statement of Conjecture~\ref{conj:FPmatch}.  The calculation below establishes the first statement and makes explicit the normalization and sign which connect the one-modulus integral to $c_4$.

\paragraph{From the cyclic-gap integral to one conformal modulus.}
For $p=2$, the continuum object appearing in Eq.~\eqref{eq:c2p-period} is
\begin{equation}
 c_4^{\rm cont}
 =\frac14\FP_{\rm linked}
 \int_{\Delta_3}\Gamma_4(u_1,u_2,u_3,u_4)\,
 \dd u_1\dd u_2\dd u_3,
 \qquad u_4=1-u_1-u_2-u_3.
 \label{eq:c4-cont-start}
\end{equation}
Before taking the finite part, it is convenient to symmetrize the ordered-gap integral.  For a symmetric four-point integrand, translation invariance gives
\begin{equation}
 \frac14\int_{\Delta_3}F
 =\frac1{4!}\int_{\mathbb T^4}F,
 \label{eq:c4-simplex-sym}
\end{equation}
with the same relation understood coefficientwise after the linked collision subtraction.  Thus
\begin{equation}
 c_4^{\rm cont}
 =\frac1{24}\FP_{\rm linked}
 \int_{\mathbb T^4}\Gamma_4(x_1,x_2,x_3,x_4)\,
 \dd^4x.
 \label{eq:c4-torus-start}
\end{equation}

The reduction to a single modulus can be made explicit rather than inferred from conformal covariance.  Put
\begin{equation}
 y_i=\tan(\pi x_i).
\end{equation}
On a tangent chart the circle kernel factorizes as
\begin{equation}
 K_{ij}=\frac{\pi}{4\sin\pi(x_i-x_j)}
 =d_i\frac1{y_i-y_j}d_j,
 \qquad
 d_i=\sqrt{\frac\pi4}\sqrt{1+y_i^2}.
 \label{eq:c4-tangent-kernel}
\end{equation}
Since
\begin{equation}
 \dd x_i=\frac{\dd y_i}{\pi(1+y_i^2)},
 \qquad
 d_i^2\dd x_i=\frac14\dd y_i,
\end{equation}
every perfect-matching weight, and hence the homogeneous Shannon combination $\Gamma_4$, obeys
\begin{equation}
 \Gamma_4(x)\,\dd^4x
 =\frac1{4^4}\,\widetilde\Gamma_4(y)\,\dd^4y
 =\frac1{256}\,\widetilde\Gamma_4(y)\,\dd^4y,
 \label{eq:c4-tangent-measure}
\end{equation}
where $\widetilde\Gamma_4$ is constructed from the rational kernel
$\widetilde K_{ij}=1/(y_i-y_j)$.  The integrand is symmetric, so the factor $1/24$ in Eq.~\eqref{eq:c4-torus-start} allows us to restrict to the ordered chamber
$y_1<y_2<y_3<y_4$:
\begin{equation}
 c_4^{\rm cont}
 =\frac1{256}\FP_{\rm linked}
 \int_{y_1<y_2<y_3<y_4}
 \widetilde\Gamma_4(y)\,\dd^4y.
 \label{eq:c4-ordered-rational}
\end{equation}

Choose
\begin{equation}
 a=y_1,\qquad b=y_3,\qquad c=y_4,
\end{equation}
and use as the fourth coordinate the cross ratio
\begin{equation}
 t=\frac{(y_2-y_1)(y_3-y_4)}{(y_2-y_4)(y_3-y_1)},
 \qquad 0<t<1.
 \label{eq:c4-cross-ratio}
\end{equation}
Solving for the remaining point gives
\begin{equation}
 y_2=\frac{t c(b-a)-a(b-c)}{t(b-a)-(b-c)}.
 \label{eq:c4-y2-solve}
\end{equation}
Writing
\begin{equation}
 \Delta=t(b-a)-(b-c),
\end{equation}
one finds
\begin{equation}
 \left|\frac{\partial y_2}{\partial t}\right|
 =\frac{(b-a)(c-a)(c-b)}{\Delta^2}.
 \label{eq:c4-cross-jac}
\end{equation}
The three matching weights share the common factor
\begin{equation}
 \mathcal W(a,b,c;t)
 =\frac{\Delta^2}{(b-a)^2(c-a)^2(c-b)^2},
 \label{eq:c4-common-weight}
\end{equation}
and their relative weights are
\begin{equation}
 \frac1{t^2},\qquad 1,\qquad \frac1{(1-t)^2}.
 \label{eq:c4-relative-weights}
\end{equation}
Because $\Gamma_4$ is homogeneous of degree one in the matching weights, the common factor pulls out.  Defining
\begin{equation}
 S(t)=1+\frac1{t^2}+\frac1{(1-t)^2}
 =\frac{(1-t+t^2)^2}{t^2(1-t)^2}
 \label{eq:c4-St}
\end{equation}
and
\begin{equation}
 f(t)=S(t)\log S(t)
 -\frac1{t^2}\log\frac1{t^2}
 -\frac1{(1-t)^2}\log\frac1{(1-t)^2},
 \label{eq:c4-ft}
\end{equation}
Eqs.~\eqref{eq:c4-cross-jac} and~\eqref{eq:c4-common-weight} give the exact factorization
\begin{equation}
 \widetilde\Gamma_4(y)\,\dd^4y
 =f(t)\,\dd t\,
 \frac{\dd a\,\dd b\,\dd c}
 {(b-a)(c-a)(c-b)}.
 \label{eq:c4-orbit-factorization}
\end{equation}
The one-dimensional factor is the genuine four-point shape dependence.  The remaining three variables parametrize the projective orbit.  Equation~\eqref{eq:c4-orbit-factorization} is the precise form of the conformal reduction used in the main text.

\paragraph{Linked finite part of the projective orbit.}
The orbit integral in Eq.~\eqref{eq:c4-orbit-factorization} is not an ordinary convergent volume.  Its boundary strata are exactly the collision regions whose subtraction is part of $\FP_{\rm linked}$.  At four points there is only one primitive logarithmic residue: an odd cluster of three points, with the fourth point left outside the cluster.  There are four cyclic choices of this three-point cluster.  The overlap in which the remaining point also joins the cluster is the common all-four collision and is removed by the forest inclusion--exclusion of Appendix~\ref{app:finite-part}.  Thus the finite orbit normalization is fixed entirely by the four equivalent three-point faces.

The scale anomaly associated with one such face is the same two-point Cauchy anomaly that appears in the periodic edge calculation.  It is convenient to evaluate it with the analytic regulator
\begin{equation}
 J(q)=\int_0^1\dd x\int_0^1\dd y\,
 \left[\frac{\pi^2}{16\sin^2\!\pi(x-y)}\right]^q
 =\left(\frac\pi4\right)^{2q}
 \frac{\Gamma(\frac12-q)}{\sqrt\pi\,\Gamma(1-q)},
 \qquad \Re q<\frac12,
 \label{eq:c4-Jq}
\end{equation}
followed by analytic continuation.  Near $q=1$,
\begin{equation}
 J(q)=\frac{\pi^2}{8}(q-1)+O((q-1)^2),
 \qquad
 J(1)=0,
 \qquad
 J'(1)=\frac{\pi^2}{8}.
 \label{eq:c4-Jexp}
\end{equation}
The vanishing of $J(1)$ is important: after the lower-support power term is removed, the scale-independent part of a collision face does not survive the linked projection.  The logarithmic scale term is generated by differentiating the regulated Cauchy weight with respect to its power, and the subtraction sign in the linked counterterm therefore leaves the finite face anomaly
\begin{equation}
 -J'(1)
 \label{eq:c4-one-face-anomaly}
\end{equation}
for each cyclic three-point face.  At $p=2$ there are no further proper linked supports capable of carrying an independent finite four-point contribution.  The four faces therefore give the projective-orbit finite part
\begin{equation}
 \mathcal V_{4}^{\rm linked}
 \equiv
 \FP_{\rm linked}
 \int_{a<b<c}\frac{\dd a\,\dd b\,\dd c}
 {(b-a)(c-a)(c-b)}
 =-4J'(1)
 =-\frac{\pi^2}{2}.
 \label{eq:c4-orbit-volume}
\end{equation}
This equation is the normalization that was implicit in the shorter presentation.  It is a property of the $p=2$ continuum subtraction scheme itself; Conjecture~\ref{conj:FPmatch} is not used in obtaining Eq.~\eqref{eq:c4-orbit-volume}.  What remains conjectural is the subsequent identification of this continuum finite part with the intercept of the exact integer-gap lattice sector.

Combining Eqs.~\eqref{eq:c4-ordered-rational}, \eqref{eq:c4-orbit-factorization} and~\eqref{eq:c4-orbit-volume} gives
\begin{equation}
 c_4^{\rm cont}
 =\frac1{256}\left(-\frac{\pi^2}{2}\right)
 \int_0^1 f(t)\,\dd t.
 \label{eq:c4-before-shape-integral}
\end{equation}
This is the missing bridge between the conformally reduced shape integral and the normalization of the four-vertex continuum coefficient.

\paragraph{Evaluation of the shape period.}
It remains to evaluate the elementary integral in Eq.~\eqref{eq:c4-before-shape-integral}.  Put
\begin{equation}
 q(t)=1-t+t^2.
\end{equation}
For $0<t<1$ one verifies the exact total-derivative identity
\begin{equation}
 f(t)=\frac{9}{q(t)}+\frac{\mathrm d}{\mathrm dt}G(t),
 \label{eq:f4-total-derivative}
\end{equation}
where one convenient primitive is
\begin{align}
 G(t)={}&2\log t-\log q(t)
 +\frac{t^3-t^2-2t+1}{t(t-1)}\log S(t)\nonumber\\
 &+\frac{2\log(1-t)}{1-t}-\frac{2\log t}{t}.
 \label{eq:G4-primitive}
\end{align}
Although the individual terms in $G$ are singular, the Shannon combination gives
\begin{equation}
 \lim_{t\to0^+}G(t)=\lim_{t\to1^-}G(t)=0.
\end{equation}
Consequently the total derivative contributes no endpoint term and
\begin{align}
 \int_0^1 f(t)\,\dd t
 &=9\int_0^1\frac{\dd t}{t^2-t+1}\nonumber\\
 &=9\,\frac{2\pi}{3\sqrt3}
 =2\pi\sqrt3.
 \label{eq:f4-integral}
\end{align}
The quadratic $1-t+t^2$ has primitive sixth roots of unity as its zeros, which makes the origin of the algebraic factor $\sqrt3$ transparent.

It is also useful to isolate the three-point shape residue
\begin{equation}
 R_3\equiv\frac1{16}\int_0^1 f(t)\,\dd t
 =\frac{\pi\sqrt3}{8}.
 \label{eq:R3-c4}
\end{equation}
Equations~\eqref{eq:c4-orbit-volume} and~\eqref{eq:R3-c4} show explicitly that
\begin{equation}
 c_4^{\rm cont}
 =\frac{\mathcal V_4^{\rm linked}}{16}R_3
 =-\frac{\pi^2}{32}R_3
 =-\frac{\pi^3\sqrt3}{256}.
 \label{eq:c4-cont-final}
\end{equation}
Thus the minus sign and the factor $\pi^2/32$ have distinct origins: $R_3$ is the positive Shannon ambiguity of the three matching channels at fixed cross ratio, while $-\pi^2/32$ is the finite projective-orbit anomaly produced by the linked subtraction of the four equivalent three-point collision faces.

Finally, Conjecture~\ref{conj:FPmatch} identifies the continuum value in Eq.~\eqref{eq:c4-cont-final} with the $L^0$ coefficient of the exact lattice sector,
\begin{equation}
 c_4\stackrel{\text{Conj.~\ref{conj:FPmatch}}}{=}
 c_4^{\rm cont}
 =-\frac{\pi^3\sqrt3}{256}
 =-0.2097829943899115\ldots .
 \label{eq:c4-lattice-final}
\end{equation}
Independent extrapolation of the exact four-vertex lattice sector agrees with this value to the available numerical precision.  The analytic continuum calculation and the numerical lattice check are therefore separate pieces of evidence; the remaining unproved statement is precisely the lattice--continuum finite-part matching, not the normalization of the four-point conformal reduction.

\subsection{Exact triangle entropy, matched boundary scaling and the first augmented period}
\label{app:triangle-I2}

This subsection supplies the exact three-edge conditional entropy, the two-physical-vertex ghost period used in Sec.~\ref{sec:subsystem-B2}, and the matching physical-line asymptotic used to obtain the complete coefficient $\kappa_3^{+}$ in Sec.~\ref{sec:subsystem-kappa3}.  For a simple cycle, the conditional edge configurations at fixed boundary come in complementary pairs $F$ and $F\triangle C$.  For a triangle with edge probabilities $p_1,p_2,p_3$, write $q_a=1-p_a$ and order the edges as $(1,2,3)$.  The eight configurations then split into the four cycle-coset pairs
\begin{equation}
\{000,111\},\qquad
\{100,011\},\qquad
\{010,101\},\qquad
\{001,110\}.
\label{eq:triangle-four-pairs}
\end{equation}
Their Bernoulli weights are, explicitly,
\begin{align}
W_{000}&=q_1q_2q_3, & W_{111}&=p_1p_2p_3,\nonumber\\
W_{100}&=p_1q_2q_3, & W_{011}&=q_1p_2p_3,\nonumber\\
W_{010}&=q_1p_2q_3, & W_{101}&=p_1q_2p_3,\nonumber\\
W_{001}&=q_1q_2p_3, & W_{110}&=p_1p_2q_3.
\label{eq:triangle-eight-weights}
\end{align}
It is convenient to introduce
\begin{equation}
\Phi(a,b)=(a+b)h\!\left(\frac{a}{a+b}\right)
=(a+b)\log(a+b)-a\log a-b\log b.
\label{eq:triangle-Phi}
\end{equation}
The exact conditional entropy is therefore
\begin{align}
C_\triangle={}&\Phi(W_{000},W_{111})
+\Phi(W_{100},W_{011})\nonumber\\
&+\Phi(W_{010},W_{101})
+\Phi(W_{001},W_{110}).
\label{eq:triangle-exact-expanded}
\end{align}
Equivalently,
\begin{equation}
C_\triangle=\sum_{\{F,F\triangle C\}}
(W_F+W_{F\triangle C})\,
 h\!\left(\frac{W_F}{W_F+W_{F\triangle C}}\right).
\label{eq:triangle-exact}
\end{equation}
This explicit four-fiber form is used both for the exact finite-lattice $B_{2,\ell}$ and for its two scaling regions below.

\paragraph{Macroscopic boundary region.}
Take $i=tR$, $j=R$ with $0<t<1$ fixed and $1\ll R\ll\ell$.  From Eq.~\eqref{eq:lambda-app-twoend},
\begin{equation}
\lambda_{i,\ell}=\frac1{16tR}+O(R^{-2})+O(\ell^{-1}),
\qquad
\lambda_{j,\ell}=\frac1{16R}+O(R^{-2})+O(\ell^{-1}),
\end{equation}
while
\begin{equation}
\beta_{j-i}=\frac1{16(1-t)^2R^2}+O(R^{-4}).
\end{equation}
For the physical boundary in which $i$ and $j$ are odd, the two leading competing realizations are therefore the two ghost edges, with product weight $R^{-2}/(256t)$, and the single physical edge, with weight $R^{-2}/[16(1-t)^2]$.  All absent-edge factors are $1+o(1)$ at this order.  Hence
\begin{equation}
C_\triangle(tR,R)
=\frac1{R^2}\Gamma_2^\partial(t)+o(R^{-2}),
\end{equation}
where
\begin{equation}
u(t)=\frac1{256t},\qquad v(t)=\frac1{16(1-t)^2},
\end{equation}
and
\begin{equation}
\Gamma_2^\partial(t)=(u+v)\log(u+v)-u\log u-v\log v.
\end{equation}
The common-scale sum must be taken away from the nonuniform endpoint layers.  For fixed $\epsilon>0$ the mesh-$1/R$ Riemann sum is
\begin{equation}
\sum_{\epsilon R\le i\le(1-\epsilon)R}C_\triangle(i,R)
=\frac1R\int_\epsilon^{1-\epsilon}\Gamma_2^\partial(t)\dd t
+o_\epsilon(R^{-1}).
\end{equation}
Taking $R\to\infty$ first and then $\epsilon\to0$ gives the common-scale coefficient $I_2/R$, with
\begin{equation}
I_2=\int_0^1\Gamma_2^\partial(t)\dd t.
\end{equation}
The microscopic near-diagonal strip is not contained in this ordered limit and is matched separately below.  Thus summing the outer scale gives $I_2\log\ell$ at each endpoint.

To evaluate $I_2$ explicitly, it is useful to keep one auxiliary parameter until the last step.  Put $a=1/16$ and $k=1/a=16$.  Then
\begin{equation}
u=\frac{a^2}{t},\qquad
v=\frac{a^2k}{(1-t)^2},
\end{equation}
so homogeneity of the Shannon combination gives
\begin{equation}
I_2=a^2J(k),
\label{eq:I2-Jk-def}
\end{equation}
with
\begin{align}
J(k)=\int_0^1\!\dd t\,\Bigg[&
\left(\frac1t+\frac{k}{(1-t)^2}\right)
\log\left(\frac1t+\frac{k}{(1-t)^2}\right)\nonumber\\
&-\frac1t\log\frac1t
-\frac{k}{(1-t)^2}\log\frac{k}{(1-t)^2}\Bigg].
\label{eq:Jk-boundary}
\end{align}
The common factor $a^2$ drops out of the logarithms because the three coefficients in the Shannon combination sum to zero.  Differentiating Eq.~\eqref{eq:Jk-boundary} with respect to $k$ removes the remaining logarithmic cancellation problem and gives the convergent integral
\begin{equation}
J'(k)=\int_0^1\frac{\dd t}{(1-t)^2}
\log\left[1+\frac{(1-t)^2}{kt}\right].
\label{eq:Jk-prime-t}
\end{equation}
Now set $x=(1-t)/t$.  After one integration by parts,
\begin{align}
J'(k)
&=\int_0^\infty\frac{\dd x}{x^2}
\log\frac{x^2+kx+k}{k(x+1)}\nonumber\\
&=\int_0^\infty
\frac{x+2}{(x+1)(x^2+kx+k)}\dd x.
\label{eq:Jk-prime-rational}
\end{align}
This last integral is elementary.  For $k\ge4$ write
\begin{equation}
k=4\cosh^2\eta,
\end{equation}
for which
\begin{equation}
x^2+kx+k=
\bigl(x+2e^{\eta}\cosh\eta\bigr)
\bigl(x+2e^{-\eta}\cosh\eta\bigr).
\end{equation}
Partial fractions, together with
\begin{equation}
\int_0^\infty\left(\frac1{x+A}-\frac1{x+B}\right)\dd x
=\log\frac{B}{A},
\end{equation}
then yield
\begin{equation}
J'(k)=\log(2\cosh\eta)-\eta\tanh\eta.
\label{eq:Jk-prime-hyperbolic}
\end{equation}
The integration constant can be fixed without an asymptotic matching.  On the branch $0<k<4$, write $k=4\cos^2\theta$; analytic continuation of Eq.~\eqref{eq:Jk-prime-hyperbolic} gives
\begin{equation}
J'(k)=\log(2\cos\theta)+\theta\tan\theta.
\end{equation}
Since $J(0)=0$, integration from $k=0$ to $4$ gives
\begin{align}
J(4)
&=4\int_0^{\pi/2}\sin2\theta
\bigl[\log(2\cos\theta)+\theta\tan\theta\bigr]\dd\theta\nonumber\\
&=4\left(\log2-\frac12\right)
+8\left(\frac{\pi^2}{16}+\frac14\right)
=\frac{\pi^2}{2}+4\log2.
\label{eq:J4-boundary}
\end{align}
Integrating Eq.~\eqref{eq:Jk-prime-hyperbolic} from $\eta=0$ therefore gives
\begin{equation}
J(\eta)=\frac{\pi^2}{2}
+4\cosh^2\eta\,\log(2\cosh\eta)
-2\eta\sinh2\eta+2\eta^2.
\label{eq:Jeta-boundary}
\end{equation}
At the Ising value $k=16$ one has $\cosh\eta=2$, hence
\begin{equation}
\eta=\arcosh2=\log(2+\sqrt3)\equiv L_3,
\qquad
\sinh2\eta=4\sqrt3.
\end{equation}
Using $a^2=1/256$ in Eq.~\eqref{eq:I2-Jk-def} finally gives
\begin{equation}
I_2=\frac1{256}\left[
\frac{\pi^2}{2}+2L_3^2+32\log2-8\sqrt3L_3
\right],
\qquad L_3=\log(2+\sqrt3).
\label{eq:I2-app-derived}
\end{equation}
The factorization
\begin{equation}
t^2+14t+1=[t+(2+\sqrt3)^2][t+(2-\sqrt3)^2]
\end{equation}
is the same hyperbolic parametrization written directly in the original $t$ variable and makes the origin of $L_3$ transparent.

\paragraph{Microscopic fixed-separation boundary layer.}
Now hold $r=j-i$ fixed and translate the pair a distance $R$ from the same endpoint.  Put
\begin{equation}
c=\frac1{16R},\qquad p=\beta_r.
\end{equation}
Then the two ghost probabilities are $c+O(R^{-2})$ while the physical edge probability is $p$.  In the edge ordering $((0,i),(0,j),(i,j))$, Eq.~\eqref{eq:triangle-exact} shows that the complementary pairs $\{100,011\}$ and $\{010,101\}$ have total weights $c+O(c^2)$ and conditional probabilities $1-p+O(c)$ and $p+O(c)$, respectively.  They therefore contribute $c h(p)+O(c^2|\log c|)$ each.  The pairs $\{000,111\}$ and $\{001,110\}$ have ambiguity only at order $c^2|\log c|$.  Consequently
\begin{equation}
C_\triangle(\lambda_{R,\ell},\lambda_{R+r,\ell},\beta_r)
=2c\,h(\beta_r)+O(c^2|\log c|)
=\frac{h(\beta_r)}{8R}+O(R^{-2}\log R),
\end{equation}
which proves Eq.~\eqref{eq:triangle-micro-tail}.  Summing $R^{-1}$ and then $r$ yields $s_{\edge}/8$ per endpoint and therefore $s_{\edge}/4$ in total.

The macroscopic integral and the fixed-$r$ term are the bulk and boundary-layer pieces of the same lattice sum rather than two independent counts of the same configurations.  One may make the matching explicit by introducing an intermediate separation cutoff $M$: treat $r\le M$ with the fixed-$r$ expansion and $r>M$ with the Riemann sum.  Because
\begin{equation}
\sum_{r>M}h(\beta_r)=O\!\left(\frac{\log M}{M}\right)
\end{equation}
and the endpoint singularity of $\Gamma_2^\partial$ is integrable, the coefficient of $\log\ell$ is independent of the order in which the limits $\ell\to\infty$ and then $M\to\infty$ are taken.  This gives Eq.~\eqref{eq:delta2-structural}.

\paragraph{Matching physical-line triangle.}
For the three-physical-vertex sector, translation invariance gives Eqs.~\eqref{eq:A3-span}--\eqref{eq:Phi3-span}.  Hold $r$ fixed and send the outer span $R\to\infty$.  With $p=\beta_r$ and $c=1/(16R^2)$, the two long physical edges have probabilities $c+O(R^{-3})$.  Applying the same four-fiber formula \eqref{eq:triangle-exact}, the two complementary pairs with total weight $c+o(c)$ contribute $c h(p)$ each, while the other pairs are $O(c^2|\log c|)$.  Hence
\begin{equation}
C_\triangle(\beta_r,\beta_{R-r},\beta_R)
=\frac{h(\beta_r)}{8R^2}+O(R^{-3}\log R),
\end{equation}
which proves Eq.~\eqref{eq:line-triangle-tail}.  The opposite collision face $R-r=O(1)$ gives the same contribution.  The common-scale interior is $o(R^{-2})$ after summing its $O(R)$ positions, so the same intermediate-cutoff matching gives Eq.~\eqref{eq:Phi3-tail}.  The span sum then has logarithmic coefficient $-s_{\edge}/4$, proving Eq.~\eqref{eq:alpha3-line} and the exact cancellation in Eq.~\eqref{eq:kappa3-main}.

\paragraph{General Cauchy family.}
The same calculation extends directly to the auxiliary Cauchy deformation.  Put $a=\nu^2/4$, $u=a^2/t$, $v=a/(1-t)^2$ and $k=1/a=4/\nu^2$.  With $I_2(\nu)=a^2J(k)$, Eqs.~\eqref{eq:Jk-prime-rational}--\eqref{eq:Jk-prime-hyperbolic} give
\begin{equation}
J'(k)=\int_0^{\infty}\frac{x+2}{(x+1)(x^2+kx+k)}\dd x
=\log(2\cosh\eta)-\eta\tanh\eta,
\qquad k=4\cosh^2\eta.
\end{equation}
Using the normalization fixed in Eq.~\eqref{eq:J4-boundary} therefore yields
\begin{equation}
J(\eta)=\frac{\pi^2}{2}+4\cosh^2\eta\log(2\cosh\eta)-2\eta\sinh2\eta+2\eta^2,
\end{equation}
which is equivalent to Eq.~\eqref{eq:I2nu}.  The fixed-separation matching deforms at the same time.  Since $\beta_R(\nu)=\nu^2/(4R^2)+O(R^{-4})$, the two collision ends of the physical-line triangle give $\alpha_3^{\linesys}(\nu)=-\nu^2s_{\edge}(\nu)$.  Since $\lambda_R^{(\nu)}=\nu^2/(4R)+O(R^{-2})$, the ghost boundary layer contributes $+\nu^2s_{\edge}(\nu)$ to $\delta_2^{\partial}(\nu)$.  The two terms cancel in the complete augmented sector, proving Eq.~\eqref{eq:kappa3nu}.  Thus the hyperbolic coordinate used in the deformed problem is not an additional ansatz: it evaluates the first complete augmented boundary period through $\kappa_3^{+}(\nu)=2I_2(\nu)$.

\subsection{Higher linked periods: explicit next-order constructions}
\label{app:higher-periods}

The preceding two subsections solve the first nontrivial periodic and boundary periods completely.  Here we spell out how the same machinery produces the \emph{next} analytic integrals.  Without attempting to evaluate their finite parts, the construction reduces the periodic six-point sector to an explicit three-modulus rational/logarithmic integral and the boundary three-point sector to an explicit two-variable regulated integral.  In both examples the remaining difficulty is the linked subtraction on intersecting singular faces, not the construction of the integrand itself.  The resulting representations provide starting points for symbolic integration, sector decomposition, differential-equation methods, or high-precision numerical quadrature.

The singular domains and the recursive appearance of lower-support residues are summarized in \Cref{fig:higher-period-renormalization}.

\begin{figure}[t]
\centering
\includegraphics[width=\textwidth]{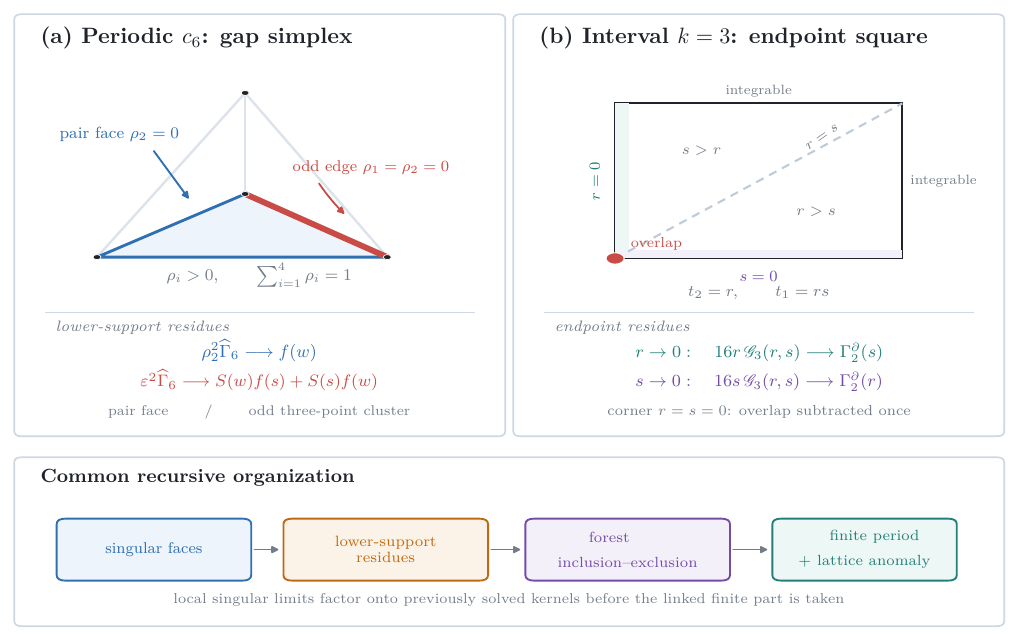}
\caption{Singular geometry of the first higher linked periods.  (a) For the periodic six-vertex problem, the positive gap variables $\rho_i$ form a three-dimensional simplex.  The pair-collision face $\rho_2=0$ carries the already solved four-point residue $f$, while the codimension-two edge $\rho_1=\rho_2=0$ gives the odd three-point-cluster factorization involving $S$ and $f$.  (b) For the interval $k=3$ problem, $t_2=r$ and $t_1=rs$ map the ordered shape domain to the unit square.  The faces $r=0$ and $s=0$ carry the lower-order kernel $\Gamma_2^\partial$; the top and right faces are locally integrable, and the corner $r=s=0$ is an overlap removed once by sector/forest inclusion--exclusion.  The lower strip summarizes the common recursive organization: singular faces factor onto lower-support kernels before the linked finite part is taken and the finite lattice face anomaly is restored.}
\label{fig:higher-period-renormalization}
\end{figure}

\subsubsection{Periodic example: the six-vertex period}
\label{app:c6-setup}

For $2p=6$, symmetrization of Eq.~\eqref{eq:c2p-period} gives
\begin{equation}
 c_6^{\rm cont}
 =\frac1{6!}\,\FP_{\rm linked}
 \int_{\mathbb T^6}\Gamma_6(x_1,\ldots,x_6)\,\dd^6x.
 \label{eq:c6-torus}
\end{equation}
The tangent transformation $y_i=\tan(\pi x_i)$ used for $c_4$ again removes all one-body factors.  Since every perfect matching contains each of the six vertices once,
\begin{equation}
 \Gamma_6(x)\,\dd^6x
 =\frac1{4^6}\,\widetilde\Gamma_6(y)\,\dd^6y,
 \qquad
 \widetilde K_{ij}=\frac1{y_i-y_j}.
 \label{eq:c6-tangent-measure}
\end{equation}
The integrand is symmetric, so the factor $1/6!$ restricts the rational integral to the ordered chamber:
\begin{equation}
 c_6^{\rm cont}
 =\frac1{4096}\,\FP_{\rm linked}
 \int_{y_1<\cdots<y_6}\widetilde\Gamma_6(y)\,\dd^6y.
 \label{eq:c6-ordered}
\end{equation}

A convenient projective chart makes the three genuine shape variables explicit.  Choose
\begin{equation}
 y_1=a,\qquad y_5=b,\qquad y_6=c,
 \qquad a<b<c,
\end{equation}
and write the remaining points as
\begin{equation}
 y_2=g(u),\qquad y_3=g(v),\qquad y_4=g(w),
 \qquad 0<u<v<w<1,
 \label{eq:c6-projective-coordinates}
\end{equation}
where the unique increasing M\"obius map sending $(0,1,\infty)$ to $(a,b,c)$ may be written
\begin{equation}
 g(z)=\frac{cz+ad}{z+d},
 \qquad
 d=\frac{c-b}{b-a}.
 \label{eq:c6-mobius-map}
\end{equation}
It obeys
\begin{equation}
 g'(z)=\frac{d(c-a)}{(z+d)^2}.
\end{equation}
Let
\begin{equation}
 z_1=0,\qquad z_2=u,\qquad z_3=v,\qquad
 z_4=w,\qquad z_5=1,\qquad z_6=\infty .
\end{equation}
For finite $z_i,z_j$,
\begin{equation}
 g(z_i)-g(z_j)
 =
 \frac{d(c-a)(z_i-z_j)}
 {(z_i+d)(z_j+d)},
 \qquad
 c-g(z_i)=\frac{d(c-a)}{z_i+d}.
 \label{eq:c6-difference-transform}
\end{equation}
Every six-point matching contains one edge ending at $z_6=\infty$ and each of the other five marked points exactly once.  Hence all fifteen matching weights acquire the same M\"obius factor.  Meanwhile the full coordinate Jacobian is triangular after $(y_1,y_5,y_6)$ are ordered first and equals
\begin{equation}
 \left|\frac{\partial(y_1,\ldots,y_6)}
 {\partial(a,b,c,u,v,w)}\right|
 =
 g'(u)g'(v)g'(w).
 \label{eq:c6-projective-jacobian}
\end{equation}
For completeness, the cancellation can be displayed explicitly.  Writing
\begin{equation}
 A=d(c-a),
\end{equation}
the common factor carried by any matching before the coordinate Jacobian is
\begin{equation}
 \mathcal W_{\rm Mob}
 =\frac{d^2(1+d)^2(u+d)^2(v+d)^2(w+d)^2}{A^6},
\end{equation}
where the factors $d^2$ and $(1+d)^2$ come from the marked points $z_1=0$ and $z_5=1$, while the remaining three factors come from $u,v,w$.  Since
\begin{equation}
 g'(u)g'(v)g'(w)
 =\frac{A^3}{(u+d)^2(v+d)^2(w+d)^2},
\end{equation}
one obtains
\begin{equation}
 \mathcal W_{\rm Mob}\,g'(u)g'(v)g'(w)
 =\frac{d^2(1+d)^2}{A^3}
 =\frac1{(b-a)(c-b)(c-a)},
 \label{eq:c6-mobius-factor-cancel}
\end{equation}
where in the last equality $d=(c-b)/(b-a)$ was used.  This gives the exact projective factorization
\begin{equation}
 \widetilde\Gamma_6(y)\,\dd^6y
 =
 \frac{\dd a\,\dd b\,\dd c}
 {(b-a)(c-b)(c-a)}
 \,
 \widehat\Gamma_6(u,v,w)\,
 \dd u\,\dd v\,\dd w .
 \label{eq:c6-projective-factorization}
\end{equation}
Thus all nontrivial six-point shape dependence is three-dimensional.  The orbit factor is the same projective measure that appeared in the four-point calculation; the new information is entirely in $\widehat\Gamma_6$.

The reduced kernel can be written without any hidden combinatorics.  Define
\begin{equation}
 r_{ij}=\frac1{(z_i-z_j)^2},
 \qquad 1\le i<j\le5,
 \label{eq:c6-rij}
\end{equation}
so explicitly
\begin{align}
&r_{12}=u^{-2},\quad r_{13}=v^{-2},\quad r_{14}=w^{-2},\quad r_{15}=1,\nonumber\\
&r_{23}=(v-u)^{-2},\quad r_{24}=(w-u)^{-2},\quad r_{25}=(1-u)^{-2},\nonumber\\
&r_{34}=(w-v)^{-2},\quad r_{35}=(1-v)^{-2},\quad r_{45}=(1-w)^{-2}.
\label{eq:c6-rij-explicit}
\end{align}
After the common edge-to-infinity factor is removed, the fifteen perfect-matching weights are
\begin{equation}
\begin{array}{lll}
 r_{23}r_{45},&r_{24}r_{35},&r_{25}r_{34},\\
 r_{13}r_{45},&r_{14}r_{35},&r_{15}r_{34},\\
 r_{12}r_{45},&r_{14}r_{25},&r_{15}r_{24},\\
 r_{12}r_{35},&r_{13}r_{25},&r_{15}r_{23},\\
 r_{12}r_{34},&r_{13}r_{24},&r_{14}r_{23}.
\end{array}
\label{eq:c6-fifteen-weights}
\end{equation}
The five rows correspond to the five possible partners of the point at infinity.  With
\begin{equation}
 \widehat Z_6=\sum_{\alpha=1}^{15}\widehat w_\alpha,
\end{equation}
the explicit three-modulus Shannon kernel is
\begin{equation}
 \widehat\Gamma_6(u,v,w)
 =
 \widehat Z_6\log\widehat Z_6
 -\sum_{\alpha=1}^{15}\widehat w_\alpha\log\widehat w_\alpha .
 \label{eq:c6-shape-kernel}
\end{equation}
Equations~\eqref{eq:c6-projective-factorization}--\eqref{eq:c6-shape-kernel} turn Eq.~\eqref{eq:c6-ordered} into the fully explicit analytic representation
\begin{equation}
 c_6^{\rm cont}
 =
 \frac1{4096}\,
 \FP_{\rm linked}
 \int_{a<b<c}\frac{\dd a\,\dd b\,\dd c}
 {(b-a)(c-b)(c-a)}
 \int_{0<u<v<w<1}
 \widehat\Gamma_6(u,v,w)\,
 \dd u\,\dd v\,\dd w .
 \label{eq:c6-explicit-integral}
\end{equation}
The notation in Eq.~\eqref{eq:c6-explicit-integral} is deliberately a \emph{single} linked finite part.  Although the ordinary integrand factorizes into an orbit measure and a shape kernel, the collision forest contains overlaps between orbit and shape boundaries.  One should therefore not replace $\FP_{\rm linked}$ by the product of two independently chosen finite parts without proving their compatibility.

For analytic work it is convenient to replace $(u,v,w)$ by the four positive gaps
\begin{equation}
 \rho_1=u,\qquad
 \rho_2=v-u,\qquad
 \rho_3=w-v,\qquad
 \rho_4=1-w,
 \qquad
 \sum_{i=1}^4\rho_i=1.
 \label{eq:c6-shape-gaps}
\end{equation}
All adjacent shape collisions are then coordinate faces of the simplex.  To expose their residues it is useful to introduce the \emph{local shape regulator}
\begin{equation}
 \mathcal P_6(\boldsymbol\alpha)
 =
 \int_{\substack{\rho_i>0\\ \sum_i\rho_i=1}}
 \left(\prod_{i=1}^4\rho_i^{\alpha_i}\right)
 \widehat\Gamma_6(\rho)\,
 \dd\rho_1\,\dd\rho_2\,\dd\rho_3,
 \label{eq:c6-analytic-regulator}
\end{equation}
initially for real parts of the $\alpha_i$ large enough to make every shape face integrable.  Equation~\eqref{eq:c6-analytic-regulator} is \emph{not} a definition of the complete linked finite part in Eq.~\eqref{eq:c6-explicit-integral}: it regulates the shape simplex at fixed projective chart and is used only to identify local face poles and their residues.  In particular it does not separately renormalize the projective-orbit measure, nor does it by itself resolve collision forests whose representatives mix orbit and shape boundaries.  Sector decomposition of simultaneous shape limits nevertheless produces a meromorphic continuation toward $\boldsymbol\alpha=0$, and the resulting poles supply precisely the lower-support residue data that enter the single global $\FP_{\rm linked}$.  Those residues are fixed by the exact factorization of $\widehat\Gamma_6$.

For example, let $v=u+\varepsilon$ with $u$ and $w$ fixed away from all other faces.  The three dominant matchings are precisely those containing the edge $(2,3)$, and their remaining four-point weights are
\begin{equation}
 \frac1{(1-w)^2},\qquad 1,\qquad \frac1{w^2}.
\end{equation}
Using the four-point function $f$ defined in Eq.~\eqref{eq:c4-ft}, one obtains
\begin{equation}
 \widehat\Gamma_6(u,u+\varepsilon,w)
 =
 \frac{f(w)}{\varepsilon^2}
 +O(|\log\varepsilon|).
 \label{eq:c6-pair-factorization}
\end{equation}
Thus the pair-face counterterm is explicitly the already solved four-point kernel divided by the microscopic pair scale squared.  Equation~\eqref{eq:c6-pair-factorization} is the six-point version of the lower-support power subtraction in Appendix~\ref{app:finite-part}.

The first genuinely nested example is an odd three-point cluster.  Put
\begin{equation}
 u=\varepsilon s,\qquad v=\varepsilon,
 \qquad 0<s<1,
\end{equation}
with $w$ fixed.  At order $\varepsilon^{-2}$ the nine dominant matchings factor into the three near-perfect internal pairings of $(0,u,v)$ and the three perfect matchings of the quotient four-point configuration $(0,w,1,\infty)$.  If $S$ and $f$ are the functions in Eqs.~\eqref{eq:c4-St} and~\eqref{eq:c4-ft}, the relevant factorization follows from an elementary identity.  For positive weights $a_i,b_j$, let $A=\sum_i a_i$ and $B=\sum_jb_j$.  Then
\begin{align}
 AB\log(AB)-\sum_{i,j}a_i b_j\log(a_i b_j)
 ={}&B\left(A\log A-\sum_i a_i\log a_i\right)\nonumber\\
 &+A\left(B\log B-\sum_j b_j\log b_j\right).
 \label{eq:product-entropy-identity}
\end{align}
Applying this identity to the internal and quotient matching weights gives
\begin{equation}
 \widehat\Gamma_6(\varepsilon s,\varepsilon,w)
 =
 \frac1{\varepsilon^2}
 \left[S(w)f(s)+S(s)f(w)\right]
 +O\!\left(\varepsilon^{-1}|\log\varepsilon|\right).
 \label{eq:c6-triple-factorization}
\end{equation}
The two terms have a transparent interpretation: one is the entropy of the internal odd-cluster channel weighted by the quotient partition function, and the other is the quotient four-point entropy weighted by the internal near-perfect-matching partition function.  The remaining pair, triple and nested/disjoint faces are generated by projective images and compatible combinations of the same local factorizations.

Equations~\eqref{eq:c6-explicit-integral}--\eqref{eq:c6-triple-factorization} provide an executable analytic route to $c_6$: enumerate the singular sectors, replace each by its lower-support factorization, combine compatible counterterms by forest inclusion--exclusion, analytically continue the regulated remainder, and finally restore the finite lattice face anomalies.  What is not supplied here is the completed three-modulus finite part itself.  Identifying that continuum result with the lattice intercept remains precisely the $p=3$ instance of Conjecture~\ref{conj:FPmatch}.

\subsubsection{Boundary example: the three-physical-vertex period}
\label{app:k3-setup}

The next interval kernel is already given in Eqs.~\eqref{eq:Gamma3-boundary-main} and~\eqref{eq:Gamma3-boundary}.  A change of variables makes its analytic structure considerably more explicit.  Set
\begin{equation}
 t_2=r,\qquad t_1=rs,\qquad
 0<r<1,\quad 0<s<1,
 \label{eq:k3-square-coordinates}
\end{equation}
so that
\begin{equation}
 \dd t_1\,\dd t_2=r\,\dd r\,\dd s.
\end{equation}
The four monomer--dimer weights become
\begin{align}
 w_0&=\frac1{4096\,r^2s},&
 w_{12}&=\frac1{256\,r^2(1-s)^2},\nonumber\\
 w_{13}&=\frac1{256\,r(1-rs)^2},&
 w_{23}&=\frac1{256\,rs(1-r)^2}.
 \label{eq:k3-square-weights}
\end{align}
Define the square-integrand
\begin{equation}
 \mathscr G_3(r,s)
 \equiv r\,\Gamma_3^\partial(rs,r,1).
 \label{eq:k3-square-integrand}
\end{equation}
Before endpoint subtraction, the common-scale shape integral is simply
\begin{equation}
 \int_{0<t_1<t_2<1}\Gamma_3^\partial(t_1,t_2,1)\,
 \dd t_1\,\dd t_2
 =
 \int_0^1\dd r\int_0^1\dd s\,\mathscr G_3(r,s).
 \label{eq:k3-square-integral}
\end{equation}
The advantage of Eq.~\eqref{eq:k3-square-coordinates} is that the two nonintegrable endpoint hierarchies become coordinate faces of the unit square.

Indeed, Eq.~\eqref{eq:k3-common-scale} gives, at fixed $s$,
\begin{equation}
 \mathscr G_3(r,s)
 =
 \frac1{16r}\Gamma_2^\partial(s)
 +O(|\log r|),
 \qquad r\to0^+,
 \label{eq:k3-r-face}
\end{equation}
while Eq.~\eqref{eq:k3-factorization} gives, at fixed $r$,
\begin{equation}
 \mathscr G_3(r,s)
 =
 \frac1{16s}\Gamma_2^\partial(r)
 +O(|\log s|),
 \qquad s\to0^+.
 \label{eq:k3-s-face}
\end{equation}
Thus both nonintegrable residues are completely fixed by the solved two-point period.  By contrast, $s\to1$ is the ordinary collision $t_1\to t_2$ and is only logarithmically singular by Eq.~\eqref{eq:k3-pair-collision}; the analogous $r\to1$ physical collision is also locally integrable.  No new counterterm is required on those faces at the level of the shape integral.

An analytic regulator can therefore be introduced directly:
\begin{equation}
 \mathcal J_3(\alpha,\beta)
 =
 \int_0^1\dd r\int_0^1\dd s\,
 r^\alpha s^\beta\,\mathscr G_3(r,s),
 \label{eq:k3-regulator}
\end{equation}
which converges for $\Re\alpha,\Re\beta>0$ after the locally integrable pair-collision faces are treated in the ordinary sense.  Keeping the second regulator in its convergence half-plane, the two boundary residues follow immediately from Eqs.~\eqref{eq:k3-r-face} and~\eqref{eq:k3-s-face}:
\begin{align}
 \underset{\alpha=0}{\operatorname{Res}}\,
 \mathcal J_3(\alpha,\beta)
 &=
 \frac1{16}\int_0^1s^\beta\Gamma_2^\partial(s)\,\dd s,
 \label{eq:k3-alpha-residue}\\
 \underset{\beta=0}{\operatorname{Res}}\,
 \mathcal J_3(\alpha,\beta)
 &=
 \frac1{16}\int_0^1r^\alpha\Gamma_2^\partial(r)\,\dd r.
 \label{eq:k3-beta-residue}
\end{align}
At zero regulator each residue reduces to $I_2/16$, in agreement with Eq.~\eqref{eq:k3-first-residue}.  The appearance of the same lower period on both faces is the analytic origin of the nested logarithms in the raw $B_{3,\ell}$ sector.

The corner $r=s=0$ belongs to both descriptions and must not be subtracted twice.  A convenient way to make the overlap unambiguous is standard sector decomposition:
\begin{equation}
 \int_0^1\dd r\int_0^1\dd s
 =
 \int_{r>s}\dd r\,\dd s
 +
 \int_{s>r}\dd r\,\dd s,
 \label{eq:k3-sector-split}
\end{equation}
followed respectively by
\begin{equation}
 (r,s)=(\rho,\rho\sigma),
 \qquad
 (r,s)=(\rho\sigma,\rho),
 \qquad 0<\rho,\sigma<1.
 \label{eq:k3-sector-maps}
\end{equation}
In either sector the nested endpoint limit is a single coordinate face $\rho\to0$, while $\sigma$ keeps the relative hierarchy fixed.  Expanding the weights in Eq.~\eqref{eq:k3-square-weights} at fixed $\sigma$ generates the overlap counterterm, and the forest inclusion--exclusion in Appendix~\ref{app:boundary-forest} removes it exactly once.  Equivalently one may regard
\begin{equation}
 \mathcal I_{3,\rm ren}^\partial
 \equiv
 \FP_{\alpha,\beta\to0}^{\rm forest}
 \mathcal J_3(\alpha,\beta)
 \label{eq:k3-ren-period}
\end{equation}
as the renormalized continuum three-point shape period, where the superscript ``forest'' means that the Laurent finite part is taken only after the two single-face residues and their sectorwise overlap have been combined according to the compatible-forest prescription.

Equation~\eqref{eq:k3-ren-period} is not yet the complete lattice coefficient.  As in the solved $k=2$ example, a microscopic fixed-separation analysis can leave a finite face anomaly that is invisible to the naive continuum integral.  If the order-$k=3$ boundary forest construction is completed, the primitive single-log coefficient has the structural form
\begin{equation}
 \delta_3^\partial
 =
 2\,\mathcal I_{3,\rm ren}^\partial
 +\mathcal A_3^{\rm face},
 \label{eq:k3-structural-coefficient}
\end{equation}
where the factor two is for the two interval endpoints and $\mathcal A_3^{\rm face}$ denotes the net finite lattice boundary-layer anomaly after all overlaps with lower supports are removed.  Neither term on the right-hand side of Eq.~\eqref{eq:k3-structural-coefficient} is evaluated in this work, so Eq.~\eqref{eq:k3-structural-coefficient} is a roadmap rather than a numerical claim.

\paragraph{General recursive pattern.}
The two examples above show that the higher-order constructions remain finite at every fixed support before renormalization.  In the periodic problem there are $(2p-1)!!$ perfect matchings and $2p-3$ projective shape variables.  In the boundary problem the number $N_k$ of monomer--dimer coverings obeys
\begin{equation}
 N_k=N_{k-1}+(k-1)N_{k-2},
 \qquad N_0=N_1=1,
 \label{eq:telephone-recurrence}
\end{equation}
and fixing the outer scale leaves $k-1$ shape variables.  In both geometries the singular faces factor onto lower-support kernels, so the same previously computed periods reappear as residues at the next order.

At fixed support, the analytic construction is finite and recursive.  The matching or monomer--dimer entropy kernel is expressed in projective or endpoint-scale coordinates adapted to the singular strata; primitive faces are then expanded to identify their lower-support residues, disjoint and nested faces are combined by forest inclusion--exclusion, and the continuum finite part is supplemented by any finite lattice face anomaly not captured by the integral.  In the periodic problem the projective-orbit subtraction is part of this same linked operation, while in the interval problem the additional endpoint hierarchy must be included together with ordinary physical collisions.

This recursive closure is what makes an order-by-order analytic program plausible: once all lower supports have been controlled, the next support produces a finite list of new shape integrals and a finite list of counterterms inherited from earlier orders.  What remains open is a proof that this forest construction closes uniformly at arbitrary support, together with control of the resulting all-order series.  Nor does the recursion imply that all periods reduce to elementary constants: beginning already at higher support, genuine multidimensional polylogarithmic or more general periods may occur.

\section{Numerical evaluation}
\label{app:numerics}

This appendix documents the deterministic finite-size calculations used for the numerical statements in the main text: linked periodic sectors, direct periodic and interval entropies, and the stability of the extrapolations.

\subsection{Finite-size linked sectors}
\label{app:periodic-numerics}

The linked coefficients are obtained from the exact finite-size M\"obius sectors defined in Eq.~\eqref{eq:linked-expansion}.  Representative values are
\begin{center}
\begingroup
\small
\setlength{\tabcolsep}{3pt}
\begin{tabular}{crrrrrr}
\toprule
$L$&$A_3$&$A_4$&$A_5$&$A_6$&$A_7$&$A_8$\\
\midrule
12&0.78720596&0.17167279&-0.02584279&-0.00461292&0.00047104&0.00004013\\
16&1.07426446&0.27849701&-0.04047482&-0.01020244&0.00104755&0.00016019\\
20&1.35749257&0.38883276&-0.05464654&-0.01678121&0.00167215&0.00035082\\
24&1.63870655&0.50093100&-0.06846487&-0.02392760&0.00230434&0.00059286\\
28&1.91873391&0.61403526&-0.08203151&-0.03142509&0.00293173&0.00087150\\
30&2.05844546&0.67083837&-0.08874375&-0.03526572&0.00324247&0.00102122\\
\bottomrule
\end{tabular}
\endgroup
\end{center}
The extensive sign of an odd sector should not be confused with its constant part.  The macroscopic power counting excludes a new full-support odd matching kernel at order $L^0$, but the stronger statement that all odd lattice constants vanish remains unproved.

For completeness, the much smaller nine- and ten-vertex sectors are
\begin{center}
\begin{tabular}{c rr}
\toprule
$L$&$A_9$&$A_{10}$\\
\midrule
12&$-3.0481\times10^{-6}$&$-1.1134\times10^{-7}$\\
14&$-7.3615\times10^{-6}$&$-4.5672\times10^{-7}$\\
16&$-1.3342\times10^{-5}$&$-1.1679\times10^{-6}$\\
18&$-2.0613\times10^{-5}$&$-2.3129\times10^{-6}$\\
20&$-2.8831\times10^{-5}$&$-3.9131\times10^{-6}$\\
22&$-3.7725\times10^{-5}$&$-5.9599\times10^{-6}$\\
\bottomrule
\end{tabular}
\end{center}
At these sizes the inferred $L^0$ term of $A_{10,L}$ is strongly correlated with the inverse-size corrections.  We therefore do not assign a numerical value to $c_{10}$ in this work.

\subsection{Representative finite-size data}
\label{app:finite-data}

\paragraph{Periodic system.}
The full Walsh transform gives the measurement entropy directly.  The cycle entropy is then obtained independently as $H_{\cyc}=H_{\edge}-H_1^{(z)}$:
\begin{table}[h]
\centering
\caption{Finite-size periodic entropies obtained from the exact Walsh coefficients.}
\label{tab:periodic-full-data}
\begin{tabular}{c r r r}
\toprule
$L$ & $H_1^{(z)}(L)$ & $H_{\edge}(L)$ & $H_{\cyc}(L)$\\
\midrule
8  & $2.949687563358$ & $3.502820194538$ & $0.553132631179$\\
10 & $3.795785544094$ & $4.537069151390$ & $0.741283607295$\\
12 & $4.642032629666$ & $5.570963683756$ & $0.928931054091$\\
14 & $5.488365191525$ & $6.604592936651$ & $1.116227745126$\\
16 & $6.334751276926$ & $7.638028783959$ & $1.303277507033$\\
18 & $7.181173087765$ & $8.671322010005$ & $1.490148922240$\\
20 & $8.027619926909$ & $9.704507991239$ & $1.676888064330$\\
22 & $8.874084979058$ & $10.737611680427$ & $1.863526701368$\\
24 & $9.720563696818$ & $11.770651020953$ & $2.050087324134$\\
\bottomrule
\end{tabular}
\end{table}

\paragraph{Infinite-chain interval.}
The infinite-line minors in Eq.~\eqref{eq:d-infty} give
\begin{table}[h]
\centering
\caption{Interval Shannon entropies from the exact infinite-line Walsh coefficients.}
\label{tab:interval-full-data}
\begin{tabular}{c r@{\qquad}c r}
\toprule
$\ell$ & $H_1(A_\ell)$ & $\ell$ & $H_1(A_\ell)$\\
\midrule
4 & 1.805854593071358 & 14 & 6.107833679024373\\
6 & 2.674003797245196 & 16 & 6.962089515106668\\
8 & 3.536422963908597 & 18 & 7.815467577832148\\
10& 4.395517906953367 & 20 & 8.668147394540483\\
12& 5.252441034545338 & 22 & 9.520258511384979\\
\bottomrule
\end{tabular}
\end{table}

The edge--cycle split and the distinguished-vertex decomposition can be checked independently at the same finite sizes.  The quantities $E_\ell$ are obtained from Eq.~\eqref{eq:Eell}; $C_\ell=E_\ell-H_1(A_\ell)$; $C_\ell^{\linesys}$ is computed from the physical-line boundary law using Eq.~\eqref{eq:conditional-fwht}; and $B_{2,\ell}$ is the exact triangle sum in Eq.~\eqref{eq:B2-triangle}.  Write $C_\ell^\partial\equiv C_\ell-C_\ell^{\linesys}$ for the complete ghost-containing contribution.  Representative values are
\begin{table}[h]
\centering
\caption{Finite-size interval cycle decomposition.  The quantity $C_\ell^\partial$ contains all ghost-linked M\"obius sectors, while $B_{2,\ell}$ is only the first one.}
\label{tab:interval-cycle-decomp}
\begingroup
\small
\setlength{\tabcolsep}{3pt}
\begin{tabular}{c r r r r r}
\toprule
$\ell$ & $E_\ell$ & $C_\ell$ & $C_\ell^{\linesys}$ & $C_\ell^\partial$ & $B_{2,\ell}$\\
\midrule
4  & 2.11884181 & 0.31298722 & 0.05541551 & 0.25757172 & 0.24416915\\
8  & 4.25439690 & 0.71797394 & 0.26162216 & 0.45635178 & 0.40513392\\
12 & 6.36356297 & 1.11112194 & 0.52110057 & 0.59002137 & 0.50245801\\
16 & 8.46063608 & 1.49854657 & 0.80649071 & 0.69205586 & 0.57181081\\
\bottomrule
\end{tabular}
\endgroup
\end{table}
The table makes Eq.~\eqref{eq:Cell-line-ghost} concrete: the ghost-containing contribution is not the complete conditional entropy, and even at these sizes the first ghost sector does not exhaust it.

The same data also make the cancellation behind Eq.~\eqref{eq:conditional-slope-cancellation} visible.  Define the two-scale operator
\begin{equation}
\mathcal D_\ell[X]=\frac{X(2\ell)-X(\ell)}{\log2}.
\end{equation}
For the complete conditional entropy and the line contribution we first subtract the common bulk term $s_{\cyc}\ell$; the ghost pieces have no extensive term.  Indeed, with $Y=\partial F_A$ and $\zeta$ the independent ghost-edge vector,
\begin{equation}
C_\ell^\partial=C_\ell-C_\ell^{\linesys}
=H(\zeta)+H(Y)-H(Y\oplus\zeta),
\qquad
0\le C_\ell^\partial\le2H(\zeta)=2E_\ell^{\rm ghost}=O(\log^2\ell),
\end{equation}
where the bound follows from the entropy inequalities for $Y$, $\zeta$, and $Y\oplus\zeta$.  Using $s_{\cyc}=0.09306147$, the displayed data give
\begin{table}[h]
\centering
\caption{Two-scale effective slopes illustrating the cancellation of the interval conditional logarithm.  The first two columns are bulk-subtracted.  These are finite-size diagnostics, not separately defined universal coefficients for the higher raw sectors.}
\label{tab:interval-cycle-effective-slopes}
\begingroup
\small
\setlength{\tabcolsep}{4pt}
\begin{tabular}{c r r r r}
\toprule
$\ell$ & $\mathcal D_\ell[C-s_{\cyc}\ell]$
& $\mathcal D_\ell[C^{\linesys}-s_{\cyc}\ell]$
& $\mathcal D_\ell[B_2]$
& $\mathcal D_\ell[C^\partial-B_2]$\\
\midrule
4 & $0.047235$ & $-0.239544$ & $0.232223$ & $0.054556$\\
8 & $0.052054$ & $-0.287995$ & $0.240464$ & $0.099585$\\
\bottomrule
\end{tabular}
\endgroup
\end{table}
The columns add exactly, up to the rounding of Table~\ref{tab:interval-cycle-decomp}, as required by Eq.~\eqref{eq:Cell-three-piece}.  The total bulk-subtracted slope is already close to $\gamma_C=\gamma_E-\gamma_1\simeq0.05077$, while the individual pieces are much larger and cancel strongly.  The drift of the higher ghost column should not be interpreted as convergence to a separate single-log coefficient: \Cref{sec:subsystem-higher-boundary} shows that raw $k\ge3$ ghost sectors contain nested logarithms before forest subtraction.

The analytic coefficient $\delta_2^\partial$ can be checked at sizes much larger than those accessible to the full $2^\ell$ Walsh reconstruction because $B_{2,\ell}$ requires only a double sum of triangle entropies.  Define the two-scale effective slope
\begin{equation}
\delta_{2,\rm eff}^{\partial}(\ell)
=\frac{B_{2,2\ell}-B_{2,\ell}}{\log2}.
\end{equation}
Direct lattice evaluation gives
\begin{center}
\begin{tabular}{c r}
\toprule
$\ell$ & $\delta_{2,\rm eff}^{\partial}(\ell)$\\
\midrule
20  & 0.23905364\\
50  & 0.23451640\\
100 & 0.23154266\\
200 & 0.22935789\\
\bottomrule
\end{tabular}
\end{center}
The slow approach is consistent with the analytic asymptote in Eq.~\eqref{eq:delta2-boundary}.

\paragraph{Decimation check.}
As a direct finite-size check of Eqs.~\eqref{eq:Gm-block}--\eqref{eq:decimation-general}, the correlation matrix can be evaluated from Eq.~\eqref{eq:Gm-decimation} and the corresponding determinant probabilities compared with the sum of TFI block entropies.  Representative values are
\begin{center}
\begin{tabular}{c c c c c}
\toprule
$m$ & $\ell$ & block sizes $\{\ell_r\}$ & direct $H_1^{(m)}$ & TFI block sum\\
\midrule
2 & 7 & $\{4,3\}$ & 3.173825205087676 & 3.173825205087676\\
3 & 7 & $\{3,2,2\}$ & 3.218852722600712 & 3.218852722600712\\
4 & 9 & $\{3,2,2,2\}$ & 4.144293777892909 & 4.144293777892909\\
\bottomrule
\end{tabular}
\end{center}
The two evaluations agree at double-precision roundoff for these sizes.

\subsection{Finite-size extrapolations}
\label{app:fit-stability}

All sequences below are deterministic evaluations of the exact finite-size formulas.  Parenthetical uncertainties therefore denote the spread under changes of the minimum size and correction order, not statistical standard errors.  For a periodic sequence we use
\begin{equation}
 X(L)=sL+c+\sum_{j=1}^{J}\frac{a_j}{L^j},
 \label{eq:periodic-fit-family}
\end{equation}
and for the interval entropy
\begin{equation}
 H_1(A_\ell)=s\ell+\gamma\log\ell+c+\sum_{j=1}^{J}\frac{a_j}{\ell^j}.
 \label{eq:interval-fit-family}
\end{equation}
The tables below show representative variations of these fits.

\paragraph{Fixed-support linked sectors.}
Fitting the displayed linked-sector values over the available sizes gives
\begin{table}[H]
\centering
\caption{Dependence of the linked-sector extrapolations on the highest inverse power $J$ retained in Eq.~\eqref{eq:periodic-fit-family}.  The table is based on the values printed above and is used only to determine the precision quoted in the main text.}
\label{tab:linked-fit-order}
\begin{tabular}{c r r r}
\toprule
$J$ & fitted $c_6$ & fitted $c_8$ & fitted $c_{10}$\\
\midrule
2 & $0.03471790$ & $-0.00278935$ & $8.93\times10^{-5}$\\
3 & $0.03469398$ & $-0.00305098$ & $1.401\times10^{-4}$\\
4 & $0.03458790$ & $-0.00305737$ & $1.644\times10^{-4}$\\
\bottomrule
\end{tabular}
\end{table}
For $c_6$ and $c_8$ the higher-order fits support the precision used in Eq.~\eqref{eq:coeffs-main}, namely $c_6\simeq0.0346$ and $c_8\simeq-0.00305$.  The much larger relative drift of the ten-vertex fit is the reason no value of $c_{10}$ is quoted.

\paragraph{Full periodic cycle entropy.}
Using the values in Table~\ref{tab:periodic-full-data}, representative fits give
\begin{table}[H]
\centering
\caption{Finite-size fits for $H_{\cyc}(L)$ under changes of the lower cutoff and correction order.}
\label{tab:periodic-full-fit}
\begin{tabular}{c c c r r}
\toprule
$L_{\min}$ & $J$ & sizes & $s_{\cyc}$ & $c_{\cyc}$\\
\midrule
8  & 2 & 9 & 0.0930619888 & -0.1781259999\\
8  & 3 & 9 & 0.0930614485 & -0.1780942365\\
8  & 4 & 9 & 0.0930614703 & -0.1780958504\\
10 & 3 & 8 & 0.0930614583 & -0.1780948708\\
10 & 4 & 8 & 0.0930614691 & -0.1780957479\\
12 & 4 & 7 & 0.0930614685 & -0.1780956983\\
14 & 3 & 6 & 0.0930614644 & -0.1780953151\\
\bottomrule
\end{tabular}
\end{table}
The $J=2$ fit is visibly under-corrected.  The remaining windows support the conservative values
\begin{equation}
 s_{\cyc}=0.09306147(2),\qquad c_{\cyc}=-0.1780956(3),
\end{equation}
which are the values used consistently in the main text and in the deformation comparison below.

\paragraph{Interval entropy.}
Allowing both $s$ and $\gamma$ to vary gives
\begin{table}[H]
\centering
\caption{Finite-size fits of the interval entropy in Eq.~\eqref{eq:interval-fit-family}.  The degrees of freedom are the number of data points minus the number of fitted parameters; zero-degree-of-freedom rows are displayed only as interpolation diagnostics and are not used as independent evidence for the uncertainty.}
\label{tab:interval-fit}
\begin{tabular}{c c c c r r}
\toprule
$\ell_{\min}$ & $J$ & sizes & d.o.f. & $s$ & $\gamma$\\
\midrule
4  & 2 & 10 & 5 & 0.4232729855 & 0.06004524\\
4  & 3 & 10 & 4 & 0.4232734978 & 0.06002324\\
4  & 4 & 10 & 3 & 0.4232735945 & 0.06001789\\
6  & 3 & 9  & 3 & 0.4232735577 & 0.06002023\\
6  & 4 & 9  & 2 & 0.4232735783 & 0.06001892\\
8  & 3 & 8  & 2 & 0.4232735683 & 0.06001964\\
8  & 4 & 8  & 1 & 0.4232735704 & 0.06001949\\
10 & 3 & 7  & 1 & 0.4232735692 & 0.06001958\\
10 & 4 & 7  & 0 & 0.4232735665 & 0.06001979\\
12 & 3 & 6  & 0 & 0.4232735682 & 0.06001965\\
\bottomrule
\end{tabular}
\end{table}

A cleaner test fixes the extensive density to the central value $s=0.42327356$ obtained independently from the periodic calculation and retains $J=3$:
\begin{table}[H]
\centering
\caption{Interval fits with the bulk density fixed to the independent periodic value $s=0.42327356$.}
\label{tab:interval-fit-fixed-s}
\begin{tabular}{c c c r}
\toprule
$\ell_{\min}$ & sizes & d.o.f. & $\gamma$\\
\midrule
4  & 10 & 5 & 0.06002061\\
6  & 9  & 4 & 0.06002012\\
8  & 8  & 3 & 0.06002010\\
10 & 7  & 2 & 0.06002015\\
12 & 6  & 1 & 0.06002020\\
\bottomrule
\end{tabular}
\end{table}
The fixed-$s$ sequence removes most of the $s$--$\gamma$ correlation and is stable at the few-$10^{-7}$ level over the last four displayed windows.

Because raw boundary support sectors can contain nested logarithms, it is useful to test rather than silently assume that no logarithmically enhanced correction can mimic part of the finite-size drift of the complete entropy.  As a diagnostic, fix the same $s$ and fit
\begin{equation}
H_1(A_\ell)=s\ell+\gamma\log\ell+c
+b\frac{\log\ell}{\ell}+\frac{a_1}{\ell}+\frac{a_2}{\ell^2}.
\label{eq:interval-logcorr-fit}
\end{equation}
The resulting values are
\begin{center}
\begin{tabular}{c c r}
\toprule
$\ell_{\min}$ & d.o.f. & $\gamma$\\
\midrule
4  & 5 & 0.06000879\\
6  & 4 & 0.06001298\\
8  & 3 & 0.06001556\\
10 & 2 & 0.06001706\\
12 & 1 & 0.06001798\\
\bottomrule
\end{tabular}
\end{center}
The log-enhanced fits drift upward toward the pure-power estimate as the lower cutoff is raised, but the available size range is too short to identify the correction basis uniquely.  We therefore retain the conservative quotation $\gamma_1=0.060020(3)$ as the envelope of the stable higher-cutoff extrapolations used in the main text, while emphasizing that this parenthetical number is not a rigorous bound over arbitrary asymptotic ans\"atze.  The literature value $b_1/8$ is not used as a fit input.

\section{Cauchy deformation, grand organization and quadratic-chain generalizations}
\label{app:grand-details}

This appendix follows the progression of Sec.~\ref{sec:cauchy-family}: we first test the auxiliary Cauchy deformation, then record the formal grand determinant/Hafnian construction, and finally separate from those auxiliary structures the exact decimation identity for a genuine family of critical quadratic chains.

\subsection{Cauchy deformation as an independent coefficient generator}
\label{app:deformation-test}

The deformation introduced in \Cref{sec:cauchy-family} gives a practical, independent test of the linked expansion.  For the finite ring let
\begin{equation}
(G_{\nu,L})_{jk}=\frac{\sin\pi\nu}{L\sin[\pi(j-k+\nu)/L]},
\qquad z=4\nu^2.
\end{equation}
Up to diagonal gauge signs (which do not change principal minors), its principal-minor factorization has
\begin{equation}
A_{\nu,L}=\frac{\sin\pi\nu}{L\sin(\pi\nu/L)},\qquad
g_m^{(\nu,L)}=\frac{1}{1-z_m^{(\nu,L)}},\qquad
z_m^{(\nu,L)}=\frac{\sin^2(\pi\nu/L)}{\sin^2(\pi m/L)},
\end{equation}
and the finite product $A_{\nu,L}^2\prod_{m=1}^{L-1}g_m^{(\nu,L)}=1$ gives the exact Bernoulli edge law
\begin{equation}
\beta_m^{(\nu,L)}=\frac12\left[1-\sqrt{1-z_m^{(\nu,L)}}\right].
\end{equation}
At fixed macroscopic separation $m=xL$ the Bernoulli activity is
\begin{equation}
 \beta_{xL}^{(\nu,L)}
 =z\,\frac{\pi^2}{16L^2}\csc^2(\pi x)+O(L^{-4}).
\end{equation}
A $2p$-vertex perfect matching contains $p$ macroscopic edges, so the full-support entropy kernel is exactly homogeneous of degree $p$ in $z$.  The same homogeneity is inherited by its local collision counterterms.  Conditional on the lattice--continuum finite-part matching of Conjecture~\ref{conj:FPmatch}, the corresponding lattice coefficients would therefore obey
\begin{equation}
c_{2p}(z)=z^p c_{2p}(1),\qquad
 C(z)=\sum_{p\ge2}c_{2p}z^p.
\label{eq:Cz-app}
\end{equation}
The Cauchy deformation itself is an exact finite-size Cauchy/parity probability family; no generic-$\nu$ local parent Hamiltonian is assumed.  What is conditional here is the identification of its universal cycle constant with the all-order linked series $C(z)$.  The finite-size tests below provide an independent numerical check of that identification at several $z<1$.

We can test \eqref{eq:Cz-app} without enumerating a fixed high vertex support.  For a chosen $z<1$, construct the complete finite-$L$ Walsh distribution of the deformed Cauchy probability law, compute its Shannon entropy, subtract the analytically known deformed edge sector, and extrapolate the cycle constant.  Higher linked orders are suppressed by $z^p$.  Representative values are given in Table~\ref{tab:Cz-test}.
\begin{table}[h]
\centering
\caption{Finite-size comparison for the Cauchy deformation.  The second column is obtained by fitting the complete deformed distributions at $L=12,14,\ldots,24$ with $J=4$ in Eq.~\eqref{eq:periodic-fit-family}.  The third uses the exact $c_4$ and the rounded values $c_6=0.0346$, $c_8=-0.00305$ quoted in Eq.~\eqref{eq:coeffs-main}.}
\label{tab:Cz-test}
\begin{tabular}{c r r r}
\toprule
$z$ & $C(z)$, full distribution & $c_4z^2+c_6z^3+c_8z^4$ & difference\\
\midrule
0.25 & $-0.01258302228$ & $-0.01258272621$ & $-2.96\times10^{-7}$\\
0.50 & $-0.04830976245$ & $-0.04831137360$ & $+1.61\times10^{-6}$\\
0.75 & $-0.10434347929$ & $-0.10437109841$ & $+2.76\times10^{-5}$\\
1.00 & $-0.17809568530$ & $-0.17823299439$ & $+1.37\times10^{-4}$\\
\bottomrule
\end{tabular}
\end{table}
The discrepancy decreases as $z$ is lowered, as expected when the first omitted linked term is of order $z^5$.  At $z=1$ the fixed-window value in the second column lies inside the conservative fit envelope in Eq.~\eqref{eq:ccyc-direct}.  The table is used only to test the organization of the deformation; no higher linked coefficient is inferred from it.

\subsection{Grand determinant/Hafnian derivation}
\label{app:grand-derivation}

This subsection supplies the coefficient normalization and Hafnian summation underlying the formal grand organization in Sec.~\ref{sec:grand}; the renormalization and $z=1$ resummation assumptions remain exactly those stated there.  For a symmetric integrand of $n=2p$ points on the circle, translation invariance gives
\begin{equation}
\frac1n\int_{\Delta_{n-1}}F=\frac1{n!}\int_{\mathbb T^n}F.
\end{equation}
This converts the linked coefficient normalization to the factorial normalization used in the grand series.

At fixed $p$ define
\begin{equation}
D_{2p}(q;x)=[\det K_{2p}(x)]^q,
\qquad
H_{2p}(q;x)=\Hf[(K_{ij}(x)^2)^q]
=\sum_M\prod_{(ij)\in M}K_{ij}^{2q}.
\end{equation}
For distinct points, the Pfaffian--Hafnian identity gives $D_{2p}(1;x)=H_{2p}(1;x)=Z_p(x)$.  Differentiating before any collision renormalization,
\begin{align}
\left.\partial_qD_{2p}(q;x)\right|_{q=1}&=Z_p\log Z_p,\\
\left.\partial_qH_{2p}(q;x)\right|_{q=1}&=\sum_Mw_M\log w_M,
\end{align}
so their difference is exactly $\Gamma_{2p}$.  This fixed-support derivative is the part of the grand construction that is already contained in the linked analysis of the main text.

The matching side can be summed independently in a region where no finite part is needed.  For $\Re q<1/2$ the two-point integral is absolutely convergent:
\begin{align}
J(q)&=\int_0^1\dd x\int_0^1\dd y
\left[\frac{\pi^2}{16\sin^2\pi(x-y)}\right]^q\\
&=\left(\frac\pi4\right)^{2q}
\frac{\Gamma(\frac12-q)}{\sqrt\pi\Gamma(1-q)}.
\end{align}
Every perfect matching factorizes into $p$ copies of $J(q)$, while
\begin{equation}
\frac{(2p-1)!!}{(2p)!}=\frac1{2^pp!}.
\end{equation}
Hence, in the convergence half-plane,
\begin{equation}
\cB(q,z)=\sum_{p=0}^{\infty}\frac1{p!}
\left(\frac{zJ(q)}2\right)^p
=e^{zJ(q)/2}.
\end{equation}
The explicit Gamma ratio gives its meromorphic continuation.  Expanding that continuation at $q=1+\epsilon$ yields
\begin{equation}
\frac{\Gamma(\frac12-q)}{\sqrt\pi\Gamma(1-q)}
=2\epsilon+O(\epsilon^2),
\end{equation}
so that
\begin{equation}
J'(1)=\frac{\pi^2}{8},
\qquad
\left.\partial_q\cB(q,z)\right|_{q=1}=\frac{\pi^2}{16}z.
\end{equation}
This calculation is exact as analytic continuation of the convergent Hafnian series; it does not require a linked collision finite part.

The determinant side is different.  The linked prescription constructed in Appendix~\ref{app:finite-part} renormalizes the Shannon combination obtained after differentiating at $q=1$.  It does not by itself produce a family of separately renormalized integrals $\operatorname{FP}_{q,\rm linked}\int D_{2p}(q)$ for $q$ near one.  Defining such a family would require counterterms depending coherently on $q$, including on nested collision forests, and then proving that differentiation at $q=1$ commutes with that renormalization.  This is the compatibility condition written explicitly in Eq.~\eqref{eq:q-ren-compatibility}.

A finite-matrix analogue helps explain why the proposed determinant-power object is natural even though the continuum renormalization remains open.  For a finite matrix $A$,
\begin{equation}
M_q(A;t)=\sum_St^{|S|}[\det A_S]^q
\end{equation}
reduces at $q=1$ to $\det(I+tA)$, while
\begin{equation}
\left.\partial_qM_q(A;t)\right|_{q=1}
=\sum_St^{|S|}\det A_S\log\det A_S,
\end{equation}
which is the SPPM structure studied in Ref.~\cite{NajafiRamezanpourRajabpour2025}.  The continuum grand determinant in Eq.~\eqref{eq:Aqz} is motivated by this analogy, but its $q$-dependent linked renormalization is an additional problem rather than a consequence of the finite-matrix identity.

Finally, even if the compatibility condition is established coefficientwise, evaluation at the physical point $z=1$ requires convergence or a controlled summation of the support series and the all-order odd-support statement.  Thus Eqs.~\eqref{eq:C-grand}, \eqref{eq:grand-master-z} and \eqref{eq:grand-master-x} should be read in the formal hierarchy stated in the main text, not as additional proved equalities for the physical constant.

\subsection{Exact computational-basis decimation for \texorpdfstring{$f_m(z)=z^m-1$}{fm(z) = z to the m minus 1}}
\label{app:decimation}
The quadratic models $f_m(z)=z^m-1$ considered in Ref.~\cite{Tarighi2022} provide a useful extension of the TFI calculation.  Write
\begin{equation}
 q_m(k)=\frac{e^{imk}-1}{|e^{imk}-1|}=q_1(mk),
\end{equation}
where $q_1$ is the phase of the critical TFI symbol $f_1(z)=z-1$.  In the thermodynamic limit their correlation matrices are Fourier coefficients of this phase,
\begin{equation}
 G^{(m)}_{ab}=\frac{1}{2\pi}\int_0^{2\pi}q_1(mk)e^{-i(b-a)k}\dd k.
 \label{eq:Gm-decimation}
\end{equation}
Split the integral into the $m$ intervals $2\pi r/m<k<2\pi(r+1)/m$ and set $\theta=mk-2\pi r$.  The sum over intervals produces the root-of-unity projector
\begin{equation}
 \frac1m\sum_{r=0}^{m-1}e^{-2\pi i r(b-a)/m}
 =\mathbf 1_{m\mid(b-a)}.
\end{equation}
Consequently
\begin{equation}
 G^{(m)}_{ab}=
 \begin{cases}
 G^{(1)}_{0,(b-a)/m},&m\mid(b-a),\\
 0,&m\nmid(b-a),
 \end{cases}
 \label{eq:Gm-sparse}
\end{equation}
up to the harmless diagonal gauge convention used for the TFI matrix.  For a contiguous interval $A_\ell$, permute the sites so that equal residue classes modulo $m$ are adjacent.  If $\ell_r$ is the number of sites in residue class $r$, Eq.~\eqref{eq:Gm-sparse} becomes the exact block decomposition
\begin{equation}
 P G_{A_\ell}^{(m)}P^{-1}
 =\bigoplus_{r=0}^{m-1}G_{A_{\ell_r}}^{\rm TFI}.
 \label{eq:Gm-block}
\end{equation}
The diagonal configuration matrix is block diagonal under the same permutation.  The formation-probability determinant therefore factorizes configuration by configuration,
\begin{equation}
 p_{A_\ell}^{(m)}(\bm\sigma)
 =\prod_{r=0}^{m-1}p_{A_{\ell_r}}^{\rm TFI}(\bm\sigma^{(r)}).
 \label{eq:p-decimation}
\end{equation}
Additivity of the Shannon--R\'enyi entropy for product distributions applied to Eq.~\eqref{eq:p-decimation} proves the factorization of the entropy for every $\alpha>0$.  Writing $\ell=mq+s$, $0\le s<m$, gives the more explicit identity
\begin{equation}
 H_\alpha^{(m)}(A_\ell)
 =s H_\alpha^{\rm TFI}(A_{q+1})+(m-s)H_\alpha^{\rm TFI}(A_q).
 \label{eq:decimation-general}
\end{equation}
At the Shannon point, let
\begin{equation}
 H_1^{\rm TFI}(A_q)=s_1q+\gamma_1\log q+\kappa_1+o(1).
\end{equation}
The extensive pieces in Eq.~\eqref{eq:decimation-general} combine exactly to $s_1\ell$.  For the logarithmic piece, using $q=(\ell-s)/m$ gives, at fixed residue $s$,
\begin{align}
 s\log(q+1)+(m-s)\log q
 ={}&m\log\frac{\ell}{m}
 -\frac{ms(m-s)}{2\ell^2}\nonumber\\
 &+\frac{ms(m-s)(m-2s)}{3\ell^3}
 +O(\ell^{-4}).
 \label{eq:decimation-log-expand}
\end{align}
In particular there is no residue-class-dependent $1/\ell$ term from the universal logarithm.  Equation~\eqref{eq:decimation-general} therefore yields, for every fixed residue class $s=0,\ldots,m-1$,
\begin{equation}
 H_1^{(m)}(A_\ell)
 =s_1\ell+m\gamma_1\log\ell
 +m(\kappa_1-\gamma_1\log m)+o(1),
 \label{eq:decimation-asympt}
\end{equation}
so the universal logarithmic coefficient is multiplied by $m$ for all interval lengths, not only along the subsequence $\ell=mq$.  The first residue dependence generated by the logarithmic term itself occurs at order $\ell^{-2}$; model-dependent corrections contained in the $o(1)$ remainder are not constrained by this observation.  Since this family has $c=m/2$, its Shannon logarithm is $2c\,\gamma_1$, not the naive $c/8$ branch.

\end{document}